\documentclass[12pt]{article} 
\usepackage[T1]{fontenc}
\usepackage[margin=0.75in]{geometry}
\title{Structured Cooperative Multi-Agent Reinforcement Learning: a Bayesian Network Perspective}
\author{Shahbaz P Qadri Syed and He Bai}
\usepackage{amsfonts,amssymb,amsmath, amsthm, makecell}
\usepackage{graphicx}
\usepackage{subfigure}
\usepackage{booktabs}
\usepackage{hyperref}
\usepackage{mathtools, cancel}
\usepackage{wrapfig}
\newtheorem{thm}{Theorem}[section]

\newtheorem{defin}[thm]{Definition}
\newtheorem{assume}[thm]{Assumption}
\newtheorem{remark}[thm]{Remark}
\newtheorem{property}[thm]{Property}

\usepackage{tikz}
\usetikzlibrary{arrows.meta}
\usetikzlibrary{bayesnet}
\usepackage{enumerate, cases,xcolor,algorithm,gensymb,algorithmic}
\ExplSyntaxOn
\NewDocumentCommand{\longdash}{ O{2} }
  { \prg_replicate:nn { #1 - 1 } { \negthinspace -- } }
\ExplSyntaxOff

\newcommand{\R}{\mathbb{R}}
\newcommand{\Z}{\mathbb{Z}}
\newcommand{\I}{\mathcal{I}}

\newcommand{\mS}{\mathcal{S}}
\newcommand{\mZ}{\mathcal{Z}}
\newcommand{\A}{\mathcal{A}}

\newcommand{\gq}{\mathbf{g}^i_Q}
\newcommand{\gc}{\mathbf{g}^i_C}
\newcommand{\gpi}{\mathbf{g}^\pi_i}
\newcommand{\var}{\textbf{tVar}}
\newcommand{\Var}{\textbf{Var}}
\newcommand{\tr}[1]{\text{Tr}\left(#1\right)}
\newcommand{\dq}{\delta_Q}
\newcommand{\dqh}{\delta_{\widehat{Q}}}

\newcommand{\tbf}[1]

\begin{document}

\maketitle
\begin{abstract}
The empirical success of multi-agent reinforcement learning (MARL) has motivated the search for more efficient and scalable algorithms for large scale multi-agent systems. 
However,  existing state-of-the-art algorithms do not fully exploit inter-agent coupling information to develop MARL algorithms. In this paper, we propose a systematic approach to leverage structures in the inter-agent couplings for efficient model-free reinforcement learning. We model the cooperative MARL problem via a Bayesian network and characterize the subset of agents, termed as the \textit{value dependency set}, whose information is required by each agent to estimate its local action value function exactly. 
Moreover, we propose a \textit{partially decentralized training decentralized execution} (P-DTDE) paradigm based on the value dependency set. 
We theoretically establish that the total variance of our P-DTDE policy gradient estimator is
less than the centralized training decentralized execution (CTDE) policy gradient estimator. We derive a multi-agent policy gradient theorem based on the P-DTDE scheme and develop a scalable actor-critic algorithm. We demonstrate the efficiency and scalability of the proposed algorithm on multi-warehouse resource allocation and multi-zone temperature control examples. For dense value dependency sets, we propose an approximation scheme based on truncation of the Bayesian network and empirically show that it achieves a faster convergence than the exact value dependence set {for applications with a large number of agents}.  
\end{abstract}

\section{Introduction}
The recent success of reinforcement learning (RL) has demonstrated strong evidence in achieving and exceeding human-level performance in complex tasks. 
RL has been studied extensively in a wide range of application domains such as games, robotics, logistics, transportation systems, and power systems. When multiple learning agents interact with the environment, the RL problem becomes much more challenging  due to non-unique learning goals, non-stationarity, and the curse of dimensionality~(\cite{zhang2021multi}).  
The main hindrance to addressing the scalability issue in multi-agent reinforcement learning (MARL) is the \textit{combinatorial nature} of MARL~(\cite{hleal2019,zhang2021multi,marl-book}). Owing to the non-stationarity of the MARL problem, each agent needs to account for the joint-action value space which increases exponentially in the number of agents. 
In addition, for applications involving a large number of agents with limited communication such as networked systems, it is practically infeasible for an agent to have access to the global state, i.e., the states of all the agents. Instead, each agent has access to its local observations, which are a subset of the global state. This is referred to as \textit{partial observability}\footnote{An alternative interpretation of partial observability has been used in the POMDP literature (see e.g.,~\cite{spaan2012partially}), where each agent observes the state through noisy sensors. However, we do not use this interpretation in this paper.} in MARL, which hinders the agents from learning the underlying policies of other agents, exacerbating the non-stationarity viewed by individual agents. 

Scalability issues and partial observability have motivated the quest for finding more computationally efficient and scalable algorithms. Several information structures have been proposed to address these issues. The extreme case of partial observability is known as the \textit{independent learning} (IL) scheme where each agent has access to only its own local action, observation and reward. This scheme suffers from non-convergence in general~(\cite{tan1993}). To address the non-stationarity caused by the partial observability, it is common in the MARL literature to assume that the agents have access to a centralized entity that provides joint actions, joint observations and joint rewards during training and each agent uses its local observations during execution. This is referred to as the \textit{centralized training and decentralized execution} (CTDE) scheme, which has been shown to be effective in a wide range of applications~(see e.g., \cite{lowe2017,foerster2018counterfactual,rashid2018qmix}). In the cooperative setting, the CTDE paradigm greatly simplifies the analysis, facilitating the use of tools developed in the analysis of single-agent setting~(\cite{zhang2021multi}). However, in many real-world scenarios the existence of a central entity either is  infeasible or incurs a significant communication overhead due to a large amount of information exchange, degrading the scalability of the approach. 

 In this work, we focus on the scalability issue in the cooperative MARL setting, where the objective of the MARL problem is to find an optimal policy for each agent that maximizes the total accumulated rewards of all the agents. In many real-world applications, {e.g., vehicle platooning, warehouse resource allocation, and multi-zone temperature control, structural information of the couplings between cooperative agents is known even in the model-free setting, either by design (e.g., placement of sensors) or by expert knowledge (e.g., the physical interactions/constraints among agents~\cite{yang2024physics}). 
How to capture such information in designing efficient MARL algorithms is the main subject of this paper.

\textbf{Contributions.}{This paper introduces a Multi-agent Bayesian Network (MABN) framework that enables exact Q-function decomposition for cooperative multi-agent reinforcement learning. Our key contributions are as follows:
\begin{itemize}
    \item We develop a systematic approach to leverage inter-agent couplings for \textit{exact} Q-function decomposition without approximation errors. Unlike existing methods that approximate the value function through separable MDPs~\cite{jin2024approximate}, approximate factorization~\cite{lu2024overcoming}, spatial exponential decay~\cite{qu2019value,Qu2022}, or monotonicity constraints~\cite{sunehag2017value,rashid2018qmix}, our MABN framework identifies the minimal set of agents whose information is needed to compute the Q-function of each agent, thereby reducing the curse of dimensionality while maintaining optimality.
    \item  This paper extends our recent work~\cite{syed2025exploiting} from the linear quadratic regulator (LQR) setting to a generic setting, where agents' dynamics, rewards and observations can be nonlinear. Departing from the set-theoretic properties employed in~\cite{syed2025exploiting}, we develop the MABN, a unified probabilistic graphical model that explicitly encodes the individual inter-agent couplings in a single graph,  allowing for efficient deduction of value dependency structures. The Bayesian network approach is more general and applicable to nonlinear systems with structural constraints. 
    \item The MABN framework applies to general partially observable stochastic cooperative games (POSCG) without restrictive assumptions on network topology, coupling density. It explicitly handles partial observability and remains effective even in densely coupled systems where convergence guarantees for approximation methods may not hold.
    \item While the generality of MABN prevents from obtaining convergence guarantees or sample complexity bounds as in~\cite{syed2025exploiting}, we establish that our exact decomposition yields lower total variance compared to centralized training decentralized execution (CTDE) scheme suggesting improved sample complexity for gradient-based algorithms. We develop a cooperative policy gradient theorem for both stochastic and deterministic policies, which can be used to derive various RL algorithms.
    \item We design a partially decentralized training and decentralized execution architecture (P-DTDE) based on the MABN structure for scalable learning. Further, we develop a \textit{multi-agent structured actor critic (MAStAC) algorithm} based on the P-DTDE scheme which demonstrates superior performance across benchmark tasks.
\end{itemize}

\noindent
\textbf{Outline.} The paper is outlined as follows. Section~\ref{sec:coopmarl} introduces the problem formulation. Section~\ref{sec:main} presents the main technical content, including the MABN approach for decomposition, our multi-agent policy gradient theorem, and the proposed MARL algorithm. Section~\ref{sec:variance} analyzes the variance in the gradient estimation and compares it with the CTDE scheme. Experimental results are presented in Section~\ref{sec:exp}, followed by conclusions and limitations in Section~\ref{sec:conclusion}.

\section{Problem formulation}\label{sec:coopmarl}
We consider a {partially-observable stochastic cooperative game} described by the tuple $ \mathcal{M} = (\mathcal{V}, \mathcal{E}_S, \mathcal{E}_O, \mathcal{E}_R,  $ $ \prod_{i \in \mathcal{V}}\mathcal{M}_i,\prod_{i \in \mathcal{V}}\mathcal{O}_i,\prod_{i \in \mathcal{V}}r_i, \gamma)$, where 
 
 - $\mathcal{V} = \{1,\cdots,N\}$ is the set of agent indices in the various graphs describing inter-agent couplings.
 
- $\mathcal{E}_S \subseteq \mathcal{V} \times \mathcal{V}$ is the edgeset of the \textit{state graph} $\mathcal{G}_S = \{\mathcal{V},\mathcal{E}_S\}$ which defines the dynamics coupling between agents. An edge $(j,i) \in \mathcal{E}_S$, if agent $j$ affects the state evolution of agent $i$. Define the state index set of agent $i$ as $\mathcal{I}^i_S = \{j|(j,i) \in \mathcal{E}_S\} \cup \{i\}$. 

- $\mathcal{E}_O \subseteq \mathcal{V} \times \mathcal{V}$ is the edgeset of the \textit{observation graph} $\mathcal{G}_O = \{\mathcal{V},\mathcal{E}_O\}$ which defines the partial observability of agents. An edge $(j,i) \in \mathcal{E}_O$, if agent $j$ is observed by agent $i$. Define the observation index set of agent $i$ as $\mathcal{I}^i_O = \{j|(j,i) \in \mathcal{E}_O\} \cup \{i\}$. {\textit{Partial observability} in this paper means that each agent observes the states of its in-neighbors in $\mathcal{E}_O$, \textit{instead of all the agents}}.

- $\mathcal{E}_R \subseteq \mathcal{V} \times \mathcal{V}$ is the edgeset of the \textit{reward graph} $\mathcal{G}_R = \{\mathcal{V},\mathcal{E}_R\}$ which defines the reward coupling among agents. An edge $(j,i) \in \mathcal{E}_R$, if the state and action of agent $j$ affect the reward of agent $i$. Define the reward index set of agent $i$ as $\mathcal{I}^i_R = \{j|(j,i) \in \mathcal{E}_R\} \cup \{i\}$.

- $\mathcal{M}_i = (\mathcal{S}_i, \mathcal{A}_i, \mathcal{P}_i)$ is the model that governs the state evolution of agent $i$ in time.

- $\mathcal{S}_i, \mathcal{A}_i$ are the state space and action space (discrete or continuous) of agent $i$ respectively. 
  
- $\mathcal{P}_i:\prod_{j \in \mathcal{I}^i_S} \mathcal{S}_j \times \prod_{j \in \mathcal{I}^i_S} \mathcal{A}_j \rightarrow \mathcal{P}(\mathcal{S}_i)$
is the transition probability function of agent $i$ in the next time step, i.e., it specifies the probability distribution of agent $i$ transitioning to $s'_i \in \mathcal{S}_i$ in the next time step when each agent $j \in \mathcal{I}^i_S$ takes action $a_j \in \mathcal{A}_j$ in the current state $s_j \in \mathcal{S}_j$ in the current time step.

- $\mathcal{O}_i = \prod_{j \in \mathcal{I}^i_O} \mathcal{S}_j$ is defined as the partial observation of agent $i$ which includes the state information of all the agents $j$ that are observable by $i$.

- $r_j: \prod_{j \in \mathcal{I}^i_R} \mathcal{S}_j \times \prod_{j \in \mathcal{I}^i_R} \mathcal{A}_j \rightarrow \mathbb{R}$ is the instantaneous reward received by agent $i$ when each agent $j \in \mathcal{I}^i_R$ takes an action $a_j \in \mathcal{A}_j$ in the current state $s_j \in \mathcal{S}_j$. 

- $\gamma \in [0,1]$ is the discount factor that trades off the immediate and future rewards received by the agents.

To illustrate the three graphs $\mathcal{G}_S$, $\mathcal{G}_O$, and $\mathcal{G}_R$ in a real-world setting, consider a multi-zone temperature control problem (with details in Appendix I). The temperature in each zone is affected by the heat conduction through walls from adjacent zones resulting in inherent physical couplings in the state transition dynamics. Thus, $\mathcal{G}_S$ is an undirected graph with edges between zones that share a wall. If each zone can only measure its own temperature, the corresponding observation graph $\mathcal{G}_O$ is fully decentralized. If the objective of the controller in each zone is to maintain the temperature in that zone, the resulting reward graph $\mathcal{G}_R$ is also fully decentralized.

Let $s(t) = \begin{bmatrix}
    s_1(t)&s_2(t)&\cdots&s_N(t)
\end{bmatrix}$ and $a(t) = \begin{bmatrix}
    a_1(t)&a_2(t)&\cdots&a_N(t)
\end{bmatrix}$ denote the global state and global action of the multi-agent system (MAS) at time $t$, respectively. Let $\pi:\mathcal{S} \rightarrow \mathcal{P}(\mathcal{A})$ and $\pi_i:\mathcal{O}_i \rightarrow \mathcal{P}(\mathcal{A}_i)$ be the global policy function of the entire MAS and the local policy function of agent $i$, respectively. To simplify exposition, we have chosen the policies $\pi_i$'s to be reactive. The formulation in this paper can be extended to allow history-based policies.
We define the observation function for each agent $i \in V$, as a mapping $f_o: \bigcup_{j\in \mathcal{I}^i_O} \mS_j \rightarrow \mathcal{O}_i$ such that $o_i(t) = f_o(s_{\mathcal{I}^i_O}(t)),~\forall~i \in \mathcal{V}$. 

At time $t$ of the game, agent $i \in \mathcal{V}$ observes the partial state $o_i(t)$, takes an action $a_i(t) \in \mathcal{A}_i$ following its local policy function $\pi_i(a_i(t)|o_i(t))$ and receives a local reward $r_i(s_{\mathcal{I}^i_R}(t), a_{\mathcal{I}^i_R}(t))$. The accumulated discounted global reward is defined as
\begin{align}
    R = \sum_{t=0}^\infty \gamma^t r(s(t),a(t)),
   \label{eq:accgr}
\end{align}
where {$r(s(t),a(t)) = \sum_{i=1}^N r_i(s_{\mathcal{I}^i_R}(t), a_{\mathcal{I}^i_R}(t))$ }is the global reward of the MAS at time $t$.
The global state value function and the global state-action value function  for a given policy $\pi$ are defined as {$V^\pi(s)=\mathbb{E}\left[R | s(0) = s\right]$, and $Q^\pi(s,a) =$ $ \mathbb{E}\left[R| s(0) = s, a(0) = a\right]$}, respectively, where the expectation is taken over the state and state-action distributions, respectively. Define $V^\pi_i(s) = \mathbb{E}\left[\sum_{t=0}^\infty \gamma^t r_i(s_{\mathcal{I}^i_R}(t), a_{\mathcal{I}^i_R}(t)) | s(0) = s\right]$ as agent $i$'s local state value function and $Q^\pi_i(s,a) = $ $\mathbb{E}\left[\sum_{t=0}^\infty \gamma^t r_i(s_{\mathcal{I}^i_R}(t), a_{\mathcal{I}^i_R}(t)) | s(0) = s, a(0) = a\right]$ as its local state-action value function. 

The objective of the cooperative MARL problem is to find a globally optimal policy $\pi^*$ that maximizes the expected long term global reward $J(\pi)$ when the MAS starts from a global state $s$, takes a global action $a$ and follows a global policy $\pi$, i.e.,
\begin{align}
    \pi^* &= \underset{\pi}{\text{argmax }} J(\pi)
    = \underset{\pi}{\text{argmax }}
    \mathbb{E}_{s,a \sim \mathcal{D}} \left[R|s(0) = s\right]
    = \underset{\pi}{\text{argmax }} \mathbb{E}_{s \sim \mathcal{D}} ~V^\pi(s),
    \label{eq:objpi}
\end{align}
where $\mathcal{D}$ is the distribution over the initial global state and global action. We parameterize the global  policy $\pi$ by $\theta = \begin{bmatrix}
    \theta_1^\intercal&\theta_2^\intercal&\cdots&\theta_N^\intercal
\end{bmatrix}^\intercal$ and denote it as $\pi_\theta$. Similarly, the local policy of agent $i$ is denoted by $\pi_{\theta_i}$. Given the global state $s \in \mathcal{S}$, there exists a bijective mapping between the global policy $\pi_\theta(a(t)|s(t))$ and the collection of local policies $\bigcup_{i \in \mathcal{V}} \pi_{\theta_i} (a_i(t)|o_i(t))$. Specifically, the global policy can be decomposed as $\pi_\theta(a(t)|s(t)) = \prod_{i \in \mathcal{V}} \pi_{\theta_i} (a_i(t)|o_i(t))$ as the agents act independently based on their local observation $o_i(t) = f_o(s(t))$. Conversely, the joint distribution over actions defined by the collection of local policies induces global policy when conditioned on the global state $s(t)$. The equivalence holds because in our formulation we assume that the global state contains all the information needed to determine each agent's observation $o_i(\cdot)$, and the joint action is uniquely determined by individual actions.. The optimization problem \eqref{eq:objpi} can be rewritten in terms of $\theta$ as
\begin{align}
    \theta^* 
    &= \underset{\theta}{\text{argmax }}  J(\theta) = \underset{\theta}{\text{argmax }} \mathbb{E}_{s \sim \mathcal{D}} ~V^{\pi_\theta}(s).
    \label{eq:globalobj}
\end{align}

We present a concrete example to illustrate the formulation. Consider a multi-agent system with time-invariant dynamics, where each agent's dynamics and cost may depend on other agents' states and actions and its observations are the states of a subset of all the agents. Specifically,  the dynamics of agent $i \in \mathcal{V}$ is given by 
\begin{align}
    s_i({t+1})&= f_i(\cup_{j\in \I^i_S} s_j(t), \cup_{j\in \I^i_S}  a_j(t), w_i(t))
    \label{eq:inddyn}
\end{align}
where $s_i(t)\in\mathbb{R}^{n_s}$, $a_i(t)\in\mathbb{R}^{n_a}$, $w_i(t)$ is the process noise, and the set $\mathcal{I}^i_S$ contains the indices of the agents impacting the dynamics of agent $i$.
The reward (or negative cost) incurred by agent $i \in \mathcal{V}$ at time $t$ is given by
\begin{align}
    r_i(t) &= g_i(\cup_{j \in \I^i_R}s_j(t),~\cup_{j \in \I^i_R}a_j(t)), \label{eq:indcost}
\end{align}
where the set $\mathcal{I}^i_R$ contains the indices of the agents impacting the cost of agent $i$.
Agent $i$ observes the states of the agents in a set $\I^i_O$ and prescribes a structured static feedback control \begin{align}
    a_i(t) &= \pi_{\theta_i}(\cup_{j\in \I^i_O}s_j(t)). 
    \label{eq:indcontrol}
\end{align}
In this paper, we focus on solving the stochastic optimal control problem in a model-free setting,  i.e.,
\begin{eqnarray}
\underset{\theta}{\text{min }}& \mathbb{E}_{s(0), a(0)} ~\left[Q(s(0), a(0))|\mathcal{G}_S,\mathcal{G}_O,\mathcal{G}_C\right]\nonumber\\
    \text{subject to }&\eqref{eq:inddyn},~\eqref{eq:indcost},~\eqref{eq:indcontrol}~\forall~i\in \mathcal{V}
\end{eqnarray}
where $f_i(\cdot),~g_i(\cdot)$ are unknown but we can query $s(t+1), r(t)$ given $s(t),~a(t)$ from a simulator. 

\section{Multi-agent Bayesian network and policy gradient theorem}\label{sec:main}
Bayesian networks (BN) are an example of probabilistic graphical models (PGM) whose nodes represent random variables and edges encode the conditional dependence associated with them. 
RL and optimal control can be formulated using PGM to facilitate the use of probabilistic inference approaches 
This framework is referred to as the RL-as-inference framework that interprets finding optimal policies in RL as a form of probabilistic inference in graphical models. In this view, the agent's objective to maximize rewards can be viewed as inferring the most probable actions or trajectories. Interested readers are referred to~\cite{levine2018, noorani2022probabilistic, O'Donoghue2020, tarbouriech2023}. The existing literature is largely focused on single agent decision making problems. We next propose a systematic approach to modeling the cooperative MARL problem in {{Section~\ref{sec:coopmarl}}} 
as a BN, which we hope will enable future extensions of \textit{RL-as-inference}} paradigm to the cooperative MARL setting.}
\subsection{Construction of multi-agent Bayesian network}\label{sec:MABN}
 We define a BN which we refer to as the \textit{multi-agent Bayesian network (MABN)} to represent the evolution of the MAS in a finite-horizon setting. The states and the actions of the agents are latent variables and a binary random variable $\mZ_i(t)$ for each agent is introduced as an observable variable. Assuming a negative reward structure, similar to the inference-based control approaches~(see e.g., \cite{gd,aico,constraico,aicot,convaico,policyup,i2c,effsoc, syed2025approximate}), the log-likelihood that agent $i$ is optimal at time $t$ is proportional to $r_i(s_{\I^i_R}(t),a_{\I^i_R}(t))$, i.e.,
\begin{align}
    p(\mZ_i(t) = 1|s_{\I^i_R}(t),a_{\I^i_R}(t)) \propto \exp\{ r_i(s_{\I^i_R}(t),a_{\I^i_R}(t))\}.
    \label{eq:expreln}
\end{align}
Given $\mathcal{M}$, the MABN is constructed as follows.
1) Assign a random variable to each $s_i(t)$, $a_i(t)$, and $\mZ_i(t)$ $\forall~t,$ $\forall~i\in \mathcal{V}$.

2) For each $i$, at every $t$
\begin{itemize}
    \item  Construct directed edges $(s_j(t), s_i(t\!+\!1))$, $(a_j(t), s_i(t+1))$, $\forall~j \in \mathcal{I}^i_{S}$, e.g., see Fig.~\ref{fig:BNs}.
    
\item Construct directed edges $(s_j(t), a_i(t))$, $\forall~j \in \mathcal{I}^i_{O}$, e.g., see Fig.~\ref{fig:BNo}.

\item Construct directed edges $(s_j(t), \mZ_i(t))$, $(a_j(t), \mZ_i(t))$, $\forall~j \in \mathcal{I}^i_{R}$, e.g., see Fig.~\ref{fig:BNr}\footnote{Fig.~\ref{fig:BNs}--\ref{fig:BNr} are in Appendix~\ref{sec:MABN_exp}.}.
\end{itemize}    

\subsection{Deduction of action value function dependence from MABN}
\label{sec:Qdecomp}
 The formulation in Section~\ref{sec:coopmarl} and the existing literature, e.g.,~\cite{jing2024distributed}, assume that the individual action value of agent $i$, $Q^\pi_i(s,a)$, depends on the global state $s$ and global action $a$. However, during critic training in MARL, having computations dependent on the global state is undesirable as it incurs high communication cost and significantly hinders the scalability of the algorithms. Thus, a natural question arises: 
\textit{What is the minimal set  of agents whose information is sufficient for an agent to estimate its action value function?} 
We employ the MABN in Section~\ref{sec:MABN} and derive a systematic approach to obtain such a minimal set. 
Consider the formulation in Section~\ref{sec:coopmarl} with a finite horizon $T$. We express  $Q^\pi_i(s(t),a(t))$, $\forall i\in\mathcal{V}$, as
\begin{align}
&Q^\pi_i(s(t),a(t)) = 
r_i(s_{\mathcal{I}^i_R}(t), a_{\mathcal{I}^i_R}(t)) +\underset{\substack{s(t\!+\!1)\sim \mathcal{P}\\ a(t\!+\!1) \sim \pi_\theta}}{\mathbb{E}}\bigg[\gamma r_i(s_{\mathcal{I}^i_R}(t\!+\!1), a_{\mathcal{I}^i_R}(t\!+\!1)) +\cdots \nonumber\\
&\cdots +\underset{\substack{s(T\!-\!1)\sim \mathcal{P}\\ a(T\!-\!1) \sim \pi_\theta}}{\mathbb{E}}\bigg[ \gamma^{T-t-1}r_i(s_{\mathcal{I}^i_R}(T\!-\!1), a_{\mathcal{I}^i_R}(T\!-\!1))+ \underset{\substack{s(T)\sim \mathcal{P}\\ a(T) \sim \pi_\theta}}{\mathbb{E}}\big[ \gamma^{T-t} r_i(s_{\mathcal{I}^i_R}(T), a_{\mathcal{I}^i_R}(T))\big] \bigg]\bigg].
\label{eq:QZeq}
\end{align}
 To make the formulation more general and better reveal the dependency of $Q_i^\pi(s(t),a(t))$ on the inter-agent couplings, we assume that the sets $\I_O^i$, $\I_R^i$, and $\I_S^i$ can be time-varying.
Given the time-varying inter-agent couplings $\I^i_O(t)$, $\I^i_S(t)$, $\I^i_R(t)$, for all $0\leq t \leq T$, we define $U^{\tau}_i = \{\I^j_O(\tau)\}_{j\in \I^i_R(\tau)} \cup \{\{\I^k_O(\tau)\}_{k \in\I^j_S(\tau\!+\!1)}\}_{j\in U^{\tau\!+\!1}_i}$, and $U^T_i = \{\I^j_O(T)\}_{j\in \I^i_R(T)}$ as the time-dependent reachability sets of agent $i$'s reward in the MABN at intermediate time $\tau$ (corresponding to $\mathcal{Z}_i(\tau)$) and terminal time $T$ (corresponding to $\mathcal{Z}_i(T)$), respectively.
We further let \begin{equation}
    \I^i_Q(\tau) = U^\tau_i\cup U^{\tau\!+\!1}_i\cup \cdots \cup U^T_i \subseteq \mathcal{V}.\label{eq:Qset}
\end{equation}  Theorem~\ref{thm:Qsetdecomp} below establishes that $Q_i^\pi(s(t),a(t))$  depends only on the state-action of agents in $\I^i_Q(t)$. Thus, we refer to $\I^i_Q(t)$ as the \textit{value dependency set} at time $t$ and the corresponding graph such that the in-neighbors of $i$ correspond to $\I^i_Q(t)$ as the \textit{value dependency graph}, denoted by $\mathcal{G}_{\text{VD}}(t) = \{\mathcal{V}, \mathcal{E}_{\text{VD}}(t)\}$.
\begin{remark}
    In general, a direct relationship between $U^t_i$ and $U^{t+1}_i$ cannot be established without the knowledge of $\I^i_S(t)$, $\I^i_O(t+1)$, and $\I^i_R(t+1)$. However, under the formulation presented in this paper, the construction of MABN ensures that $U^t_i \subseteq U^{t+1}_i$ whenever $\I^i_R(t) \subseteq \{\I^j_S(t)\}_{j\in\I^i_R(t+1)}$.
    Futhermore, by the definition of $\I^i_Q(\tau)$ in~\eqref{eq:Qset}, we note that for all $i \in \mathcal{V}$, $\I^i_Q(\tau)$ accounts for one additional step of dependencies compared to $\I^i_Q(\tau+1)$. Hence, it follows that $\mathcal{G}_{\text{VD}}(\tau)$ would generally  be denser than $\mathcal{G}_{\text{VD}}(\tau+1)$. However, a precise relationship cannot be established without knowledge of the temporal evolution of $\mathcal{G}_S,~\mathcal{G}_O,$ and $\mathcal{G}_R$.
\end{remark}
\begin{thm}
    For any $i\in \mathcal{V}$, $Q^\pi_i(s(t),a(t))$ depends only on $(s_{\I^i_Q(t)}(t), a_{\I^i_Q(t)}(t))$.
    \label{thm:Qsetdecomp}
\end{thm}

The proof of Theorem~\ref{thm:Qsetdecomp} is deferred to Appendix~\ref{sec:Qsetdecomp}. With a slight abuse of notation, we have $Q^\pi_i(s(t),a(t)) = Q^\pi_i(s_{\I^i_Q(t)}(t), a_{\I^i_Q(t)}(t))$.

The MABN constructed in {Section~\ref{sec:MABN}} captures temporal evolution of the inter-agent couplings and provides an elegant way to compute $U^\tau_i$ in~\eqref{eq:Qset}, $\forall~\tau,~\forall~i$. Owing to the causal relationship between variables in the MABN and the relationship between the local rewards and binary random variables in~\eqref{eq:expreln}, $U^\tau_i$ can be obtained by tracing the predecessors of $\mZ_i(\tau'),~\forall~\tau\leq \tau'\leq T$ until time $\tau$. 
We formalize this result as a property of the value dependency graph.  
\begin{property}[\textbf{\textit{Value dependency graph}}]
    At any time $t$, an edge $(j,i) \in \mathcal{E}_{\text{VD}}(t)\subseteq \mathcal{V}\times \mathcal{V}$ if and only if $\exists$ $\tau,\tau'$ satisfying $t \leq \tau \leq \tau' \leq T$ such that $s_j(\tau)$ or $a_j(\tau)$ can reach $\mZ_i(\tau')$ in the MABN.
\label{defin:valdep}
\end{property}

\begin{remark}
    Note that for the formulation presented in this paper, if agent $j$ affects the value function of agent $i$, then the state $s_j(\tau)$ can definitely reach $\mathcal{Z}_i(\tau')$. In contrast, $a_j(\tau)$ can reach $\mathcal{Z}_i(\tau')$ if and only if 
    agent $j$ influences agent $i$ indirectly through its state or reward neighbors. For example, observe that in Fig.~\ref{fig:BNfull}, $s_1(0)$ can reach $\mathcal{Z}_3(1)$ which implies $1\in \I^3_{\text{VD}}$. But $a_1(0)$ cannot reach $\mathcal{Z}_3(1)$ because agent $1$ cannot affect agent $3$ through its state, or reward neighbors i.e., $1\not\in \{\I^j_S \cup \I^j_R\}_{j\in \I^3_{\text{VD}}\setminus \{1\}}$. However, the MABN generalizes to broader scenarios where only the actions of other agents are observable while the states remain hidden, or where the reward depends exclusively on the states of neighboring agents and  are independent of their actions.
\end{remark}

From Property~\ref{defin:valdep}, a brute-force solution to computing $\I^i_Q(t)$ is to find the predecessors of $\mZ_i(\tau)$, $\tau = t,t+1,\cdots,T$ in the MABN using a path finding algorithm. However, the size of the PGM grows linearly in the horizon and exponentially in the number of agents. Therefore, performing path finding on the full MABN is computationally intensive. For \textit{time-invariant} inter-agent couplings, the structure of the MABN repeats every two consecutive time steps (see e.g., Fig.~\ref{fig:BNfull} in Appendix~\ref{sec:MABN_exp}). We leverage this redundancy in the structure to make the path finding efficient for time-invariant inter-agent couplings. Specifically, we consolidate the MABN to two time steps by making edges between $s_i(t)$ and $s_i(t+1)$, $\forall~i$,  (see e.g., Fig.~\ref{fig:BNfolded} in Appendix~\ref{sec:MABN_exp}). We refer to this consolidated two time-step PGM as the \textit{folded MABN} $\mathcal{G}_F$. The intuition for the bidirectional edge is that every traversal of the bidirectional edge in the folded graph implies evolution of a timestep in the full MABN. 
Because $\I_Q^i$ is time-invariant, Theorem~\ref{thm:Qsetdecomp} indicates that
$Q^\pi_i$ depends only on $(s_{\I^i_Q}, a_{\I^i_Q})$, i.e., $Q^\pi_i(s(t),a(t)) = Q^\pi_i(s_{\I^i_Q}(t), a_{\I^i_Q}(t))$.

\begin{remark}
    We remark that the MABN is capable of modeling the evolution of MAS in both finite and infinite horizon settings provided that the inter-agent couplings are known. However, since the identification of VD set employs path finding on the MABN (Property~\ref{defin:valdep}), a finite (finite-horizon setting) or receding (infinite-horizon setting) horizon is required for computational tractability.
\end{remark}

\subsection{Value gradient dependence and the multi-agent policy gradient theorem}
\label{sec:mapg}
The use of policy gradient theorems in our formulation requires computing the gradient of the global action value function $Q^\pi(s,a)$ with respect to individual action parameters $\theta_i$, $\forall~i\in \mathcal{V}.$ From Property~\ref{defin:valdep}, it follows that the action of agent $i$ affects the rewards of its out-neighbors in $\mathcal{G}_{\text{VD}}.$ Thus, we define the \textit{gradient dependency} graph $\mathcal{G}_{\text{GD}} = \{\mathcal{V}, \mathcal{E}^\intercal_{\text{VD}}\}$ and the corresponding index set $\I^i_{\text{GD}} = \{j | (j,i)\in \mathcal{E}^\intercal_{\text{VD}}\}\cup i$, $\forall~i\in \mathcal{V}$. We show in Theorem~\ref{thm:graddecomp} that \textit{the gradient of $Q^\pi(s,a)$ with respect to $\theta_i$ can be decomposed as the sum of the gradients of $Q^\pi_j(s_{\I^j_Q},a_{\I^j_Q})$, $\forall j\in \I^i_{\text{GD}
}$}. Theorem~\ref{thm:graddecomp} illustrates the interplay between the inter-agent couplings and the decomposition of the gradient of the objective function. The gradient dependency graph corresponds to the learning graph proposed in~\cite{jing2024distributed}. However,~\cite{jing2024distributed} neither considers the value dependency graph nor the MABN which are the two main contributions of this paper. We also establish that $\mathcal{G}_{\text{GD}}$ is the transpose graph of $\mathcal{G}_{\text{VD}}$. 
 
\begin{thm}
    [Gradient decomposition theorem]
If $Q^\pi(s,a)$ and $Q^\pi_i(s_{\I^i_Q},a_{\I^i_Q})$ are continuous and differentiable with respect to $\theta_i$, $\forall$ $i \in \mathcal{V}$, then  
        $\nabla_{\theta_i} Q^\pi(s,a) =  \nabla_{\theta_i} \left(\sum_{j \in \mathcal{I}_{\text{GD}}^i} Q^\pi_j(s_{\I^j_Q},a_{\I^j_Q})\right).$
    \label{thm:graddecomp}
\end{thm}
\begin{proof}
Refer to Appendix~\ref{sec:graddecomp}.
\end{proof}

We further let $\I^i_{\widehat{Q}} = \bigcup_{j\in \I^i_{\text{GD}}} \I^j_Q$ and define
\begin{align}
\widehat{Q}^\pi_i(s_{\I^j_{\widehat{Q}}},a_{\I^j_{\widehat{Q}}})= \sum_{j \in \mathcal{I}_{\text{GD}}^i} Q^\pi_j(s_{\I^j_Q},a_{\I^j_Q}).
\label{eq:Qhat}
\end{align}
 Based on Theorem~\ref{thm:Qsetdecomp} and~\ref{thm:graddecomp}, Theorem~\ref{thm:pgt} below presents the multi-agent policy gradient theorem for the formulation in Section~\ref{sec:coopmarl}.
\begin{thm}
Let $\pi_{\theta_i}(\cdot)$ be the  policy for agent $i$ parameterized by $\theta_i$ and $\hat{\pi}_i = \prod_{j \in \I^i_{\widehat{Q}}} \pi_{\theta_j}(a_j|s_{\I^j_O})$,
 where $\widehat{Q}_i(\cdot)$ is defined in~\eqref{eq:Qhat}. Let $J(\theta)$ be the expected global return.
\begin{enumerate}[(a)]
\item if ~$\pi_{\theta_i}$ is deterministic, then $\nabla_{\theta_i}J(\theta) =\underset{s_{\I^i_{\widehat{Q}}}\sim d^\pi(s_{\I^i_{\widehat{Q}}})}{\mathbb{E}}\left[ \nabla_{\theta_i} \pi_{\theta_i}(s_{\mathcal{I}^i_O}) \nabla_{a_i} \widehat{Q}_i(s_{\I^i_{\widehat{Q}}},a_{\I^i_{\widehat{Q}}})\big|_{a_i = \pi_{\theta_i}(\cdot)}\right]$.
    \item if ~$\pi_{\theta_i}$ is stochastic, then $\nabla_{\theta_i}J(\theta) =\underset{s_{\I^i_{\widehat{Q}}}\sim d^\pi(s_{\I^i_{\widehat{Q}}}),a_{\I^i_{\widehat{Q}}} \sim \hat{\pi}^i}{\mathbb{E}} \left[ \widehat{Q}_{i}(s_{\I^i_{\widehat{Q}}},a_{\I^i_{\widehat{Q}}} )\nabla_{\theta_i}\ln\pi_{\theta_i}(a_i|s_{\I^i_O}) \right]$.
\end{enumerate}
\label{thm:pgt}
\end{thm}
The proof of Theorem~\ref{thm:pgt} is deferred to Appendix~\ref{sec:pgt}. This theorem establishes that the gradient of the global objective function in~\eqref{eq:globalobj} with respect to  $\theta_i$, $\forall~i\in \mathcal{V}$, can be obtained locally using the corresponding score function $\nabla_{\theta_i} \pi_{\theta_i}(\cdot)$ and the decomposed action-value function $\widehat{Q}^\pi_i(\cdot)$. In contrast, the existing multi-agent policy gradient theorem proposed in~\cite{zhang2018fully} requires estimation of the global value or advantage functions that depend on the global state, action, and reward. Therefore, the existing MARL algorithms usually require some form of centralized architecture or consensus during training to diffuse the local information of an agent over the network. Theorem~\ref{thm:pgt} facilitates a sparse architecture for information exchange between a subset of agents, which we refer to as \textit{partially decentralized training and decentralized execution (P-DTDE)} paradigm. Moreover, we theoretically establish variance reduction of the P-DTDE policy gradient estimator compared to the CTDE policy gradient estimators in Section~\ref{sec:variance} and empirically demonstrate the scalability and faster convergence of the P-DTDE paradigm in Section~\ref{sec:exp}. The formulation presented in Section~\ref{sec:coopmarl} is classified as \textit{Reactive Policy under Full Observability}~\cite{baisero2021unbiased} i.e., the policy is reactive and the union of observations recovers the full state. Therefore, the state-based critics in our P-DTDE paradigm are unbiased~\cite[Theorem 4.3]{baisero2021unbiased}.

\section{Variance analysis}\label{sec:variance}
In this section, we analyze the variance of the policy gradient estimator based on the proposed P-DTDE scheme.  
We obtain a general expression of the difference between the total variances of the CTDE and the P-DTDE stochastic policy gradient estimators, assuming that the gradient estimators produce biased estimates of the action-value function with non-zero variances. When both estimators achieve unbiased estimation, we characterize upper and lower bounds for the difference between the total variances and establish that the total variance of the P-DTDE policy gradient estimator is less than the CTDE policy gradient estimators.

\begin{assume}
    The individual reward of each agent at any time $t$ is uniformly bounded i.e., $\forall~i\in \mathcal{V},~t\in \Z_{+}$, $\exists$ an $r_u\in \R$ such that $|r_i(\cdot)| \leq r_u.$
    \label{assume:boundedr}
\end{assume} 

We consider two policy gradient estimators: one based on the CTDE paradigm ($\gc$), where each agent estimates the global action value function $Q^\pi(s,a)$ based on the global state and action $(s,a)$, and another based on the proposed P-DTDE paradigm ($\gq$), where agent $i$, $\forall i$, estimates $\widehat{Q}^\pi_i(\hat{s}_i,\hat{a}_i )$  given in~\eqref{eq:Qhat}, where we have let $\hat{s}_i = s_{\I^i_{\widehat{Q}}}$ and $\hat{a}_i = a_{\I^i_{\widehat{Q}}}$. We further consider that the estimation of $Q^\pi(s,a)$ and $\widehat{Q}^\pi_i(\hat{s},\hat{a})$ is not exact. Let $\dq$ and $\dqh$ be the errors in the estimation of $Q^\pi(s,a)$ and $\widehat{Q}^\pi_i(\hat{s}_i, \hat{a}_i)$, respectively. The mean and variance of $\dq$ and $\dqh$ are denoted as $(\mu_Q,\sigma^2_Q)$ and $(\mu_{\widehat{Q}},\sigma^2_{\widehat{Q}})$, respectively. For stochastic policy gradient, the CTDE and P-DTDE estimators are given by
$$\mathbf{g}^i_C = (Q^\pi(s,a)-\dq)\nabla_{\theta_i} \ln{\pi_{\theta_i}(a_i|s_{\I^i_O})}~\mbox{and}~
\mathbf{g}^i_Q = (\widehat{Q}^\pi_i(\hat{s}_i,\hat{a}_i) - \dqh)\nabla_{\theta_i} \ln{\pi_{\theta_i}(a_i|s_{\I^i_O})}.$$

We obtain a general expression for the difference between the total variances of $\gc$ and $\gq$ in Appendix~\ref{sec:vardiff}. Our result shows that the difference in the total variances depends  the variance of the policy gradient as well as the accuracy of the estimation in terms of $(\mu_Q,\sigma^2_Q)$ and $(\mu_{\widehat{Q}},\sigma^2_{\widehat{Q}})$. Assuming that the estimation is unbiased, i.e., $\mu_Q=\mu_{\widehat{Q}}=0$, Theorem~\ref{thm:vardiff} below establishes the upper and lower bounds on the difference between the total variances of the two estimators. 

To state the theorem, we define the action value function after a subset of agents, say $I\subset \mathcal{V}$, take their actions as 
    $Q^\pi(s, a_{I}) = \mathbb{E}_{a_{-{I}}\sim \pi_{\theta_{-I}}}\left[{Q}^\pi(s, a)\right]$, and the corresponding advantage function of the subset of agents $I$ as ${A_I(s, a_{-I},a_{{I}})} = {Q}^\pi(s, a) - {Q}^\pi(s, a_{{-I}})$, 
where $-I = \mathcal{V}\setminus I.$

\begin{thm}
Define the score function $\gpi = \nabla_{\theta_i} \ln{\pi_{\theta_i}(a_i|s_{\I^i_O})}$. Then
    the total variances (\textbf{\textit{tVar}}) of $\gc$ (the CTDE estimator) and $\gq$ (the P-DTDE estimator) with respect to the distributions of $s,a,\delta_Q,\delta_{\widehat{Q}}$ satisfy
    \begin{align}
&\mathbb{E}_{s} \bigg[N^2_i\mathbb{E}_{a_{-i} } \left[ \left(\sum_{j\in \mathcal{V}\setminus \I^i_{\widehat{Q}}}{Q}_j(s,{a}_{\I^j_Q})\right)^2\right]+(\sigma^2_Q  -\sigma^2_{\widehat{Q}})\mathbb{E}_{a}\left[||\gpi ||^2\right]\bigg] \nonumber\\&\hspace{30pt}\leq\var_{s,a, \dq } [\gc] - \var_{s,a ,\dqh} [\gq]  \label{eq:var_bound}\leq M_i^2\sum_{j \in \mathcal{V}\setminus \I^i_{\widehat{Q}}}\left(\epsilon_j\right)^2 +(\sigma^2_Q  -\sigma^2_{\widehat{Q}})\mathbb{E}_{s,a }\left[||\gpi ||^2\right],
\end{align}
    where $M_i \!= \!\underset{s,a}{\text{sup }} ||\gpi||$, $N_i \!  =  \! \underset{s,a}{\text{inf }} ||\gpi||$,\! and $\epsilon_i  \! \!= \!  \underset{s,a}{\text{sup }} |A_i(s,a_{-i},a_i)|.$
    \label{thm:vardiff}
\end{thm}
 \begin{proof}
Refer to Appendix~\ref{sec:vardiff}.
 \end{proof}
The proof is deferred to Appendix~\ref{sec:vardiff}. A similar upper bound was obtained in~\cite{kuba2021settling} for a decentralized policy gradient estimator.
However, in this paper, we analyze the lower bound of the total variance difference obtained by the decomposition of the Q-function. 
Given a fixed amount of samples, if $\mathcal{I}^i_{\widehat{Q}} = \mathcal{V}$, the centralized $Q_i$ has the same number of parameters as the $\widehat{Q}_i$. Therefore, using the same estimation technique, it is expected that $\sigma^2_Q = \sigma^2_{\widehat{Q}}$. However, when $\mathcal{I}^i_{\widehat{Q}} \subset \mathcal{V}$, the irrelevant agents  contribute to additional noise in the estimation of the centralized $Q_i$, leading to a strictly higher variance, i.e., $\sigma^2_Q >\sigma^2_{\widehat{Q}}$.
Therefore, it is reasonable to assume that $ \sigma^2_Q- \sigma^2_{\widehat{Q}} \geq 0$.  

Observe that the error in the total variance of the CTDE and the proposed P-DTDE policy gradient estimators is proportional to the cardinality of $\mathcal{V}\setminus \mathcal{I}^i_{\widehat{Q}}$ with $ \mathcal{I}^i_{\widehat{Q}}= \bigcup_{j\in \I^i_{\text{GD}}} \mathcal{I}^j_{\text{VD}}.$ Thus, a sparser VD graph implies a higher cardinality of $\mathcal{V}\setminus \mathcal{I}^i_{\widehat{Q}}$,  leading to a larger lower bound in Theorem~\ref{thm:vardiff} implying a lower total variance of the P-DTDE estimator.
If $N_i \neq 0$, and $ \sigma^2_Q\geq \sigma^2_{\widehat{Q}}$, the lower bound in~\eqref{eq:bound} is non-zero and increases linearly with the second moment of $\sum Q_j(\cdot),~\forall~j\in \mathcal{V}\setminus \I^i_{\widehat{Q}}$.
This signifies the effect of the proposed approach in achieving variance reduction compared to the centralized policy gradient estimator.

\subsection{Multi-agent structured actor-critic algorithm}
We now discuss the development of model-free multi-agent actor critic algorithms using Theorem~\ref{thm:Qsetdecomp},~\ref{thm:graddecomp}, and~\ref{thm:pgt}. In particular, we employ the deterministic policy gradient in Theorem~\ref{thm:pgt}(a) to develop an off-policy actor critic algorithm based on the P-DTDE paradigm, which we refer to as \textit{multi-agent structured actor critic (MAStAC)} algorithm. The insights of value dependence in Section~\ref{sec:Qdecomp},~\ref{sec:mapg} and Theorem~\ref{thm:pgt} can be leveraged to design other policy gradient MARL algorithms, such as on-policy PPO algorithms.

In MAStAC, each agent $i$ is assigned an actor network $\pi_{\theta_i}$, a critic network $Q_{\mu_i}$, and their corresponding target networks $\pi_{\theta'_i}' $ and $Q_{\mu'_i}'$. Training comprises two steps per agent. In the critic update step (Line 18), each agent updates its critic by minimizing the temporal difference (TD) error using a minibatch of size $M$ sampled from its replay buffer $\mathcal{B}_i$ which stores tuples $(o_i, s_{\mathcal{I}^i_{Q}}, a_{\mathcal{I}^i_{Q}}, r_i, {s^\prime}_{\mathcal{I}^i_{Q}}, o^{\prime}_{\mathcal{I}^i_{Q}})$. Since $(s_{\I^i_Q},a_{\I^i_Q})$ is sufficient to compute $s_{\I^i_Q}'$, it suffices for each agent to maintain its local dependency information in $\mathcal{B}_i$. In the actor update step (Line 16), each actor minimizes its respective loss function. We employ a simultaneous implementation, where all the actor and critic networks are updated at the end of each epoch, while the target networks are updated at a slower timescale. Empirically, we observed that this simultaneous implementation outperforms coordinated implementation, where each actor, critic pair is updated sequentially. The schematic of the MAStAC is shown in Fig.~\ref{fig:NN} (Appendix~\ref{sec:alg}). The MAStAC algorithm is summarized in Algorithm~\ref{alg:cac}.

\begin{algorithm}[h]
\caption{Multi-agent structured actor critic (MAStAC) algorithm}\label{alg:cac}
\begin{algorithmic}[1]
   \FOR{each epoch $e = 1:K$}
  \STATE Initialize a random process $\mathcal{N}$ for exploration.
  \IF{epoch \% max\_episode\_length == 0}
  \STATE {Sample $s_0$ from initial state distribution $\mathcal{D}$}
  \ENDIF
\FOR{each agent $i = 1:N$}{
  \STATE Compute action $a_i = \pi(o_i) + \mathcal{N}$ using the current policy and exploration.}
  \ENDFOR
 \STATE Execute actions $\mathbf{a} = (a_1,\cdots,a_N)$ and receive rewards $\mathbf{r} = (r_1,\cdots,r_N)$, and next observations $\mathbf{o}' = (o'_1,\cdots,o'_N)$.
 \STATE Store $(\mathbf{o}, \mathbf{a}, \mathbf{r}, \mathbf{o}' )$ in replay buffer $\mathcal{B}$.\\
\IF{epoch \% batch\_size == 0}
\FOR{each agent $i = 1:N$ (simultaneous implementation)}
\STATE Sample a minibatch of transitions  ~$(o_i, s_{\mathcal{I}^i_{Q}}, a_{\mathcal{I}^i_{Q}}, r_i, {s^\prime}_{\mathcal{I}^i_{Q}}, o^{\prime}_{\mathcal{I}^i_{Q}})$ of size $M$ from $\mathcal{B}$.
\STATE  Compute TD-target $y^m_i = r^m_i + Q_{\mu'_i}'( s^{\prime m}_{\mathcal{I}^i_{Q}}, \underset{k \in \mathcal{I}^i_{Q}}{\bigcup} a^{\prime}_k)\bigg|_{a^{\prime}_k = \pi'_{\theta'_k}(o^{\prime m}_k)}$.
\STATE Update actor parameters by minimizing the actor loss: 
 \STATE $\nabla \mathcal{L}_{\theta_i}= -\frac{1}{M} \sum_{m\in M}\Bigg[\nabla_{\theta_i} \pi_{\theta_i}(o^m_i)\nabla_{a_i} \sum_{j \in \mathcal{I}^i_{\text{GD}}}Q_{\mu_j}( s^m_{\mathcal{I}^j_{Q}}, a_i,a^m_{\mathcal{I}^j_{Q}\setminus i})\Big|_{a_i =\pi_{\theta_i}(o^m_i)}\Bigg]$ 
\STATE Update critic parameters by minimizing the critic loss:
\STATE $\mathcal{L}_{\mu_i} = \frac{1}{M} \sum_{m \in M} \left( y_i^m - \gamma Q_{\mu_i}( s^m_{\mathcal{I}^i_{Q}}, a^m_{\mathcal{I}^i_{Q}})\right)^2$.
\STATE Update target network parameters for the critic and actor for agent $k$ as
\STATE $\mu'_i \leftarrow \tau \mu_i + (1-\tau) \mu'_i$.
\STATE $\theta'_i \leftarrow \tau \theta_i + (1-\tau) \theta'_i$.
\ENDFOR
\ENDIF
 \ENDFOR
 \end{algorithmic}
\end{algorithm}

Depending on the individual inter-gent couplings, the resulting $\I^i_{\text{VD}}$ might be dense. For example, consider $\mathcal{E}_S =\{(i,j)\in \mathcal{V}^2| 0\leq i-j \leq 1\},~\mathcal{E}_O = \mathcal{E}^\intercal_S,~\mathcal{E}_R= \emptyset$, where {$\mathcal{E}^\intercal_S =\{(j,i)\in \mathcal{V}^2| (i,j)\in \mathcal{E}_S\}$}{, i.e., the edge set `transposed' from $\mathcal{E}_S$}. Although {$\mathcal{E}_S$, $\mathcal{E}_O$, and $\mathcal{E}_R$ are individually sparse}, 
their combined effect yields  a complete value dependency graph such that the decomposition in Theorem~\ref{thm:Qsetdecomp} requires the global state and action.
To improve scalability, we propose a strategy to approximate the dense value dependency graph by a sparse graph.
We approximate the value dependency set of an agent $i\in \mathcal{V}$ at time $t$ by considering only the predecessors of $\mZ_i(\tau)$ within $t\leq \tau\leq t+\kappa \leq T$ time steps, yielding a sparser $\kappa-$approximated value dependency graph $\mathcal{G}^\kappa_{\text{VD}}$.
The intuition behind this approximation is that in the MABN, the states and actions have higher influence on the rewards in near future compared to those in the distant future. We formalize this approximation in Definition~\ref{defin:approxvaldep}.
\begin{defin}[\textbf{\textit{$\kappa-$approximated value dependency graph}}]
     For time-invariant inter-agent couplings, $\forall$ $\kappa \in \mathbb{N}$, an edge $(j,i) \in \mathcal{E}^\kappa_{\text{VD}}(t)\subseteq \mathcal{V}\times \mathcal{V}$, $i,j \in \mathcal{V}$ if and only if either $s_j(t)$ or $a_j(t)$ can reach $\mZ_i(t+1)$ within $\kappa$ traversals of the bidirectional edge in $\mathcal{G}_F$.
\label{defin:approxvaldep}
\end{defin}
Utilizing $\I^{\kappa, i}_{\text{VD}}$ instead of $\I^i_{\text{VD}}$ in Algorithm~\ref{alg:cac} yields the `MAStAC approximated' algorithm. When each agent's reward and observation contain only its own state, {the $\kappa-$approximated value dependency graph}  
recovers the $\kappa-$hop neighbor approximation strategy~\cite{Qu2019,Qu2022,jing2024distributed}.
Our interpretation of $\kappa$ is based on temporal truncation of Q-value dependence which extends naturally to time-varying inter-agent couplings. In contrast, the $\kappa-$approximation in~\cite{Qu2019,Qu2022,jing2024distributed} is in a spatial sense, i.e., it refers to the number of hops in the time-invariant inter-agent couplings.

\section{Numerical experiments}\label{sec:exp}
We evaluate the MAStAC algorithm in Algorithm~\ref{alg:cac} using three numerical examples. {In examples 1 and 2, each actor is initialized with 3 hidden layers of 64 neurons each with ReLu activation.  A softmax function is applied at the output of each actor network to normalize the actions. In example 3,  we use a linear actor network with tanh activation. In all the examples, each critic network is initialized with 3 hidden layers of 64 neurons each with ReLu activation.} The actor and critic network parameters are initialized using xavier glorot initialization. 

We compare Algorithm~\ref{alg:cac} with a variety of MARL baselines, including MADDPG~(\cite{lowe2017}), MATD3~(\cite{MATD3}), 
FACMAC~(\cite{peng2021facmac}). 
Though the deep coordination graph (DCG) proposed in~\cite{bohmer2020deep} is not directly applicable, we implement a variant by decomposing a centralized critic in a similar fashion to the DCG in~\cite{bohmer2020deep}. We refer to this baseline as `Q-deep coordination graph' (QDCG). Specifically, we consider QDCG Decentralized and QDCG Centralized, which correspond the graph used to decompose the critic being a completely disconnected graph and a complete graph, respectively.  The critic in QDCG Decentralized is equivalent to a centralized VDN critic architecture~(\cite{sunehag2017value}). Note that the DCG formulation (and QDCG) is applicable only to homogeneous observation spaces. Since the agents in Example 1 have heterogeneous observation spaces, we compare the QDCG's performance to MAStAC in Example 2 and 3 only. {To ensure a fair comparison, we omit COMA~(\cite{foerster2018counterfactual}) as it applies to stochastic policies while MADDPG~(\cite{lowe2017}), MATD3~(\cite{MATD3}), FACMAC, and MAStAC all use deterministic policies. Similarly, COVDN and COMIX~(\cite{peng2021facmac}) are omitted as they do not have an actor network.}

We run 15 Monte-Carlo (MC) simulations for each algorithm in each example. Figure~\ref{fig:9warehouse} shows the comparison of the total average reward from the simulations across different benchmarks for the three examples.
{Table~\ref{table:final} summarizes the mean and standard deviation of the reward averaged over the final 20\% of the epochs for the 15 simulations.}

\begin{figure*}[htpb]
  \begin{center}
\includegraphics[width=0.32\textwidth]{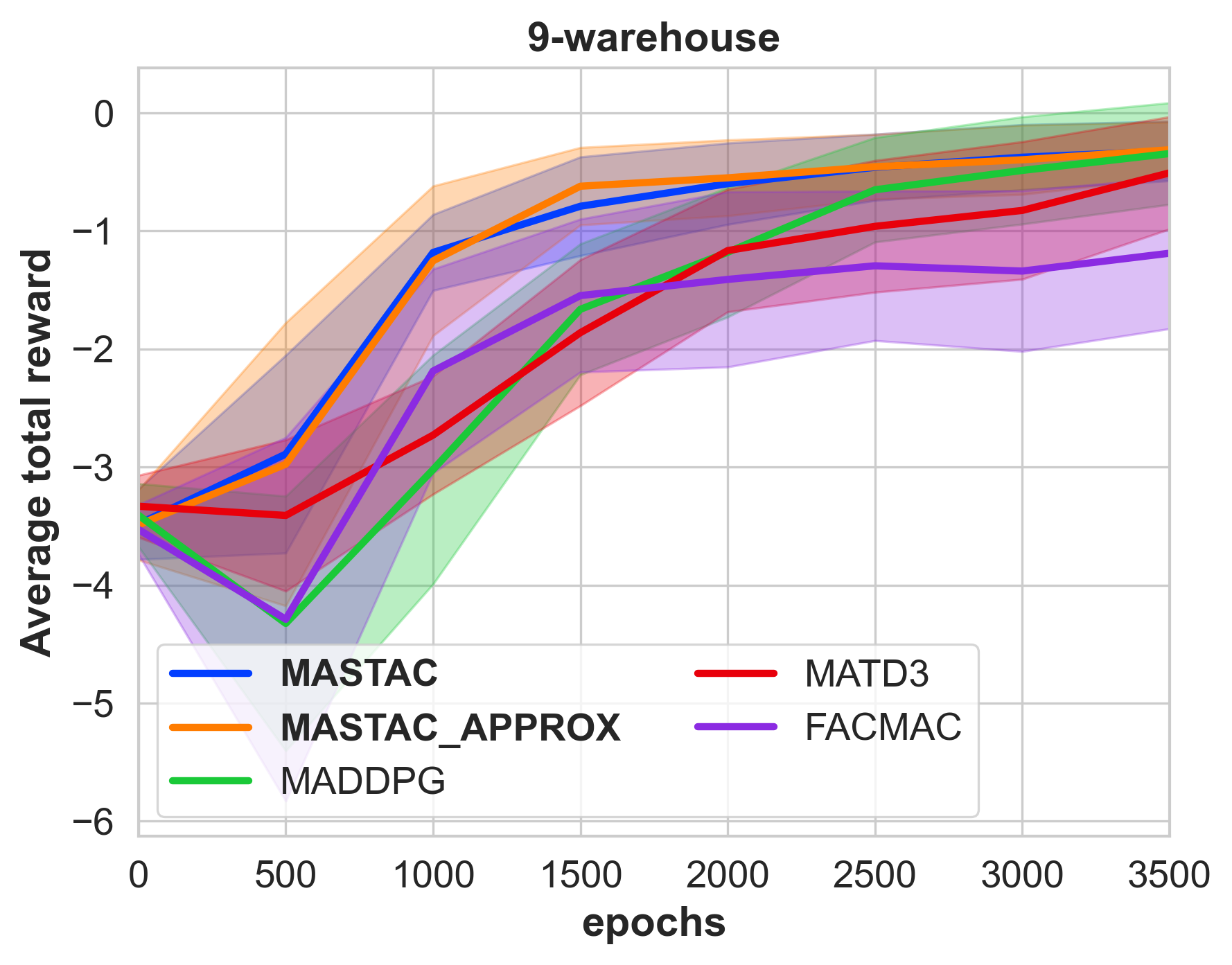}
\includegraphics[width=0.32\textwidth]{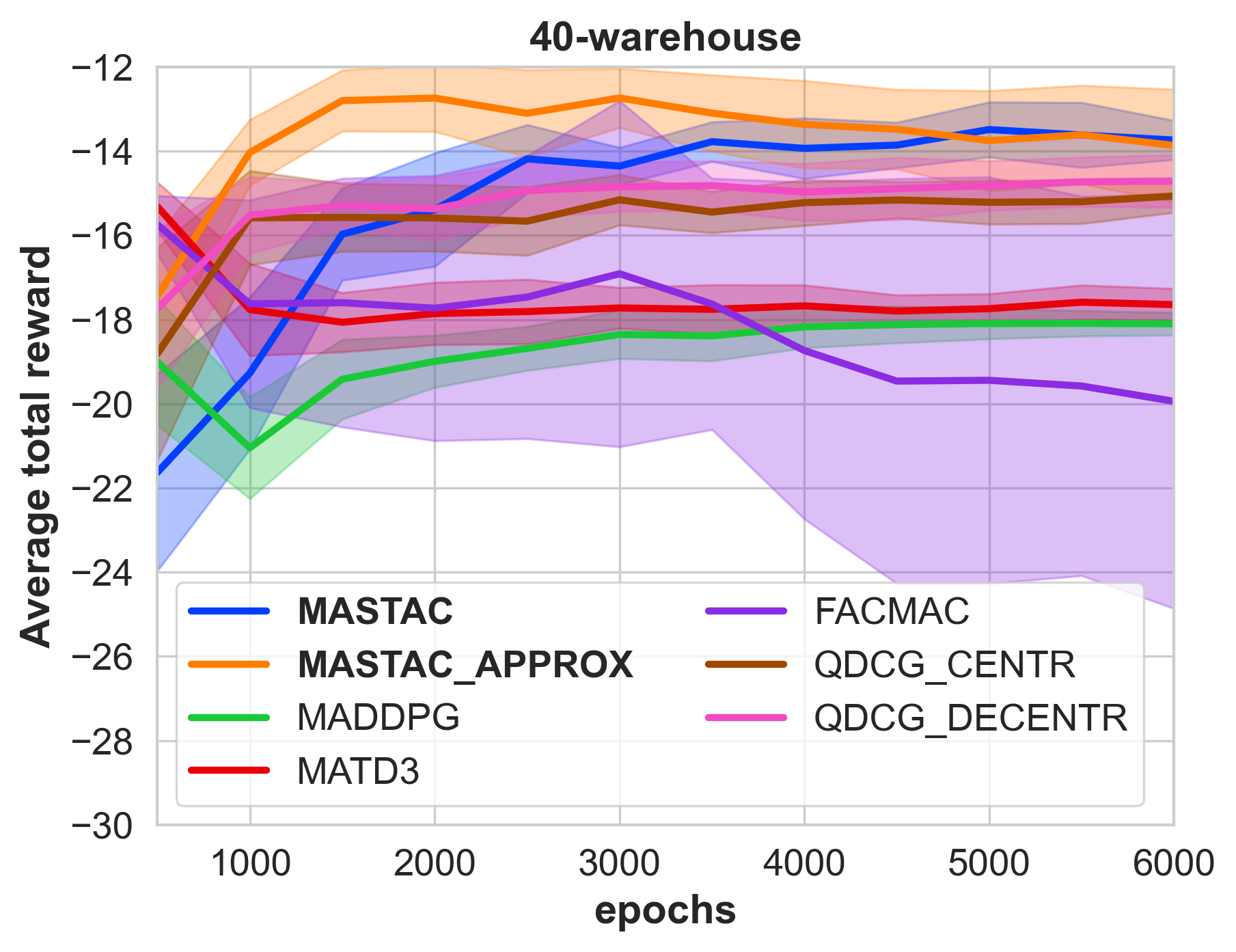}
\includegraphics[width=0.32\textwidth]{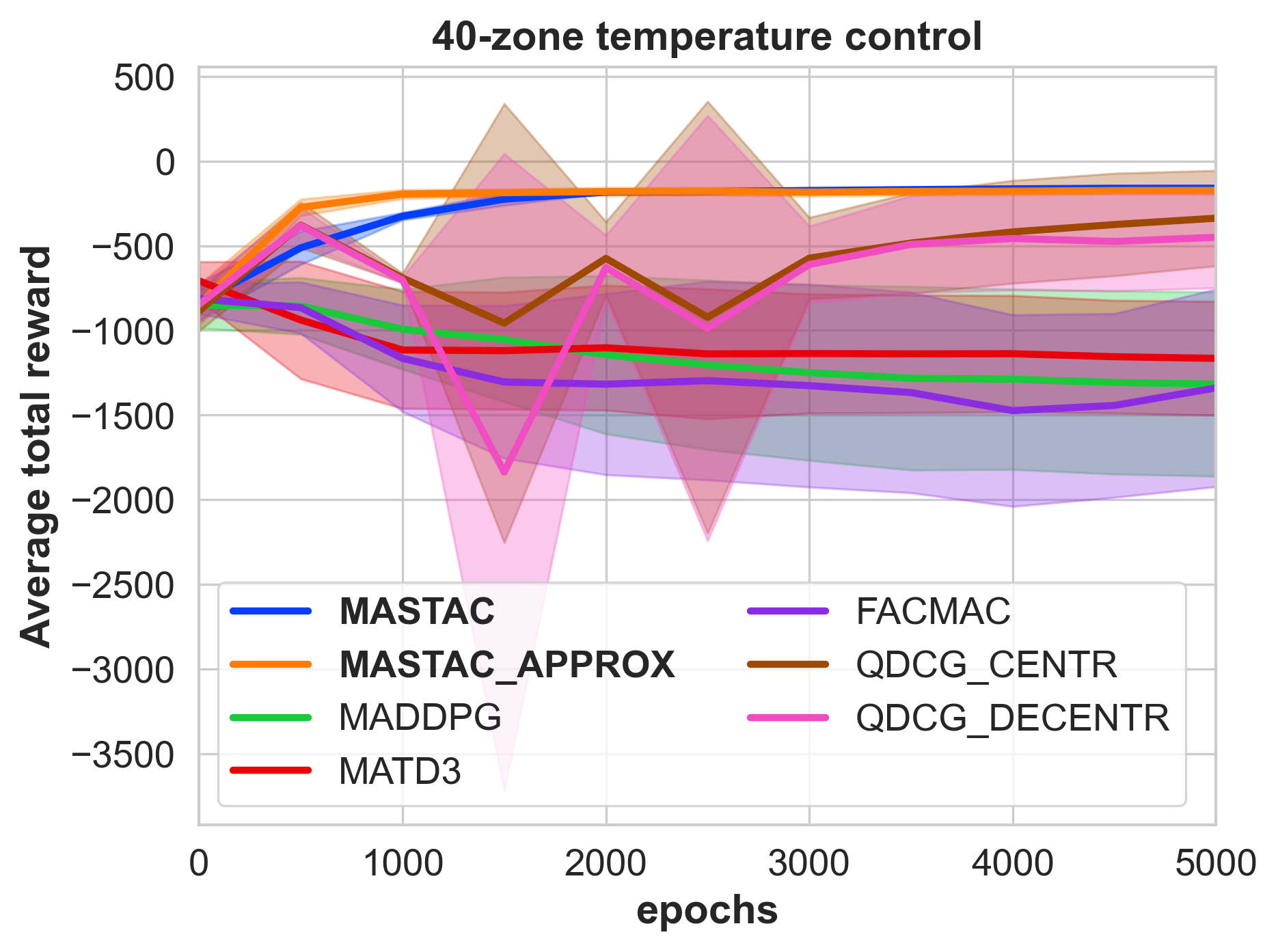}
 \end{center}
 \caption{Comparison of the total average reward for 15 MC simulations. From left to right are Example 1 to 3. }
 \label{fig:9warehouse}
\end{figure*}

\textbf{Example 1: warehouse resource allocation with a sparse $\mathcal{G}_{\text{VD}}$.}
The $N$-warehouse resource allocation problem was proposed in \cite{zhang2020cooperative} and is described in Appendix~\ref{sec:warehouse}. We consider $|\mathcal{V}| = 9$ whose inter-agent couplings are  shown in Fig.~\ref{fig:warehouse9} (Appendix~\ref{sec:warehouse}). In the simulations, we set $m_i(0) = 1$ and $z_i(t) = A_i \sin(t)$, $\forall~i$, where $A_i = 1$, if $i \in \{2,3,5,7\}$, and $A_i = -1$ otherwise. We observe from Fig.~\ref{fig:9warehouse} that the `MAStAC simultaneous and approximated' achieve a faster convergence and have lesser variance compared to the other centralized baselines. 
This corroborates the effect of the decomposition described in Section~\ref{sec:main}. We also see from Table~\ref{table:final} that the `MAStAC simultaneous and approximated' achieve the highest mean total average reward per episode and lower variance compared to the CTDE baselines. Although MADDPG, MATD3, and FACMAC may converge to comparable rewards to MAStAC (Simultaneous and Approximated), \begin{wrapfigure}{r}{0.4\textwidth}
    \centering
    \includegraphics[width=0.4\textwidth]{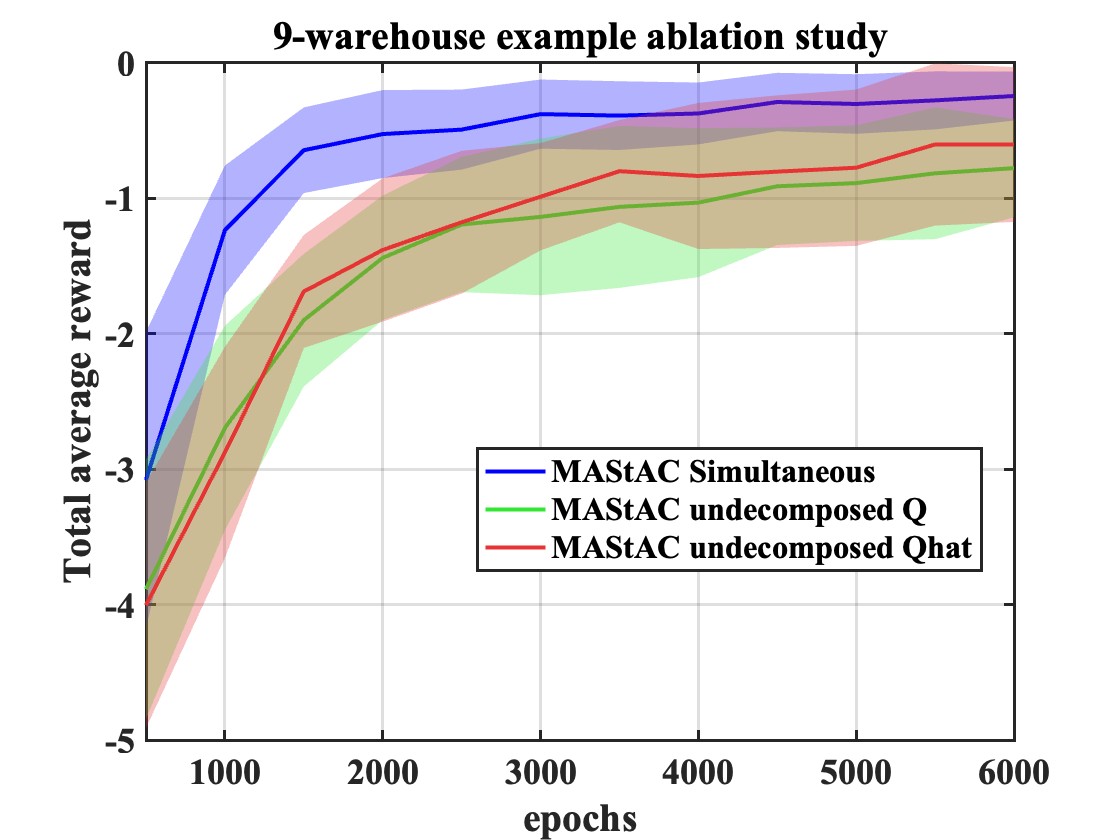}
    \caption{Ablation study of the value dependency set in the 9-warehouse example.}
    \label{fig:9wablated}
\end{wrapfigure}they exhibit much slower rates of convergence compared to MAStAC. The faster rate of convergence of MAStAC highlights its strictly better sample efficiency compared to the CTDE baselines even for this small-scale example.

\textbf{Ablation study} To investigate the effect of $\mathcal{G}_{\text{VD}}$ and $\mathcal{G}_{\text{GD}}$ discussed in Theorem~\ref{thm:Qsetdecomp} and~\ref{thm:graddecomp}, we run the 9-warehouse example using two variants of the `MAStAC simultaneous' algorithm: `MAStAC undecomposed Q' where $\I^i_Q = \mathcal{V}$ in~\eqref{eq:Qhat}, $\forall~i\in \mathcal{V}$, i.e., without using the value dependency set,   and `MAStAC undecomposed Qhat' where the critic of agent $i$ directly learns a $\widehat{Q}^\pi_i(s,a)$ to replace $\widehat{Q}^\pi_i(s_{\I^i_{\widehat{Q}}},a_{\I^i_{\widehat{Q}}})$ in~\eqref{eq:Qhat}. Figure~\ref{fig:9wablated} shows the comparison of the total average reward for the three algorithms.
After $2000$ iterations the three algorithms achieved average rewards $-0.525\pm 0.32$ (MAStAC simultaneous), $-1.439\pm 0.46$ (MAStAC undecomposed Q), and $-1.381\pm 0.53$ (MAStAC undecomposed Qhat). At the $6000^{\text{th}}$ iteration, the converged average rewards were $-0.244\pm 0.18$, $-0.777\pm 0.36$, and $-0.602\pm 0.57$, respectively. These results demonstrate that both
 `MAStAC undecomposed Q' and `MAStAC undecomposed Qhat' exhibit slower convergence, validating the effectiveness of the value dependency set prescribed in Theorem~\ref{thm:Qsetdecomp}. The performance gap underscores the computational and statistical benefits of leveraging inter-agent coupling information for exact Q-function decomposition.

\begin{table*}[htpb]
\caption{Comparison of the mean and standard deviation of the average return for the final 20\% of the epochs for 15 MC simulations across various benchmarks for the three examples.}
\label{table:final}
\begin{center}
\begin{tabular}{cccc}
    Algorithm & 9-warehouse&40-warehouse& \makecell{40-zone \\temperature control}\\       \hline
    MAStAC Sim. & 
    \textbf{-0.35 $\pm$ 0.26}
    &\textbf{-13.68 $\pm$ 0.64}
    &\textbf{-160.86 $\pm$ 5.41}
    \\
    
    MAStAC approx. & \textbf{-0.356 $\pm$ 0.26}
    &\textbf{-13.73 $\pm$ 1.26}
    &\textbf{-174.43 $\pm$ 22.38}
    \\
    MADDPG   
    &-0.417 $\pm$ 0.45
    &-18.09 $\pm$ 0.28
    &-1311.87 $\pm$ 542.13
    \\
    MATD3   
    &-0.67 $\pm$ 0.55
    &-17.62 $\pm$ 0.35
    &-1160 $\pm$ 334.59
    \\

    FACMAC  
    &-1.266$\pm$ 0.66
    &-19.76 $\pm$ 4.72
    &-1390.96$\pm$ 565.46
    \\
    QDCG (Centralized)
    &\longdash[4]
    &-15.13$\pm$ 0.47
    &-356.47 $\pm$ 292.9
    \\
    QDCG (Decentralized)
    &\longdash[4]
    &-14.72$\pm$ 0.6
    &-462.66 $\pm$ 295.75
    \\
    \hline
\end{tabular}
\end{center}
\end{table*}

\textbf{Example 2: 40-warehouse resource allocation with a dense $\mathcal{G}_{\text{VD}}$.}
To demonstrate the efficacy of the $\kappa-$approximation strategy in Definition~\ref{defin:approxvaldep} and its scalability, we consider the warehouse resource allocation problem with $|\mathcal{V}| = 40$, $\mathcal{E}_O = \mathcal{E}_S =\{(i,j)\in \mathcal{V}^2\big||i-j|=1\}\cup\{(1,N),(N,1)\},$ and individual local reward for each agent. This example yields a complete $\mathcal{G}_{\text{VD}}$, i.e., $\I^i_Q = \mathcal{V}$, $\forall~i\in \mathcal{V}.$ We observe from Fig.~\ref{fig:9warehouse} that `MAStAC approximated' with $\kappa = 2$ achieves the fastest convergence across the benchmarks. `MAStAC simultaneous' has a relatively slower convergence compared to the approximated case due to the dense $\mathcal{G}_{\text{VD}}.$ However, it converges to a higher reward compared to other centralized baselines, which demonstrates the scalability of the MAStAC approach. We note that QDCG centralized and QDCG decentralized perform better than MADDPG, MATD3, and FACMAC, owing to the coordination graph based decomposition. However, they converge to lower average rewards 
 when compared to MAStAC (Simultaneous and Approximated), which signifies the effect of the exact decomposition of the Q-function discussed in Theorem~\ref{thm:Qsetdecomp}.  Fig.~\ref{fig:kappacomp} (Appendix~\ref{sec:kappacomp}) shows the comparison of `MAStAC approximated' algorithm for different values of $\kappa$. As $\kappa$ increases, the performance of the algorithm degrades due to the increased density of $\mathcal{G}_{\text{VD}}$.

\textbf{Example 3: structured multi-zone temperature control.}
We adapt the multi-zone temperature control problem~(\cite{Lin2012thermal,Zhang2016hvac,Li2022dirl}) to a two-story building with $20$ zones on each level in a model-free setting as described in Appendix~\ref{sec:temp_control}. We prescribe $\mathcal{E}_O = \{(i,j)\in \mathcal{V}^2\big||i-j|=2\},~\mathcal{E}_S = \mathcal{E}_O \cup \mathcal{E}^\intercal_O,~$ and individual local reward for each agent.  The resulting $\mathcal{G}_{\text{VD}}$ comprises 2 strongly connected components, each consisting of zones on one level.
We observe from Fig.~\ref{fig:9warehouse} that `MAStAC approximated' with $\kappa=2$ achieves the fastest convergence owing to the sparsity in $\mathcal{G}^\kappa_{\text{VD}}.$  From Table~\ref{table:final}, we see that `MAStAC simultaneous' converges to a higher reward than `MAStAC approximated', which corroborates the exactness of the approximation in Theorem~\ref{thm:Qsetdecomp}. Among the centralized baselines, we observe a similar trend as in Example 2 that MADDPG, MATD3, and FACMAC {struggle to learn effectively} whereas QDCG centralized and decentralized converge to a comparatively higher reward. However, MAStAC simultaneous and approximated algorithms converge to higher rewards and have lower variances compared to the QDCG centralized and decentralized algorithms. 

Across the benchmark examples, the `MAStAC simultaneous and approximated' algorithms exhibit reduced variance and faster convergence compared to the existing CTDE MARL algorithms, which demonstrates the advantage of the proposed decomposition approach. {The performance gap between MAStAC and the CTDE algorithms increases considerably as the number of agents increase from Example 1 to Examples 2 and 3. Although the QDCG algorithms perform better than MADDPG, FACMAC, and MATD3, it is outperformed by the `MAStAC simultaneous and approximated' algorithms in Example 2, 3. This can be attributed to the exact decomposition of the Q-function in Theorem~\ref{thm:Qsetdecomp} in contrast to the linear decomposition of the critic into individual and pairwise functions.}

\section{Conclusions}\label{sec:conclusion}

We employ a Bayesian network approach to incorporate structural information into MARL algorithms and achieve efficient learning via exact and approximate decomposition of the gradient of the global state-action value function. We establish conditions under which the stochastic policy gradient based on the exact decomposition exhibits lower variance than the centralized gradient, leading to improved sample complexity and scalability. Based on the decomposition, we introduce the P-DTDE scheme and develop the MAStAC algorithm. Through comparison with benchmark algorithms in numerical experiments, we demonstrate the enhanced convergence speed and reduced variance of the MAStAC. 
A limitation of the proposed approach is the availability of the inter-agent coupling structures. Future work will focus on removing this limitation by learning the structures as latent variables. 
\bibliographystyle{plain}
\bibliography{references}

\newpage
\appendix
 \begin{center}
     \huge{\textbf{Supplementary Material}}
 \end{center}
\noindent
\section{Related work}\label{sec:lit}
\textbf{CTDE in MARL} Our proposed approach sparsifies the centralized structure during training in the CTDE paradigm thereby improving the rate of convergence. The MADDPG algorithm~\cite{lowe2017} employs CTDE paradigm where each agent trains a DDPG algorithm such that the actor uses only the local observations whereas the critic for each agent has access to the global observation and global policy during training. The MADDPG algorithm was demonstrated in various cooperative, competitive and mixed settings. However, as each critic accumulates the information of all of the agents it encounters the curse of dimensionality and thus does not scale to larger number of agents. Several variants of the MADDPG algorithm have been proposed to tackle different problems. For example, see \cite{chu2017, Iqbal2019, Mao2019, Ryu2018, wang2020}. Concurrently, \cite{foerster2018counterfactual} proposed the counterfactual multi-agent policy gradient algorithm (COMA) to address the credit assignment problem in cooperative multi-agent settings utilizing a centralized critic and decentralized actors. COMA algorithm uses a counterfactual baseline to marginalize out the action of a particular agent  and  compares the estimated return to that of the joint action. The key difference between MADDPG and COMA algorithms is that the COMA algorithm trains a single critic shared by all the agents whereas the MADDPG algorithm trains a separate centralized critic for each agent making it applicable to competitive environments with continuous action spaces. In addition, the MADDPG algorithm  is not designed to handle the multi-agent credit assignment problem.  The centralized value functions mitigate several challenges in MARL such as non-stationarity, partial observability, multi-agent credit assignment, equilibrium selection~\cite{christianos2023}. 
However, the exponential growth of the joint action space with the number of agents poses several computational and learning challenges: (1) difficulty in learning the centralized value function due to curse of dimensionality, (2) inefficient decentralized action selection requiring coordination among exponentially large number of joint actions, (3) high computational overhead in evaluating all possible action combinations for greedy action selection~\cite{marl-book}.

\textbf{Value function factorization} An alternative class of algorithms in the MARL literature are those based on value function factorization (see e.g.,~\cite{koller1999,guestrin2001,sunehag2017value,rashid2018qmix,son2019}) that \textit{approximate} centralized value functions with efficient training and decentralized execution. The main idea in value function factorization is to decompose the centralized action value function as a combination of individual action value functions (utilities) that can be efficiently learned and executed in a decentralized fashion. However, to ensure that the decentralized action selection with respect to individual utilities leads to effective joint actions, the decomposition should satisfy the Individual global max (IGM) property~\cite{rashid2018qmix, son2019}. The IGM property states that the greedy joint action with respect to the centralize action value function is equal to the composition of the individual greedy actions of the agents that maximize the individual utility functions. Therefore the IGM property facilitates the efficient computation of the greedy joint action during training by computing the greedy individual actions with respect to the individual utility functions.  Moreover, as the aggregation of the individual action value functions is jointly optimized, the contribution of a particular agent's action to the common reward is distinguishable which addresses the multi-agent credit assignment problem. 

Two representative examples of value function factorization are VDN~\cite{sunehag2017value} and QMIX~\cite{rashid2018qmix} algorithms. VDN proposes a linear decomposition of the centralized action value function which satisfies the IGM property~\cite[Section 9.5.2]{marl-book}. In VDN, each agent consists of an actor network and a critic network that depend on its local information. During training, the centralized action value in the loss is computed as the summation of the individual action values of all the agents allowing the agents to incorporate the impact of other agents' actions. Although the linear decomposition used in VDN is simple, in many real world scenarios the decomposition of the centralized action value function is better represented as a nonlinear function of the individual action value functions. Hence, the QMIX algorithm uses a similar architecture as the VDN except that it uses an additional mixer network that takes the individual actions value functions as input and outputs the decomposed centralized action value function. The mixer network is enforced with positive weights to ensure that the centralized action-value function is monotonic with respect to the individual utilities which is a sufficient condition to satisfy the IGM property~\cite[Section 9.5.3]{marl-book}. However, the structural constraints such as additivity and monotonicity are too restrictive to factorize the centralized value function in some environments. Therefore,~\cite{son2019} proposed the QTRAN algorithm that formulates a slightly less restrictive decomposition which still satisfies the IGM property~\cite{son2019}. This is done by learning an additional centralized value function that corrects the discrepancy between the centralized and decomposed action value functions caused due to the partial observability of agents. Several other variants of the value decomposition methods such as FACMAC~\cite{peng2021facmac}, DOP~\cite{wang2020dop}, VDA2C~\cite{su2021}, VDPPO~\cite{ma2022} have been proposed in the literature. Although the value factorization techniques perform well in many scenarios nonetheless they require a centralized scheme such as linear or nonlinear combination of the individual action value functions. In contrast, in the proposed P-DTDE scheme, each critic requires only the states and actions of agents in its value dependency set during training. In addition, we pursue the exact decomposition in the P-DTDE scheme given the coupling information while the value factorization techniques focus on learning approximated decompositions.

\textit{Coordination graphs} (CG)~\cite{guestrin2001, guestrin2002, kok2006} are another advanced value factorization technique for managing the complexity of joint action spaces. These graphs represent agents as nodes and specify coordination dependencies between them as edges, facilitating the decomposition of the centralized action value function as the sum of individual value functions that depend only on interacting subsets of agents.
CGs are an example of the coupled reward setting. Similarly, the coupled dynamics and partial observability are two other types of inter-agent couplings studied in the MARL literature. However, in the most general setting the optimal action of an agent should take into account the combined effect of the aforementioned couplings. Hence, \cite{jing2024distributed} proposed a distributed RL algorithm based on zeroth order optimization that incorporates all the three types of inter-agent couplings in the cooperative MARL setting. In particular, \cite{jing2024distributed} deduce a new graph called the learning graph that represents the information flow during learning. This work has a strong relevance to the approach present in this paper. Therefore, we highlight the key distinctions and contributions of the proposed approach over those presented in \cite{jing2024distributed}. First, \cite{jing2024distributed} use a set-theoretic approach and treat the couplings as distinct graphs whereas in this work all the inter-agent coupling information is aggregated into a single Bayesian network. Second, the individual value functions require the information of the global state whereas in this work we characterize a subset of agents that affect the individual action value function resulting in the P-DTDE architecture. Third, \cite{jing2024distributed} uses zero-order optimization and a REINFORCE algorithm which have higher sample complexity and variance in parameter updates compared to the actor-critic methods proposed in this work~\cite{lei2022}. 

\textbf{Graph neural networks (GNNs)} have gained attention recently in the field of MAS due to their ability to effectively tackle  graph-structured applications. For example, DGN~\cite{jiang2018graph} uses the graph attention network (GAT)~\cite{velivckovic2017graph} to aggregate the information of neighboring agents for each agent and passes the aggregated information as the input to the action value function. GraphComm~\cite{shen2021graphcomm} categorizes relationships among agents into explicit and implicit categories. The exchange of information about both static and dynamic relationships among agents is facilitated through GAT.   InforMARL~\cite{nayak2023scalable} treats all entities in the environment, including target points, and obstacles as nodes in the graph, and uses a GNN to aggregate information of the local neighborhoods of the agents and input it into the actor and critic. GNNs have also been used to learn effective communication models for cooperative agents (e.g.,~\cite{seraj2022learning}).  These GNN-based methods require additional communications for decentralized execution and may not satisfy observation constraints in the MAS. HetGPPO~\cite{bettini2023heterogeneous} employs a GNN to aggregate information from neighbors to empirically learn local policies and critics.
Recent work~\cite{naderializadeh2020graph,ding2023multiagent,bouton2023multi,hu2023graph} employs GNNs (with a given graph) to learn state-action value factorization via QMIX or VDN, which are approximations.

\textbf{Approximate information states} is another parallel line of research in MARL that has emerged to address partial observability and information asymmetry in multi-agent learning.~\cite{kao2022common} develops a general compression framework with approximate private state representations to construct decentralized policies.~\cite{khan2023cooperative} develops a primal-dual MARL framework with approximate information states for cooperative multi-agent constrained POMDPs. The approximate information states are characterized independently of Lagrange multipliers, enabling adaptation during learning without requiring new representations.~\cite{liu2023partially} demonstrates that information sharing among agents can achieve quasi-efficient statistical and computational guarantees in partially observable stochastic games. These information-theoretic frameworks address a complementary challenge to the structural value decomposition discussed in this paper. We leverage the known inter-agent couplings to identify \textit{which} agents should share information of Q-function evaluation (exact decomposition) whereas the approximate information state-based approaches focus on \textit{what} information should be shared among agents to enable efficient coordination.

In \textbf{distributed RL}~\cite{kar2013QD,zhang2018fully,macua2018diff,zhang2020cooperative,zhang2021finite,li2023f2a2,qu2019value}, the most common method is to use a consensus algorithm for each agent to estimate the global reward function with only local information from neighbors. The performance of such distributed RL algorithms is  worse than centralized RL algorithms, as it takes time to reach consensus and the learning is conducted with the estimated global reward. However, our approach pursues exact decomposition of the policy gradient to facilitate efficient learning in a partially decentralized fashion. When such decomposition involves only a subset of the agents or appropriate approximation methods are designed, our approach is expected to be more efficient than the centralized RL.

There is a recent line of work that derives convergence guarantees for MARL under specific assumptions such as separability~\cite{jin2024approximate}, approximate factorization~\cite{lu2024overcoming}, network sparsity/connectivity~\cite{wang2023multi, hussain2025multi}, IGM property~\cite{wang2021towards, rashid2020weighted}, \textit{and finite/discrete state-,action-spaces}~\cite{jin2024approximate,wang2021towards}. These approaches rely on approximating the underlying MDP or imposing network constraints, which can result in loss of guarantee in densely coupled MDPs or arbitrary network interaction systems.  In contrast, our proposed MABN framework is general: it yields the exact value function decomposition applicable for arbitrarily dense inter-agent couplings and explicitly handles partial observability, remaining effective even in settings where the assumptions for convergence results do not hold. In such a general framework,  obtaining provable convergence guarantees is challenging and a problem of independent interest. As we show (in Section~\ref{sec:vardiff}) that the total variance of policy gradient estimation is lower with the decomposition, we expect that the gradient descent type algorithms will show faster convergence and better sample complexity.
 \section{Illustration of MABN}\label{sec:MABN_exp}

\begin{center}

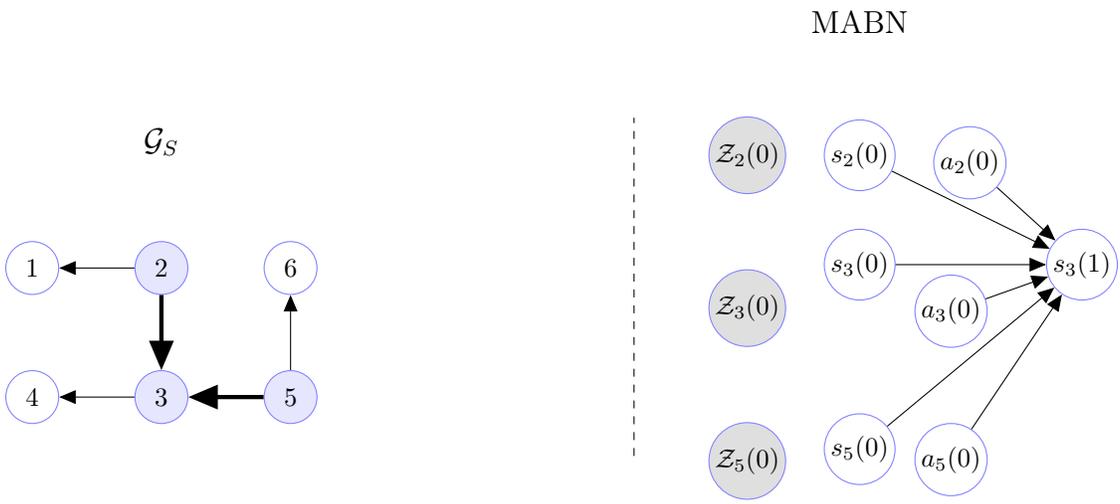
\begin{figure*}[htpb]

\begin{tikzpicture}
\begin{scope}
    \node[latent,draw=blue!50] (A) {$1$};
    \node[latent,draw=blue!50, below=of A] (B) {$4$};
    \node[latent,draw=blue!50 ,fill=blue!10,right=of A] (C) {$2$};
    \node[latent,draw=blue!50,fill=blue!10, below=of C] (D) {$3$};
    \node[latent,draw=blue!50, right=of C] (E) {$6$};
    \node[latent,draw=blue!50 ,fill=blue!10, below=of E] (G) {$5$};
    \node[above=of C] {$\mathcal{G}_S$};
    \edge {C} {A};
    \edge [ultra thick]{C} {D};
    \edge {D} {B};
    \edge [ultra thick]{G} {D};
    \edge {G} {E};
\end{scope}
\begin{scope}[yshift=-1cm]
\draw[dashed] (8,-1.5) -- (8,3);
\end{scope}
\begin{scope}[xshift=10cm]
    \node[latent,draw=blue!50, yshift= 1.5cm,xshift = 1cm] (A') {$s_2(0)$};
    \node[latent,draw=blue!50, below=of A', yshift=0.5cm] (B') {$s_3(0)$};
    \node[latent,draw=blue!50 ,below=of B', yshift=-0.5cm] (C') {$s_5(0)$};
    \node[latent,draw=blue!50, right=of A',xshift= -0.5cm, yshift=-0.1cm] (D') {$a_2(0)$};
    \node[latent,draw=blue!50, below=of D',xshift= -0.25cm] (E') {$a_3(0)$};
    \node[latent,draw=blue!50, below=of E'] (F') {$a_5(0)$};
    \node[latent,draw=blue!50 , right=of B', xshift= 1cm] (G') {$s_3(1)$};
    \node[above=of A'] {MABN};
    \node[obs,draw=blue!50, left=of A', xshift= 0.5cm] (H') {$\mZ_2(0)$};
    \node[obs,draw=blue!50, below=of H'] (I') {$\mZ_3(0)$};
    \node[obs,draw=blue!50, below=of I'] (I') {$\mZ_5(0)$};
    \edge {A'} {G'};
    \edge {B'} {G'};
    \edge {C'} {G'};
    \edge {D'} {G'};
    \edge {E'} {G'};
    \edge {F'} {G'};
\end{scope}

\end{tikzpicture}
\caption{The state graph $\mathcal{G}_S$ of a MAS (left) and the edges corresponding to agents $2,3,5$ in the MABN for $t=0,1$ (right).}
\label{fig:BNs}
\end{figure*}
\end{center}
\newpage
\begin{center}
    \begin{figure*}[htpb]
    
\begin{tikzpicture}
\begin{scope}[yshift=-2cm]
    \node[latent,draw=blue!50, fill=blue!10] (A) {$1$};
    \node[latent,draw=blue!50,fill=blue!10, below=of A] (B) {$4$};
    \node[latent,draw=blue!50 ,fill=blue!10,right=of A] (C) {$2$};
    \node[latent,draw=blue!50,fill=blue!10, below=of C] (D) {$3$};
    \node[latent,draw=blue!50, fill=blue!10, right=of C] (E) {$6$};
    \node[latent,draw=blue!50 ,fill=blue!10, below=of E] (G) {$5$};
    \node[above=of C] {$\mathcal{G}_O$};
    \edge[draw= red, ultra thick] {A} {C};
    \edge [draw= red, ultra thick] {B} {D};
    \edge [draw= red, ultra thick]{E} {G};
\end{scope}
\begin{scope}[yshift=-7cm]
\draw[dashed] (8,-1.5) -- (8,10);
\end{scope}
\begin{scope}[xshift=10cm]
    \node[latent,draw=blue!50, yshift= 1.5cm,xshift = 1cm] (A') {$s_1(0)$};
    \node[latent,draw=blue!50, below=of A', yshift=0.5cm] (B') {$s_2(0)$};
    \node[latent,draw=blue!50 ,below=of B', yshift=-0.5cm] (C') {$s_3(0)$};
    \node[latent,draw=blue!50 ,below=of C', yshift=-0.5cm] (D') {$s_4(0)$};
    \node[latent,draw=blue!50 ,below=of D', yshift=-0.5cm] (E') {$s_5(0)$};
    \node[latent,draw=blue!50 ,below=of E', yshift=-0.5cm] (F') {$s_6(0)$};
    \node[latent,draw=blue!50, right=of A',xshift= -0.5cm, yshift=-0.1cm] (G') {$a_1(0)$};
    \node[latent,draw=blue!50, below=of G',yshift=1cm] (H') {$a_2(0)$};
    \node[latent,draw=blue!50, below=of H'] (I') {$a_3(0)$};
    \node[latent,draw=blue!50, below=of I'] (J') {$a_4(0)$};
    \node[latent,draw=blue!50, below=of J'] (K') {$a_5(0)$};
    \node[latent,draw=blue!50, below=of K'] (L') {$a_6(0)$};
    \node[latent,draw=blue!50 , right=of C', xshift= 1.5cm] (M') {$s_3(1)$};
    \node[above=of A'] {MABN (contd.)};
    \node[obs,draw=blue!50, left=of A', xshift= 0.5cm] (N') {$\mZ_1(0)$};
    \node[obs,draw=blue!50, below=of N'] (O') {$\mZ_2(0)$};
    \node[obs,draw=blue!50, below=of O'] (P') {$\mZ_3(0)$};
    \node[obs,draw=blue!50, below=of P'] (Q') {$\mZ_4(0)$};
    \node[obs,draw=blue!50, below=of Q'] (R') {$\mZ_5(0)$};
    \node[obs,draw=blue!50, below=of R'] (S') {$\mZ_6(0)$};
    \edge {H'} {M'};
    \edge {C'} {M'};
    \edge {I'} {M'};
    \edge {K'} {M'};
    \edge [draw=red]{A'} {G'};\edge[draw=red] {B'} {H'};\edge [draw=red]{C'} {I'};\edge[draw=red] {D'} {J'};\edge[draw=red] {E'} {K'};\edge[draw=red] {F'} {L'};\edge[draw=red] {A'} {H'};\edge[draw=red] {D'} {I'};\edge[draw=red] {F'} {K'};
    \path [draw,->] (B') edge [bend left] node [right] {} (M');
    \path [draw,->] (E') edge [bend right] node [right] {} (M');
\end{scope}

\end{tikzpicture}
\caption{The observation graph $\mathcal{G}_O$ of a MAS (left) and the edges corresponding to agents $1,2,3,4,5,6$ in the MABN (right).}
 \label{fig:BNo}
\end{figure*}
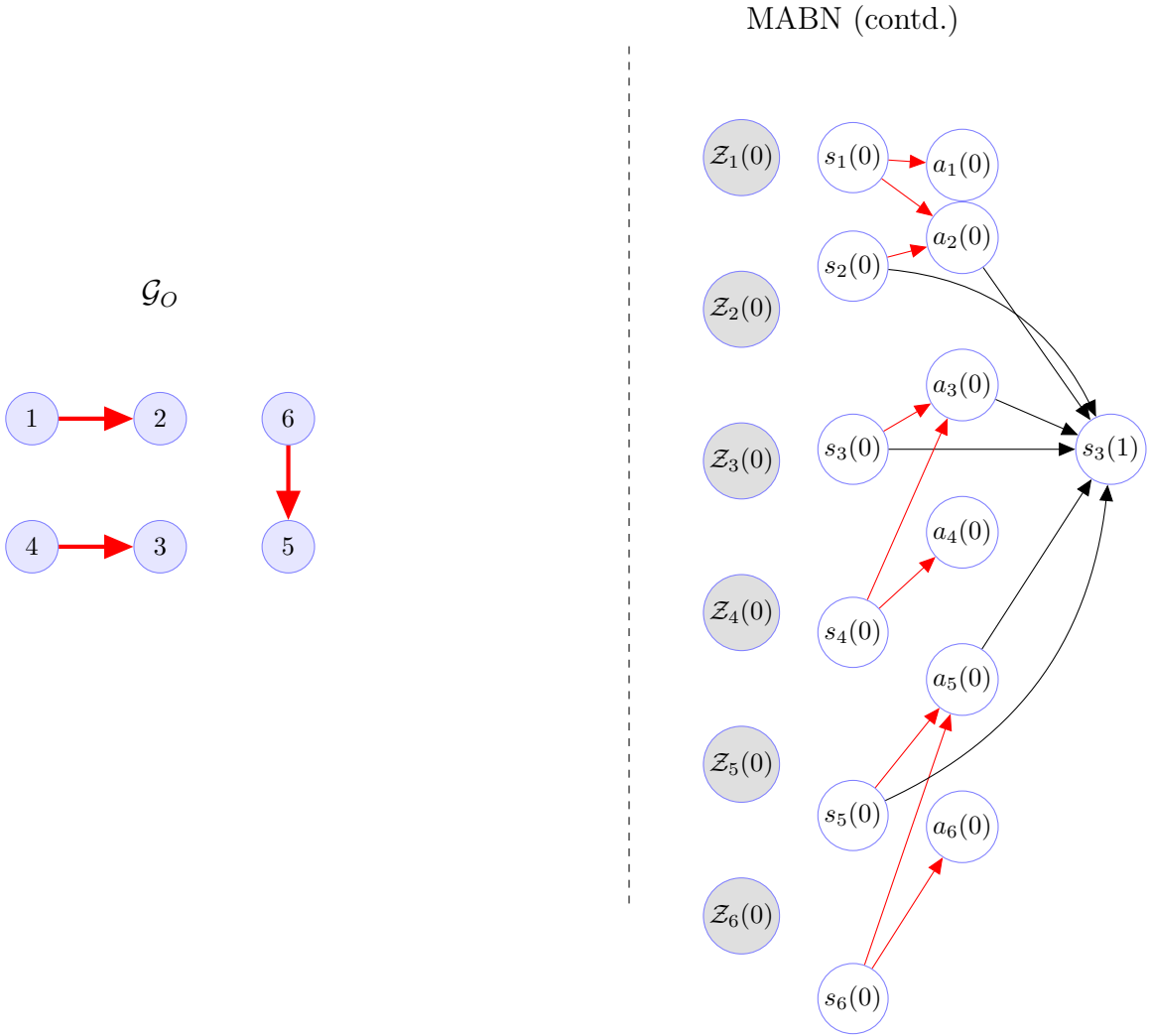
\end{center}

\begin{center}
    \begin{figure*}[htpb]
\begin{tikzpicture}
\begin{scope}[yshift=-2cm]
    \node[latent,draw=blue!50, fill=blue!10] (A) {$1$};
    \node[latent,draw=blue!50,fill=blue!10, below=of A] (B) {$4$};
    \node[latent,draw=blue!50 ,fill=blue!10,right=of A] (C) {$2$};
    \node[latent,draw=blue!50,fill=blue!10, below=of C] (D) {$3$};
    \node[latent,draw=blue!50, fill=blue!10, right=of C] (E) {$6$};
    \node[latent,draw=blue!50 ,fill=blue!10, below=of E] (G) {$5$};
    \node[above=of C] {$\mathcal{G}_R$};
    \edge[draw= violet, ultra thick] {A} {C};
    \edge [draw= violet, ultra thick] {B} {D};
    \edge [draw= violet, ultra thick]{E} {G};
\end{scope}
\begin{scope}[yshift=-7cm]
\draw[dashed] (8,-1.5) -- (8,10);
\end{scope}
\begin{scope}[xshift=10cm]
    \node[latent,draw=blue!50, yshift= 1.5cm,xshift = 1cm] (A') {$s_1(0)$};
    \node[latent,draw=blue!50, below=of A', yshift=0.6cm] (B') {$s_2(0)$};
    \node[latent,draw=blue!50 ,below=of B', yshift=-0.5cm] (C') {$s_3(0)$};
    \node[latent,draw=blue!50 ,below=of C', yshift=-0.5cm] (D') {$s_4(0)$};
    \node[latent,draw=blue!50 ,below=of D', yshift=-0.5cm] (E') {$s_5(0)$};
    \node[latent,draw=blue!50 ,below=of E', yshift=-0.5cm] (F') {$s_6(0)$};
    \node[latent,draw=blue!50, right=of A'] (G') {$a_1(0)$};
    \node[latent,draw=blue!50, below=of G',yshift=0.3cm] (H') {$a_2(0)$};
    \node[latent,draw=blue!50, below=of H',yshift=0.3cm,xshift=-0.3cm] (I') {$a_3(0)$};
    \node[latent,draw=blue!50, below=of I'] (J') {$a_4(0)$};
    \node[latent,draw=blue!50, below=of J',yshift=0.5cm] (K') {$a_5(0)$};
    \node[latent,draw=blue!50, below=of K'] (L') {$a_6(0)$};
    \node[latent,draw=blue!50 , right=of C', xshift= 2cm] (M') {$s_3(1)$};
    \node[above=of A'] {MABN (contd.)};
    \node[obs,draw=blue!50, left=of A', xshift= 0.5cm] (N') {$\mZ_1(0)$};
    \node[obs,draw=blue!50, below=of N'] (O') {$\mZ_2(0)$};
    \node[obs,draw=blue!50, below=of O'] (P') {$\mZ_3(0)$};
    \node[obs,draw=blue!50, below=of P'] (Q') {$\mZ_4(0)$};
    \node[obs,draw=blue!50, below=of Q'] (R') {$\mZ_5(0)$};
    \node[obs,draw=blue!50, below=of R'] (S') {$\mZ_6(0)$};
    \edge {H'} {M'};
    \edge {B'} {M'};
    \edge {C'} {M'};
    \edge {I'} {M'};
    \edge {K'} {M'};
    \edge [draw=red]{A'} {G'};\edge[draw=red] {B'} {H'};\edge [draw=red]{C'} {I'};\edge[draw=red] {D'} {J'};\edge[draw=red] {E'} {K'};\edge[draw=red] {F'} {L'};\edge[draw=red] {A'} {H'};\edge[draw=red] {D'} {I'};\edge[draw=red] {F'} {K'};
    \path [draw,->] (E') edge [bend right] node [right] {} (M');

    \edge [draw=violet]{A'} {N'};\edge[draw=violet] {B'} {O'};\edge [draw=violet]{A'} {O'};\edge[draw=violet] {C'} {P'};\edge[draw=violet] {D'} {Q'};\edge[draw=violet] {J'} {Q'};\edge[draw=violet] {E'} {R'};\edge[draw=violet] {K'} {R'};\edge[draw=violet] {F'} {S'};\edge[draw=violet] {L'} {S'};
    \edge[draw=violet] {D'} {P'};\edge[draw=violet] {J'} {P'};\edge[draw=violet] {F'} {R'};

    \path [draw,->] (L') edge [bend right, draw=violet] node [right] {} (R');
    \path [draw,->] (I') edge [bend right, draw=violet] node [right] {} (P');
    \path [draw,->] (H') edge [bend left, draw=violet] node [right] {} (O');
    \path [draw,->] (G') edge [bend right, draw=violet] node [right] {} (N');
    \path [draw,->] (G') edge [bend left, draw=violet] node [right] {} (O');
\end{scope}

\end{tikzpicture}
\caption{The reward graph $\mathcal{G}_R$ of a MAS (left) and the edges corresponding to agents $1,2,3,4,5,6$ in the MABN (right).}
\label{fig:BNr}
\end{figure*}
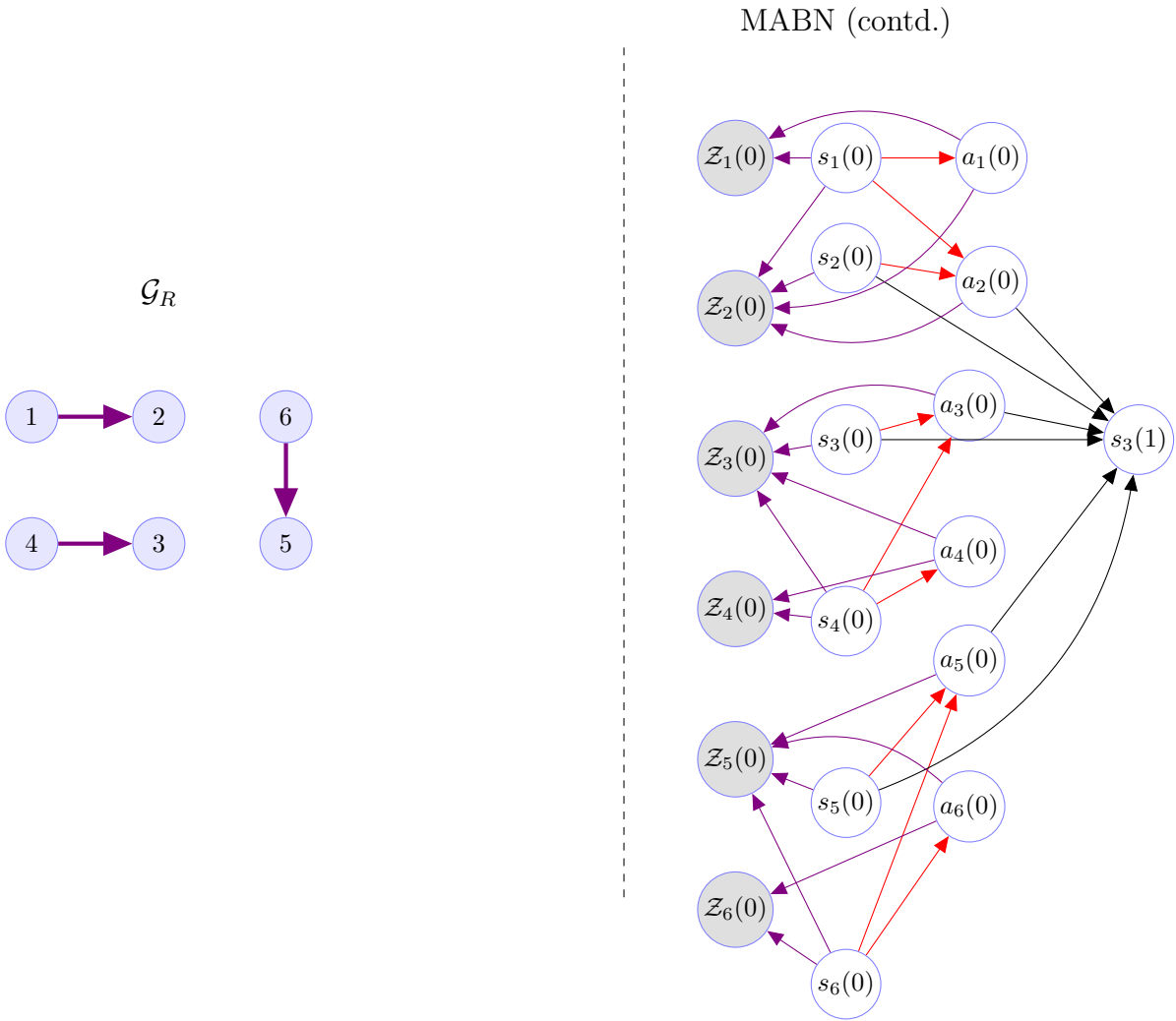
\end{center}

\begin{center}

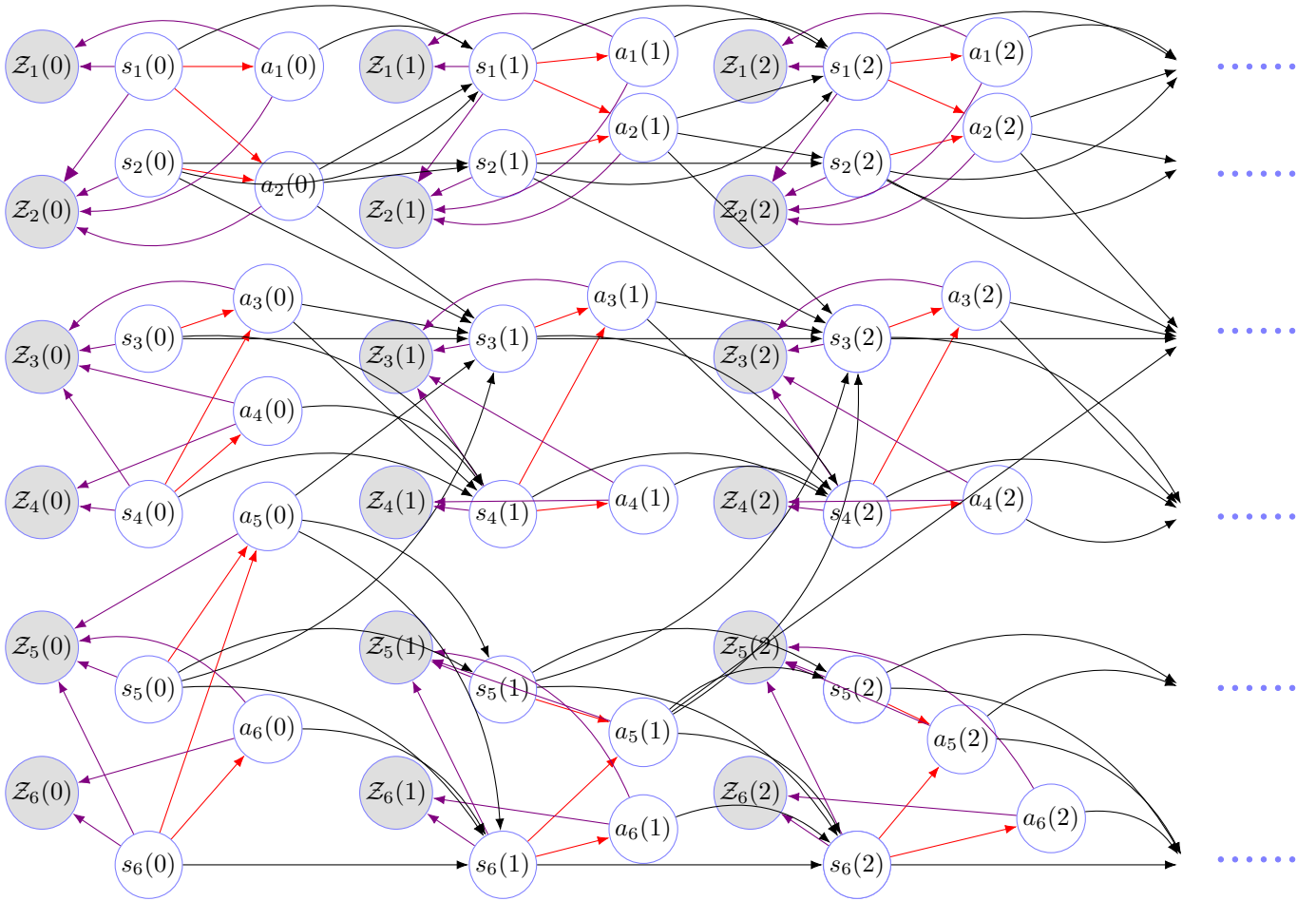
\begin{figure}[htpb]

\begin{tikzpicture}

\begin{scope}[xshift=10cm]
    \node[latent,draw=blue!50, yshift= 1.5cm,xshift = 1cm] (A') {$s_1(0)$};
    \node[latent,draw=blue!50, below=of A', yshift=0.6cm] (B') {$s_2(0)$};
    \node[latent,draw=blue!50 ,below=of B', yshift=-0.5cm] (C') {$s_3(0)$};
    \node[latent,draw=blue!50 ,below=of C', yshift=-0.5cm] (D') {$s_4(0)$};
    \node[latent,draw=blue!50 ,below=of D', yshift=-0.5cm] (E') {$s_5(0)$};
    \node[latent,draw=blue!50 ,below=of E', yshift=-0.5cm] (F') {$s_6(0)$};
    
    \node[latent,draw=blue!50, right=of A'] (G') {$a_1(0)$};
    \node[latent,draw=blue!50, below=of G',yshift=0.3cm] (H') {$a_2(0)$};
    \node[latent,draw=blue!50, below=of H',yshift=0.4cm,xshift=-0.3cm] (I') {$a_3(0)$};
    \node[latent,draw=blue!50, below=of I', yshift= 0.4cm] (J') {$a_4(0)$};
    \node[latent,draw=blue!50, below=of J',yshift=0.5cm] (K') {$a_5(0)$};
    \node[latent,draw=blue!50, below=of K', yshift= -1cm] (L') {$a_6(0)$};
    
     \node[latent,draw=blue!50, right=of A', xshift= 3cm] (A'') {$s_1(1)$};
    \node[latent,draw=blue!50, right=of B', xshift= 3cm] (B'') {$s_2(1)$};
   \node[latent,draw=blue!50 , right=of C', xshift= 3cm] (C'') {$s_3(1)$};
    \node[latent,draw=blue!50 , right=of D', xshift= 3cm] (D'') {$s_4(1)$};
    \node[latent,draw=blue!50, right=of E', xshift= 3cm] (E'') {$s_5(1)$};
    \node[latent,draw=blue!50, right=of F', xshift= 3cm] (F'') {$s_6(1)$};
    \node[latent,draw=blue!50, right=of A'',yshift=0.2cm] (G'') {$a_1(1)$};
    \node[latent,draw=blue!50, right=of B'',yshift=0.5cm] (H'') {$a_2(1)$};
    \node[latent,draw=blue!50, right=of C'',yshift=0.6cm,xshift=-0.3cm] (I'') {$a_3(1)$};
    \node[latent,draw=blue!50, right=of D'', yshift= 0.2cm] (J'') {$a_4(1)$};
    \node[latent,draw=blue!50, right=of E'',yshift=-0.6cm] (K'') {$a_5(1)$};
    \node[latent,draw=blue!50, right=of F'', yshift= 0.5cm] (L'') {$a_6(1)$};

    \node[latent,draw=blue!50, right=of A'', xshift= 3cm] (A''') {$s_1(2)$};
    \node[latent,draw=blue!50, right=of B'', xshift= 3cm] (B''') {$s_2(2)$};
   \node[latent,draw=blue!50 , right=of C'', xshift= 3cm] (C''') {$s_3(2)$};
    \node[latent,draw=blue!50 , right=of D'', xshift= 3cm] (D''') {$s_4(2)$};
    \node[latent,draw=blue!50, right=of E'', xshift= 3cm] (E''') {$s_5(2)$};
    \node[latent,draw=blue!50, right=of F'', xshift= 3cm] (F''') {$s_6(2)$};
    \node[latent,draw=blue!50,right=of A''',yshift=0.2cm] (G''') {$a_1(2)$};
    \node[latent,draw=blue!50, right=of B''',yshift=0.5cm] (H''') {$a_2(2)$};
    \node[latent,draw=blue!50, right=of C''',yshift=0.6cm,xshift=-0.3cm] (I''') {$a_3(2)$};
    \node[latent,draw=blue!50, right=of D''', yshift= 0.2cm] (J''') {$a_4(2)$};
    \node[latent,draw=blue!50, right=of E''',yshift=-0.7cm, xshift= -0.5cm] (K''') {$a_5(2)$};
    \node[latent,draw=blue!50, right=of F''', yshift= 0.65cm, xshift= 0.75cm] (L''') {$a_6(2)$};

    \node[latent,draw=blue!50, right=of A'', xshift= 3cm] (A''') {$s_1(2)$};
    \node[latent,draw=blue!50, right=of B'', xshift= 3cm] (B''') {$s_2(2)$};
   \node[latent,draw=blue!50 , right=of C'', xshift= 3cm] (C''') {$s_3(2)$};
    \node[latent,draw=blue!50 , right=of D'', xshift= 3cm] (D''') {$s_4(2)$};
    \node[latent,draw=blue!50, right=of E'', xshift= 3cm] (E''') {$s_5(2)$};
    \node[latent,draw=blue!50, right=of F'', xshift= 3cm] (F''') {$s_6(2)$};

    \node[right=of A''', xshift= 3cm] (empty1) {$ $};
    \node[right=of B''', xshift= 3cm] (empty2) {$ $};
   \node[right=of C''', xshift= 3cm] (empty3) {$ $};
    \node[right=of D''', xshift= 3cm] (empty4) {$ $};
    \node[right=of E''', xshift= 3cm] (empty5) {$ $};
    \node[right=of F''', xshift= 3cm] (empty6) {$ $};
    
    \node[obs,draw=blue!50, left=of A', xshift= 0.5cm] (N') {$\mZ_1(0)$};
    \node[obs,draw=blue!50, below=of N'] (O') {$\mZ_2(0)$};
    \node[obs,draw=blue!50, below=of O'] (P') {$\mZ_3(0)$};
    \node[obs,draw=blue!50, below=of P'] (Q') {$\mZ_4(0)$};
    \node[obs,draw=blue!50, below=of Q'] (R') {$\mZ_5(0)$};
    \node[obs,draw=blue!50, below=of R'] (S') {$\mZ_6(0)$};
    \node[obs,draw=blue!50, left=of A'', xshift= 0.5cm] (N'') {$\mZ_1(1)$};
    \node[obs,draw=blue!50, below=of N''] (O'') {$\mZ_2(1)$};
    \node[obs,draw=blue!50, below=of O''] (P'') {$\mZ_3(1)$};
    \node[obs,draw=blue!50, below=of P''] (Q'') {$\mZ_4(1)$};
    \node[obs,draw=blue!50, below=of Q''] (R'') {$\mZ_5(1)$};
    \node[obs,draw=blue!50, below=of R''] (S'') {$\mZ_6(1)$};
    \node[obs,draw=blue!50, left=of A''', xshift= 0.5cm] (N''') {$\mZ_1(2)$};
    \node[obs,draw=blue!50, below=of N'''] (O''') {$\mZ_2(2)$};
    \node[obs,draw=blue!50, below=of O'''] (P''') {$\mZ_3(2)$};
    \node[obs,draw=blue!50, below=of P'''] (Q''') {$\mZ_4(2)$};
    \node[obs,draw=blue!50, below=of Q'''] (R''') {$\mZ_5(2)$};
    \node[obs,draw=blue!50, below=of R'''] (S''') {$\mZ_6(2)$};
    \edge [>={Latex[black]}]{H'} {C''};
    \edge [>={Latex[black]}]{B'} {C''};
    \edge [>={Latex[black]}]{C'} {C''};
    \edge [>={Latex[black]}]{I'} {C''};
    \edge [>={Latex[black]}]{K'} {C''};
    \edge [draw=red,>={Latex[red]}]{A'} {G'};
    \edge[draw=red,>={Latex[red]}] {B'} {H'};
    \edge [draw=red,>={Latex[red]}]{C'} {I'};
    \edge[draw=red,>={Latex[red]}] {D'} {J'};
    \edge[draw=red,>={Latex[red]}] {E'} {K'};
    \edge[draw=red,>={Latex[red]}] {F'} {L'};
    \edge[draw=red,>={Latex[red]}] {A'} {H'};
    \edge[draw=red,>={Latex[red]}] {D'} {I'};
    \edge[draw=red,>={Latex[red]}] {F'} {K'};
    \path [draw,->,>={Latex[black]}] (E') edge [bend right] node [right] {} (C'');

    \edge [draw=violet,>={Latex[violet]}]{A'} {N'};\edge[draw=violet,>={Latex[violet]}] {B'} {O'};\edge [draw=violet]{A'} {O'};\edge[draw=violet,>={Latex[violet]}] {C'} {P'};\edge[draw=violet,>={Latex[violet]}] {D'} {Q'};\edge[draw=violet,>={Latex[violet]}] {J'} {Q'};\edge[draw=violet,>={Latex[violet]}] {E'} {R'};\edge[draw=violet,>={Latex[violet]}] {K'} {R'};\edge[draw=violet,>={Latex[violet]}] {F'} {S'};\edge[draw=violet,>={Latex[violet]}] {L'} {S'};
    \edge[draw=violet,>={Latex[violet]}] {D'} {P'};\edge[draw=violet,>={Latex[violet]}] {J'} {P'};\edge[draw=violet,>={Latex[violet]}] {F'} {R'};

    \path [draw,->,>={Latex[violet]}] (L') edge [bend right, draw=violet] node [right] {} (R');
    \path [draw,->,>={Latex[violet]}] (I') edge [bend right, draw=violet] node [right] {} (P');
    \path [draw,->,>={Latex[violet]}] (H') edge [bend left, draw=violet] node [right] {} (O');
    \path [draw,->,>={Latex[violet]}] (G') edge [bend right, draw=violet] node [right] {} (N');
    \path [draw,->,>={Latex[violet]}] (G') edge [bend left, draw=violet] node [right] {} (O');

        \edge [draw=red,>={Latex[red]}]{A''} {G''};
    \edge[draw=red,>={Latex[red]}] {B''} {H''};
    \edge [draw=red,>={Latex[red]}]{C''} {I''};
    \edge[draw=red,>={Latex[red]}] {D''} {J''};
    \edge[draw=red,>={Latex[red]}] {E''} {K''};
    \edge[draw=red,>={Latex[red]}] {F''} {L''};
    \edge[draw=red,>={Latex[red]}] {A''} {H''};
    \edge[draw=red,>={Latex[red]}] {D''} {I''};
    \edge[draw=red,>={Latex[red]}] {F''} {K''};
 \path [draw,->,>={Latex[black]}] (K'') edge [bend right] node [right] {} (C''');
    \edge [draw=violet,>={Latex[violet]}]{A''} {N''};
    \edge[draw=violet,>={Latex[violet]}] {B''} {O''};
\edge [draw=violet]{A''} {O''};
 \edge[draw=violet,>={Latex[violet]}] {C''} {P''};
 \edge[draw=violet,>={Latex[violet]}] {D''} {Q''};
 \edge[draw=violet,>={Latex[violet]}] {J''} {Q''};
 \edge[draw=violet,>={Latex[violet]}] {E''} {R''};
 \edge[draw=violet,>={Latex[violet]}] {K''} {R''};
 \edge[draw=violet,>={Latex[violet]}] {F''} {S''};
 \edge[draw=violet,>={Latex[violet]}] {L''} {S''};
    \edge[draw=violet,>={Latex[violet]}] {D''} {P''};\edge[draw=violet,>={Latex[violet]}] {J''} {P''};\edge[draw=violet,>={Latex[violet]}] {F''} {R''};

    \path [draw,->,>={Latex[violet]}] (L'') edge [bend right, draw=violet] node [right] {} (R'');
    \path [draw,->,>={Latex[violet]}] (I'') edge [bend right, draw=violet] node [right] {} (P'');
    \path [draw,->,>={Latex[violet]}] (H'') edge [bend left, draw=violet] node [right] {} (O'');
    \path [draw,->,>={Latex[violet]}] (G'') edge [bend right, draw=violet] node [right] {} (N'');
    \path [draw,->,>={Latex[violet]}] (G'') edge [bend left, draw=violet] node [right] {} (O'');

     \path [draw,->,>={Latex[black]}] (G') edge [bend left, draw] node [right] {} (A'');
     \path [draw,->,>={Latex[black]}] (A') edge [bend left, draw] node [right] {} (A'');
     \path [draw,->,>={Latex[black]}] (H') edge  node [right] {} (A'');
     \path [draw,->,>={Latex[black]}] (B') edge [bend right] node [right] {} (A'');

\path [draw,->,>={Latex[black]}] (H') edge  node[right] {} (B'');
     \path [draw,->,>={Latex[black]}] (B') edge [] node [right] {} (B'');
     
     \path [draw,->,>={Latex[black]}] (C') edge [bend left, draw] node [right] {} (D'');
     \path [draw,->,>={Latex[black]}] (D') edge [bend left, draw] node [right] {} (D'');
     \path [draw,->,>={Latex[black]}] (I') edge  node [right] {} (D'');
     \path [draw,->,>={Latex[black]}] (J') edge [bend left] node [right] {} (D'');

     \path [draw,->,>={Latex[black]}] (K') edge [bend left, draw] node [right] {} (E'');
     \path [draw,->,>={Latex[black]}] (E') edge [bend left, draw] node [right] {} (E'');

     \path [draw,->,>={Latex[black]}] (K') edge [bend left, draw] node [right] {} (F'');
     \path [draw,->,>={Latex[black]}] (E') edge [bend left, draw] node [right] {} (F'');
     \path [draw,->,>={Latex[black]}] (L') edge [bend left, draw] node [right] {} (F'');
     \path [draw,->,>={Latex[black]}] (F') edge node [right] {} (F'');

\edge [>={Latex[black]}]{H''} {C'''};
    \edge [>={Latex[black]}]{B''} {C'''};
    \edge [>={Latex[black]}]{C''} {C'''};
    \edge [>={Latex[black]}]{I''} {C'''};
       \edge [draw=red,>={Latex[red]}]{A'''} {G'''};
    \edge[draw=red,>={Latex[red]}] {B'''} {H'''};
    \edge [draw=red,>={Latex[red]}]{C'''} {I'''};
    \edge[draw=red,>={Latex[red]}] {D'''} {J'''};
    \edge[draw=red,>={Latex[red]}] {E'''} {K'''};
    \edge[draw=red,>={Latex[red]}] {F'''} {L'''};
    \edge[draw=red,>={Latex[red]}] {A'''} {H'''};
    \edge[draw=red,>={Latex[red]}] {D'''} {I'''};
    \edge[draw=red,>={Latex[red]}] {F'''} {K'''};
    \path [draw,->,>={Latex[black]}] (E'') edge [bend right] node [right] {} (C''');

    \edge [draw=violet,>={Latex[violet]}]{A'''} {N'''};
    \edge[draw=violet,>={Latex[violet]}] {B'''} {O'''};
    \edge [draw=violet]{A'''} {O'''};\edge[draw=violet,>={Latex[violet]}] {C'''} {P'''};\edge[draw=violet,>={Latex[violet]}] {D'''} {Q'''};\edge[draw=violet,>={Latex[violet]}] {J'''} {Q'''};\edge[draw=violet,>={Latex[violet]}] {E'''} {R'''};\edge[draw=violet,>={Latex[violet]}] {K'''} {R'''};\edge[draw=violet,>={Latex[violet]}] {F'''} {S'''};\edge[draw=violet,>={Latex[violet]}] {L'''} {S'''};
    \edge[draw=violet,>={Latex[violet]}] {D'''} {P'''};\edge[draw=violet,>={Latex[violet]}] {J'''} {P'''};\edge[draw=violet,>={Latex[violet]}] {F'''} {R'''};

    \path [draw,->,>={Latex[violet]}] (L''') edge [bend right, draw=violet] node [right] {} (R''');
    \path [draw,->,>={Latex[violet]}] (I''') edge [bend right, draw=violet] node [right] {} (P''');
    \path [draw,->,>={Latex[violet]}] (H''') edge [bend left, draw=violet] node [right] {} (O''');
    \path [draw,->,>={Latex[violet]}] (G''') edge [bend right, draw=violet] node [right] {} (N''');
    \path [draw,->,>={Latex[violet]}] (G''') edge [bend left, draw=violet] node [right] {} (O''');

     \path [draw,->,>={Latex[black]}] (G'') edge [bend left, draw] node [right] {} (A''');
     \path [draw,->,>={Latex[black]}] (A'') edge [bend left, draw] node [right] {} (A''');
     \path [draw,->,>={Latex[black]}] (H'') edge  node [right] {} (A''');
     \path [draw,->,>={Latex[black]}] (B'') edge [bend right] node [right] {} (A''');
     \path [draw,->,>={Latex[black]}] (H'') edge  node [right] {} (B''');
     \path [draw,->,>={Latex[black]}] (B'') edge [] node [right] {} (B''');

     \path [draw,->,>={Latex[black]}] (C'') edge [bend left, draw] node [right] {} (D''');
     \path [draw,->,>={Latex[black]}] (D'') edge [bend left, draw] node [right] {} (D''');
     \path [draw,->,>={Latex[black]}] (I'') edge  node [right] {} (D''');
     \path [draw,->,>={Latex[black]}] (J'') edge [bend left] node [right] {} (D''');

     \path [draw,->,>={Latex[black]}] (K'') edge [bend left, draw] node [right] {} (E''');
     \path [draw,->,>={Latex[black]}] (E'') edge [bend left, draw] node [right] {} (E''');

     \path [draw,->,>={Latex[black]}] (K'') edge [bend left, draw] node [right] {} (F''');
     \path [draw,->,>={Latex[black]}] (E'') edge [bend left, draw] node [right] {} (F''');
     \path [draw,->,>={Latex[black]}] (L'') edge [bend left, draw] node [right] {} (F''');
     \path [draw,->,>={Latex[black]}] (F'') edge node [right] {} (F''');

\path [draw,->,>={Latex[black]}] (G''') edge [bend left, draw] node [right] {} (empty1);
     \path [draw,->,>={Latex[black]}] (A''') edge [bend left, draw] node [right] {} (empty1);
     \path [draw,->,>={Latex[black]}] (H''') edge  node [right] {} (empty1);
     \path [draw,->,>={Latex[black]}] (B''') edge [bend right] node [right] {} (empty1);
     \path [draw,->,>={Latex[black]}] (H''') edge  node [right] {} (empty2);
     \path [draw,->,>={Latex[black]}] (B''') edge [bend right] node [right] {} (empty2);
     \edge [>={Latex[black]}]{H'''} {empty3};
    \edge [>={Latex[black]}]{B'''} {empty3};
    \edge [>={Latex[black]}]{C'''} {empty3};
    \edge [>={Latex[black]}]{I'''} {empty3};
    \path [draw,->,>={Latex[black]}] (K'') edge [] node [right] {} (empty3);
    \path [draw,->,>={Latex[black]}] (C''') edge [bend left, draw] node [right] {} (empty4);
     \path [draw,->,>={Latex[black]}] (D''') edge [bend left, draw] node [right] {} (empty4);
     \path [draw,->,>={Latex[black]}] (I''') edge  node [right] {} (empty4);
     \path [draw,->,>={Latex[black]}] (J''') edge [bend right] node [right] {} (empty4);
    \path [draw,->,>={Latex[black]}] (K''') edge [bend left, draw] node [right] {} (empty5);
     \path [draw,->,>={Latex[black]}] (E''') edge [bend left, draw] node [right] {}(empty5);
    \path [draw,->,>={Latex[black]}] (K''') edge [bend left, draw] node [right] {} (empty6);
     \path [draw,->,>={Latex[black]}] (E''') edge [bend left, draw] node [right] {} (empty6);
     \path [draw,->,>={Latex[black]}] (L''') edge [bend left, draw] node [right] {} (empty6);
     \path [draw,->,>={Latex[black]}] (F''') edge node [right] {} (empty6);

\node[mark size=1pt,color=blue!50] at (16,1.5) {\pgfuseplotmark{*}};
    \node[mark size=1pt,color=blue!50] at (16.2,1.5) {\pgfuseplotmark{*}};
    \node[mark size=1pt,color=blue!50] at (16.4,1.5) {\pgfuseplotmark{*}};
    \node[mark size=1pt,color=blue!50] at (16.6,1.5) {\pgfuseplotmark{*}};
    \node[mark size=1pt,color=blue!50] at (16.8,1.5) {\pgfuseplotmark{*}};
    \node[mark size=1pt,color=blue!50] at (17,1.5) {\pgfuseplotmark{*}};
    \node[mark size=1pt,color=blue!50] at (17.2,1.5) {\pgfuseplotmark{}};

\node[mark size=1pt,color=blue!50] at (16,0) {\pgfuseplotmark{*}};
    \node[mark size=1pt,color=blue!50] at (16.2,0) {\pgfuseplotmark{*}};
    \node[mark size=1pt,color=blue!50] at (16.4,0) {\pgfuseplotmark{*}};
    \node[mark size=1pt,color=blue!50] at (16.6,0) {\pgfuseplotmark{*}};
    \node[mark size=1pt,color=blue!50] at (16.8,0) {\pgfuseplotmark{*}};
    \node[mark size=1pt,color=blue!50] at (17,0) {\pgfuseplotmark{*}};
    \node[mark size=1pt,color=blue!50] at (17.2,0) {\pgfuseplotmark{}};

 \node[mark size=1pt,color=blue!50] at (16,-2.2) {\pgfuseplotmark{*}};
    \node[mark size=1pt,color=blue!50] at (16.2,-2.2) {\pgfuseplotmark{*}};
    \node[mark size=1pt,color=blue!50] at (16.4,-2.2) {\pgfuseplotmark{*}};
    \node[mark size=1pt,color=blue!50] at (16.6,-2.2) {\pgfuseplotmark{*}};
    \node[mark size=1pt,color=blue!50] at (16.8,-2.2) {\pgfuseplotmark{*}};
    \node[mark size=1pt,color=blue!50] at (17,-2.2) {\pgfuseplotmark{*}};
    \node[mark size=1pt,color=blue!50] at (17.2,-2.2) {\pgfuseplotmark{}};

    \node[mark size=1pt,color=blue!50] at (16,-4.8) {\pgfuseplotmark{*}};
    \node[mark size=1pt,color=blue!50] at (16.2,-4.8) {\pgfuseplotmark{*}};
    \node[mark size=1pt,color=blue!50] at (16.4,-4.8) {\pgfuseplotmark{*}};
    \node[mark size=1pt,color=blue!50] at (16.6,-4.8) {\pgfuseplotmark{*}};
    \node[mark size=1pt,color=blue!50] at (16.8,-4.8) {\pgfuseplotmark{*}};
    \node[mark size=1pt,color=blue!50] at (17,-4.8) {\pgfuseplotmark{*}};
    \node[mark size=1pt,color=blue!50] at (17.2,-4.8) {\pgfuseplotmark{}};

     \node[mark size=1pt,color=blue!50] at (16,-7.2) {\pgfuseplotmark{*}};
    \node[mark size=1pt,color=blue!50] at (16.2,-7.2) {\pgfuseplotmark{*}};
    \node[mark size=1pt,color=blue!50] at (16.4,-7.2) {\pgfuseplotmark{*}};
    \node[mark size=1pt,color=blue!50] at (16.6,-7.2) {\pgfuseplotmark{*}};
    \node[mark size=1pt,color=blue!50] at (16.8,-7.2) {\pgfuseplotmark{*}};
    \node[mark size=1pt,color=blue!50] at (17,-7.2) {\pgfuseplotmark{*}};
    \node[mark size=1pt,color=blue!50] at (17.2,-7.2) {\pgfuseplotmark{}};

    \node[mark size=1pt,color=blue!50] at (16,-9.6) {\pgfuseplotmark{*}};
    \node[mark size=1pt,color=blue!50] at (16.2,-9.6) {\pgfuseplotmark{*}};
    \node[mark size=1pt,color=blue!50] at (16.4,-9.6) {\pgfuseplotmark{*}};
    \node[mark size=1pt,color=blue!50] at (16.6,-9.6) {\pgfuseplotmark{*}};
    \node[mark size=1pt,color=blue!50] at (16.8,-9.6) {\pgfuseplotmark{*}};
    \node[mark size=1pt,color=blue!50] at (17,-9.6) {\pgfuseplotmark{*}};
    \node[mark size=1pt,color=blue!50] at (17.2,-9.6) {\pgfuseplotmark{}};
\end{scope}

\end{tikzpicture}
\caption{The full MABN $\mathcal{G}_B$ of the MAS shown in Figs~\ref{fig:BNs}-~\ref{fig:BNr}.}
\label{fig:BNfull}
\end{figure}
\end{center}

\begin{center}

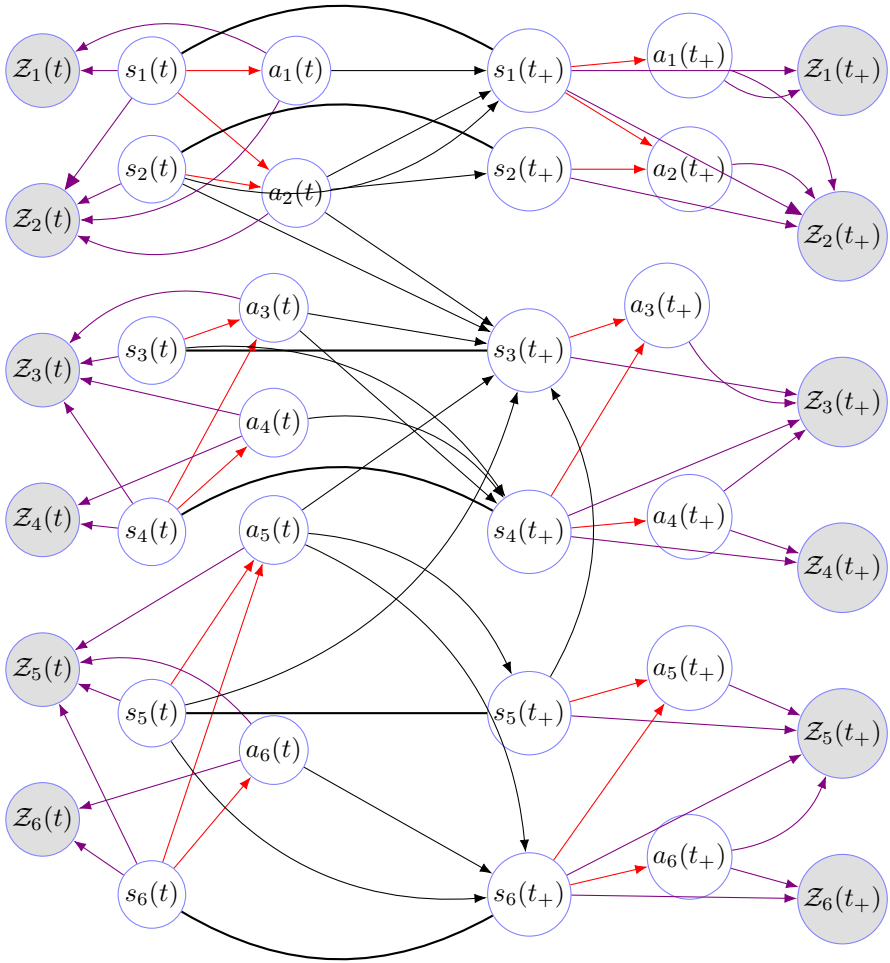
\begin{figure}[htpb]

\begin{tikzpicture}

\begin{scope}[xshift=10cm]
    \node[latent,draw=blue!50, yshift= 1.5cm,xshift = 1cm] (A') {$s_1(t)$};
    \node[latent,draw=blue!50, below=of A', yshift=0.6cm] (B') {$s_2(t)$};
    \node[latent,draw=blue!50 ,below=of B', yshift=-0.5cm] (C') {$s_3(t)$};
    \node[latent,draw=blue!50 ,below=of C', yshift=-0.5cm] (D') {$s_4(t)$};
    \node[latent,draw=blue!50 ,below=of D', yshift=-0.5cm] (E') {$s_5(t)$};
    \node[latent,draw=blue!50 ,below=of E', yshift=-0.5cm] (F') {$s_6(t)$};
    
    \node[latent,draw=blue!50, right=of A'] (G') {$a_1(t)$};
    \node[latent,draw=blue!50, below=of G',yshift=0.3cm] (H') {$a_2(t)$};
    \node[latent,draw=blue!50, below=of H',yshift=0.4cm,xshift=-0.3cm] (I') {$a_3(t)$};
    \node[latent,draw=blue!50, below=of I', yshift= 0.4cm] (J') {$a_4(t)$};
    \node[latent,draw=blue!50, below=of J',yshift=0.5cm] (K') {$a_5(t)$};
    \node[latent,draw=blue!50, below=of K', yshift= -1cm] (L') {$a_6(t)$};
    
     \node[latent,draw=blue!50, right=of A', xshift= 3cm] (A'') {$s_1(t_+)$};
    \node[latent,draw=blue!50, right=of B', xshift= 3cm] (B'') {$s_2(t_+)$};
   \node[latent,draw=blue!50 , right=of C', xshift= 3cm] (C'') {$s_3(t_+)$};
    \node[latent,draw=blue!50 , right=of D', xshift= 3cm] (D'') {$s_4(t_+)$};
    \node[latent,draw=blue!50, right=of E', xshift= 3cm] (E'') {$s_5(t_+)$};
    \node[latent,draw=blue!50, right=of F', xshift= 3cm] (F'') {$s_6(t_+)$};
    \node[latent,draw=blue!50, right=of A'',yshift=0.2cm] (G'') {$a_1(t_+)$};
    \node[latent,draw=blue!50, right=of B'',yshift=0cm] (H'') {$a_2(t_+)$};
    \node[latent,draw=blue!50, right=of C'',yshift=0.6cm,xshift=-0.3cm] (I'') {$a_3(t_+)$};
    \node[latent,draw=blue!50, right=of D'', yshift= 0.2cm] (J'') {$a_4(t_+)$};
    \node[latent,draw=blue!50, right=of E'',yshift=0.6cm] (K'') {$a_5(t_+)$};
    \node[latent,draw=blue!50, right=of F'', yshift= 0.5cm] (L'') {$a_6(t_+)$};

    \node[obs,draw=blue!50, left=of A', xshift= 0.5cm] (N') {$\mZ_1(t)$};
    \node[obs,draw=blue!50, below=of N'] (O') {$\mZ_2(t)$};
    \node[obs,draw=blue!50, below=of O'] (P') {$\mZ_3(t)$};
    \node[obs,draw=blue!50, below=of P'] (Q') {$\mZ_4(t)$};
    \node[obs,draw=blue!50, below=of Q'] (R') {$\mZ_5(t)$};
    \node[obs,draw=blue!50, below=of R'] (S') {$\mZ_6(t)$};
    \node[obs,draw=blue!50, right=of A'', xshift= 2cm] (N'') {$\mZ_1(t_+)$};
    \node[obs,draw=blue!50, below=of N''] (O'') {$\mZ_2(t_+)$};
    \node[obs,draw=blue!50, below=of O''] (P'') {$\mZ_3(t_+)$};
    \node[obs,draw=blue!50, below=of P''] (Q'') {$\mZ_4(t_+)$};
    \node[obs,draw=blue!50, below=of Q''] (R'') {$\mZ_5(t_+)$};
    \node[obs,draw=blue!50, below=of R''] (S'') {$\mZ_6(t_+)$};

    \edge [>={Latex[black]}]{H'} {C''};
    \edge [>={Latex[black]}]{B'} {C''};
    \edge [-,>={Latex[black]}, thick]{C'} {C''};
    \edge [>={Latex[black]}]{I'} {C''};
    \edge [>={Latex[black]}]{K'} {C''};
    \edge [draw=red,>={Latex[red]}]{A'} {G'};
    \edge[draw=red,>={Latex[red]}] {B'} {H'};
    \edge [draw=red,>={Latex[red]}]{C'} {I'};
    \edge[draw=red,>={Latex[red]}] {D'} {J'};
    \edge[draw=red,>={Latex[red]}] {E'} {K'};
    \edge[draw=red,>={Latex[red]}] {F'} {L'};
    \edge[draw=red,>={Latex[red]}] {A'} {H'};
    \edge[draw=red,>={Latex[red]}] {D'} {I'};
    \edge[draw=red,>={Latex[red]}] {F'} {K'};
    \path [draw,->,>={Latex[black]}] (E') edge [bend right] node [right] {} (C'');

    \edge [draw=violet,>={Latex[violet]}]{A'} {N'};\edge[draw=violet,>={Latex[violet]}] {B'} {O'};\edge [draw=violet]{A'} {O'};\edge[draw=violet,>={Latex[violet]}] {C'} {P'};\edge[draw=violet,>={Latex[violet]}] {D'} {Q'};\edge[draw=violet,>={Latex[violet]}] {J'} {Q'};\edge[draw=violet,>={Latex[violet]}] {E'} {R'};\edge[draw=violet,>={Latex[violet]}] {K'} {R'};\edge[draw=violet,>={Latex[violet]}] {F'} {S'};\edge[draw=violet,>={Latex[violet]}] {L'} {S'};
    \edge[draw=violet,>={Latex[violet]}] {D'} {P'};\edge[draw=violet,>={Latex[violet]}] {J'} {P'};\edge[draw=violet,>={Latex[violet]}] {F'} {R'};

    \path [draw,->,>={Latex[violet]}] (L') edge [bend right, draw=violet] node [right] {} (R');
    \path [draw,->,>={Latex[violet]}] (I') edge [bend right, draw=violet] node [right] {} (P');
    \path [draw,->,>={Latex[violet]}] (H') edge [bend left, draw=violet] node [right] {} (O');
    \path [draw,->,>={Latex[violet]}] (G') edge [bend right, draw=violet] node [right] {} (N');
    \path [draw,->,>={Latex[violet]}] (G') edge [bend left, draw=violet] node [right] {} (O');

        \edge [draw=red,>={Latex[red]}]{A''} {G''};
    \edge[draw=red,>={Latex[red]}] {B''} {H''};
    \edge [draw=red,>={Latex[red]}]{C''} {I''};
    \edge[draw=red,>={Latex[red]}] {D''} {J''};
    \edge[draw=red,>={Latex[red]}] {E''} {K''};
    \edge[draw=red,>={Latex[red]}] {F''} {L''};
    \edge[draw=red,>={Latex[red]}] {A''} {H''};
    \edge[draw=red,>={Latex[red]}] {D''} {I''};
    \edge[draw=red,>={Latex[red]}] {F''} {K''};
    \path [draw,->,>={Latex[black]}] (E'') edge [bend right] node [right] {} (C'');

    \edge [draw=violet,>={Latex[violet]}]{A''} {N''};\edge[draw=violet,>={Latex[violet]}] {B''} {O''};\edge [draw=violet]{A''} {O''};\edge[draw=violet,>={Latex[violet]}] {C''} {P''};\edge[draw=violet,>={Latex[violet]}] {D''} {Q''};\edge[draw=violet,>={Latex[violet]}] {J''} {Q''};\edge[draw=violet,>={Latex[violet]}] {E''} {R''};\edge[draw=violet,>={Latex[violet]}] {K''} {R''};\edge[draw=violet,>={Latex[violet]}] {F''} {S''};\edge[draw=violet,>={Latex[violet]}] {L''} {S''};
    \edge[draw=violet,>={Latex[violet]}] {D''} {P''};\edge[draw=violet,>={Latex[violet]}] {J''} {P''};\edge[draw=violet,>={Latex[violet]}] {F''} {R''};

    \path [draw,->,>={Latex[violet]}] (L'') edge [bend right, draw=violet] node [right] {} (R'');
    \path [draw,->,>={Latex[violet]}] (I'') edge [bend right, draw=violet] node [right] {} (P'');
    \path [draw,->,>={Latex[violet]}] (H'') edge [bend left, draw=violet] node [right] {} (O'');
    \path [draw,->,>={Latex[violet]}] (G'') edge [bend right, draw=violet] node [right] {} (N'');
    \path [draw,->,>={Latex[violet]}] (G'') edge [bend left, draw=violet] node [right] {} (O'');

     \path [draw,->,>={Latex[black]}] (G') edge node [right] {} (A'');
     \path [draw,-,>={Latex[black]},thick] (A') edge [bend left, draw] node [right] {} (A'');
     \path [draw,->,>={Latex[black]}] (H') edge  node [right] {} (A'');
     \path [draw,->,>={Latex[black]}] (B') edge [bend right] node [right] {} (A'');

\path [draw,->,>={Latex[black]}] (H') edge  node[right] {} (B'');
\path [draw,-,>={Latex[black]},thick] (B') edge [bend left, draw] node [right] {} (B'');
     
     \path [draw,->,>={Latex[black]}] (C') edge [bend left, draw] node [right] {} (D'');
     \path [draw,-,>={Latex[black]},thick] (D') edge [bend left, draw] node [right] {} (D'');
     \path [draw,->,>={Latex[black]}] (I') edge  node [right] {} (D'');
     \path [draw,->,>={Latex[black]}] (J') edge [bend left] node [right] {} (D'');

     \path [draw,->,>={Latex[black]}] (K') edge [bend left, draw] node [right] {} (E'');
     \path [draw,-,>={Latex[black]},thick] (E') edge [] node [right] {} (E'');

     \path [draw,->,>={Latex[black]}] (K') edge [bend left, draw] node [right] {} (F'');
     \path [draw,->,>={Latex[black]}] (E') edge [bend right, draw] node [right] {} (F'');
     \path [draw,->,>={Latex[black]}] (L') edge [ draw] node [right] {} (F'');
     \path [draw,-,>={Latex[black]},thick] (F') edge[bend right]  node [right] {} (F'');

\end{scope}

\end{tikzpicture}
\caption{{The folded MABN $\mathcal{G}_F$ obtained by making the edges $s_i(t)$ to $s_i(t_+)$ bidirectional ({\color{black}\textbf{black thick}} edges) for the MAS shown in Fig~\ref{fig:BNs}.}}
\label{fig:BNfolded}
\end{figure}
\end{center}

\begin{center}

\begin{figure}[htpb]

\begin{tikzpicture}
\begin{scope}[xshift= 12cm]
    \node[latent,draw=blue!50,xshift= 5cm] (A) {$1$};
    \node[latent,draw=blue!50, below=of A] (B) {$4$};
    \node[latent,draw=blue!50 ,right=of A] (C) {$2$};
    \node[latent,draw=blue!50, below=of C] (D) {$3$};
    \node[latent,draw=blue!50, right=of C] (E) {$6$};
    \node[latent,draw=blue!50 , below=of E] (G) {$5$};
    \node[above=of C] {$\mathcal{G}_{\text{VD}}$};
    \path (A) edge [draw= blue, >= {Latex[blue]},loop above] node (A'){}(A);
    \path (B) edge [draw= blue, >= {Latex[blue]},loop below] node (B'){}(B);
    \path (C) edge [draw= blue, >= {Latex[blue]},loop above] node (C'){}(C);
    \path (D) edge [draw= blue, >= {Latex[blue]},loop below] node (D'){}(D);
    \path (E) edge [draw= blue, >= {Latex[blue]},loop above] node (E'){}(E);
    \path (G) edge [draw= blue, >= {Latex[blue]},loop below] node (G'){}(G);
    
    \edge[draw= blue,>={Latex[blue]}] {A} {D};
    \path [draw= blue,->,>={Latex[blue]}] (B) edge [bend left] node [right] {} (D);
    \edge[draw= blue,>={Latex[blue]}] {C} {D};
    \edge[draw= blue,>={Latex[blue]}] {E} {D};
    \edge[draw= blue,>={Latex[blue]}] {G} {D};
    
    \edge[draw= blue,>={Latex[blue]}] {A} {B};
    \edge[draw= blue,>={Latex[blue]}] {C} {B};
    \edge[draw= blue,>={Latex[blue]}] {D} {B};
    \edge[draw= blue,>={Latex[blue]}] {E} {B};
    \path [draw= blue,->,>={Latex[blue]}] (G) edge [bend left] node [left] {} (B);
    
    \path [draw= blue,->,>={Latex[blue]}] (A) edge [bend left] node [right] {} (C);
    \path [draw= blue,->,>={Latex[blue]}] (C) edge [bend left] node [right] {} (A);
    
    \path [draw= blue,->,>={Latex[blue]}] (E) edge [bend left] node [right] {} (G);
    \path [draw= blue,->,>={Latex[blue]}] (G) edge [bend left] node [right] {} (E);
\end{scope}
\begin{scope}[yshift=-3cm]
\draw[dashed] (23,-1.5) -- (23,7);
\end{scope}

\begin{scope}[xshift=20cm]
    \node[latent,draw=blue!50,xshift= 5cm] (A) {$1$};
    \node[latent,draw=blue!50, below=of A] (B) {$4$};
    \node[latent,draw=blue!50 ,right=of A] (C) {$2$};
    \node[latent,draw=blue!50, below=of C] (D) {$3$};
    \node[latent,draw=blue!50, right=of C] (E) {$6$};
    \node[latent,draw=blue!50 , below=of E] (G) {$5$};
    \node[above=of C] {$\mathcal{G}_{\text{GD}}$};
    \path (A) edge [draw= orange, >= {Latex[orange]},loop above] node (A'){}(A);
    \path (B) edge [draw= orange, >= {Latex[orange]},loop below] node (B'){}(B);
    \path (C) edge [draw= orange, >= {Latex[orange]},loop above] node (C'){}(C);
    \path (D) edge [draw= orange, >= {Latex[orange]},loop below] node (D'){}(D);
    \path (E) edge [draw= orange, >= {Latex[orange]},loop above] node (E'){}(E);
    \path (G) edge [draw= orange, >= {Latex[orange]},loop below] node (G'){}(G);
    
    \edge[draw= orange,>={Latex[orange]}]{D}{A};
    \path [draw= orange,->,>={Latex[orange]}] (D) edge [bend left] node [left] {} (B);
    \edge[draw= orange,>={Latex[orange]}] {D} {C};
    \edge[draw= orange,>={Latex[orange]}] {D} {E};
    \edge[draw= orange,>={Latex[orange]}] {D} {G};
    
    \edge[draw= orange,>={Latex[orange]}] {B} {A};
    \edge[draw= orange,>={Latex[orange]}] {B} {C};
    \edge[draw= orange,>={Latex[orange]}] {B} {D};
    \edge[draw= orange,>={Latex[orange]}] {B} {E};
    \path [draw= orange,->,>={Latex[orange]}] (B) edge [bend left] node [left] {} (G);
    
    \path [draw= orange,->,>={Latex[orange]}] (A) edge [bend left] node [right] {} (C);
    \path [draw= orange,->,>={Latex[orange]}] (C) edge [bend left] node [right] {} (A);
    
    \path [draw= orange,->,>={Latex[orange]}] (E) edge [bend left] node [right] {} (G);
    \path [draw= orange,->,>={Latex[orange]}] (G) edge [bend left] node [right] {} (E);
\end{scope}

\end{tikzpicture}
\caption{$\mathcal{G}_{\text{VD}}$ (left) and $\mathcal{G}_{\text{GD}}$ (right) for the 6-agent example shown in Figs.~\ref{fig:BNs}-~\ref{fig:BNfolded}.}
\label{fig:Gvdgd}
\end{figure}
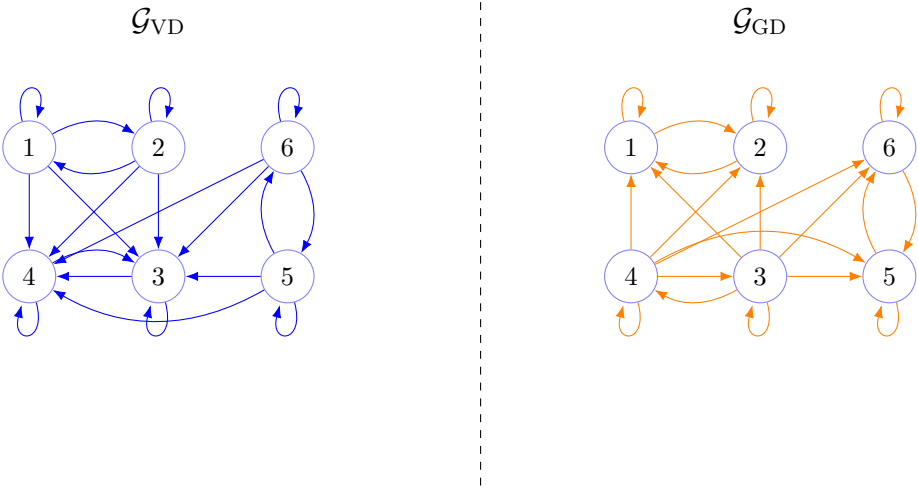
\end{center}

\newpage

 \section{Proof of Theorem~\ref{thm:Qsetdecomp}}
\label{sec:Qsetdecomp}
\begin{proof}
Consider the finite-horizon case of the {partially-observable stochastic cooperative game} described in Section~\ref{sec:coopmarl} with time-varying inter-agent couplings. For any agent $i\in \mathcal{V}$ at time $\tau$, the local reward $r_i$ is deterministic given $\{s_j(\tau),a_j(\tau)\}_{j\in \I^i_R(\tau)}.$ The state transition of agent $i$ at time $\tau$ depends only on $\{s_j(\tau\!-\!1),a_j(\tau\!-\!1)\}_{j\in \I^i_S(\tau)}$ and the action of agent $i\in \mathcal{V}$ depends only on $o_j(\tau) = \{s_j(\tau)\}_{j\in \I^i_O(\tau)}.$ Therefore, the innermost expectation in{~\eqref{eq:QZeq}} is fully defined over $(s_{U^T_i}(T),~a_{U^T_i}(T)),$ where $U^T_i =\{j \cup \I^j_O(T)\}_{j\in \I^i_R(T)}.$ Similarly, the expectation at time $T\!-\!1$ is taken over $(s_{U^{T\!-\!1}_i}(T\!-\!1),~a_{U^{T\!-\!1}_i}(T\!-\!1)),$ where $U^{T\!-\!1}_i = \{j \cup \I^j_O(T\!-\!1)\}_{j\in \I^i_R(T\!-\!1)} \cup \{\{k \cup \I^k_O(T\!-\!1)\}_{k \in\I^j_S(T)}\}_{j\in U^T_i}.$ Therefore, at any time $\tau,~t\leq \tau < T,$ the corresponding expectation in {~\eqref{eq:QZeq}} is taken over $(s_{U^{\tau}_i}(\tau),~a_{U^{\tau}_i}(\tau)),$ where $U^{\tau}_i = \{j \cup \I^j_O(\tau)\}_{j\in \I^i_R(\tau)} \cup \{\{k \cup \I^k_O(\tau)\}_{k \in\I^j_S(\tau\!+\!1)}\}_{j\in U^{\tau\!+\!1}_i}.$
Then,{~\eqref{eq:QZeq}} can be rewritten as
\begin{align}
&Q^\pi_i(s(t),a(t)) =  r_i(s_{\mathcal{I}^i_R}(t), a_{\mathcal{I}^i_R}(t)) + \nonumber\\
&\underset{\substack{s_{U^{t\!+\!1}_i}(t\!+\!1)\sim \mathcal{P}_{U^{t\!+\!1}_i}\\ a_{U^{t\!+\!1}_i}(t\!+\!1) \sim \pi_{\theta_{U^{t\!+\!1}_i}}}}{\mathbb{E}}\bigg[ \gamma r_i(s_{\mathcal{I}^i_R}(t\!+\!1), a_{\mathcal{I}^i_R}(t\!+\!1)) +\cdots +\underset{\substack{s_{U^{T\!-\!1}_i}(T\!-\!1)\sim \mathcal{P}_{U^{T-1}_i}\\ a_{U^{t+1}_i}(T\!-\!1) \sim \pi_{\theta_{U^{T\!-\!1}_i}}}}{\mathbb{E}}\bigg[ \gamma^{T-t-1} r_i(s_{\mathcal{I}^i_R}(T\!-\!1), a_{\mathcal{I}^i_R}(T\!-\!1)) \nonumber\\
&\hspace{250pt}+ \underset{\substack{s_{U^{T}_i}(T)\sim \mathcal{P}_{U^{T}_i}\\ a_{U^{T}_i}(T) \sim \pi_{\theta_{U^{T}_i}}}}{\mathbb{E}}\big[ \gamma^T r_i(s_{\mathcal{I}^i_R}(T), a_{\mathcal{I}^i_R}(T))\big] \bigg]\bigg].
\label{eq:QZ}
\end{align}
From \eqref{eq:QZ}, it is evident that the individual action value of agent $i$ at time $t$, $Q^\pi_i(s(t),a(t))$, depends only on $\I^i_Q(t) \triangleq U^t_i\cup U^{t+1}_i\cup \cdots \cup U^T_i \subseteq \mathcal{V}$ which is governed by the time-varying inter-agent couplings.
\end{proof}

\section{Proof of Theorem~\ref{thm:graddecomp}}
\label{sec:graddecomp}
\begin{proof}
Recall that
\begin{align}
&Q^\pi(s,a) = \mathbb{E}_\pi\left[ \sum_{i=1}^N \sum_{t=0}^\infty \gamma^t r_i(s_{\mathcal{I}^i_R}(t),a_{\mathcal{I}^i_R}(t)) | s(0) = s, a(0) = a\right]\nonumber\\
   &=\mathbb{E}_\pi\bigg[ \sum_{j \in \mathcal{I}_{\text{GD}}^i} \sum_{t=0}^\infty \gamma^t r_j(s_{\mathcal{I}^j_R}(t),a_{\mathcal{I}^j_R}(t)) + \sum_{j \backslash \mathcal{I}_{\text{GD}}^i} \sum_{t=0}^\infty \gamma^t r_j(s_{\mathcal{I}^j_R}(t),a_{\mathcal{I}^j_R}(t)) \bigg| s(0) = s, a(0) = a\bigg].
   \label{eq:rewdecomp}
   \end{align}
{   The first term in~\eqref{eq:rewdecomp} can be rewritten as 
   \begin{align}
   &\mathbb{E}_\pi\bigg[ \sum_{j \in \mathcal{I}_{\text{GD}}^i} \sum_{t=0}^\infty \gamma^t r_j(s_{\mathcal{I}^j_R}(t),a_{\mathcal{I}^j_R}(t)) | s(0) = s, a(0) = a\bigg]
   =\nonumber\\& \int_{s(1),r(1)} p(s(1),r(1)|s,a) \sum_{j \in \mathcal{I}_{\text{GD}}^i}\bigg[ r_j(s_{\mathcal{I}^j_R},a_{\mathcal{I}^j_R})+\nonumber\\&\hspace{90pt}\int_{\pi(a(1)|s(1))} \mathbb{E}_\pi\big[ \sum_{t=1}^\infty \gamma^t r_j(s_{\mathcal{I}^j_R}(t),a_{\mathcal{I}^j_R}(t)) |s(1), a(1)\big]da(1) \bigg]dr(1)ds(1)
   \end{align}
Assuming a deterministic global reward given global state and global action yields  
   \begin{align}
   &= \sum_{j \in \mathcal{I}_{\text{GD}}^i} \bigg[r_j(s_{\mathcal{I}^j_R},a_{\mathcal{I}^j_R})+\gamma\int_{s(1)} p(s(1)|s,a) \int_{\pi(a(1)|s(1))} Q^\pi_j(s(1),a(1))da(1) ds(1)\bigg]\nonumber\\
   &= \sum_{j \in \mathcal{I}_{\text{GD}}^i} \left[r_j(s_{\mathcal{I}^j_R},a_{\mathcal{I}^j_R})+\gamma\int_{s(1)} p(s(1)|s,a)  V^\pi_j(s(1))ds(1)\right].
   \label{eq:lavf}
   \end{align}
   Similarly, the second term in~\eqref{eq:rewdecomp} can be rewritten as
   
   \begin{align}
   &\mathbb{E}_\pi\bigg[ \sum_{j \backslash \mathcal{I}_{\text{GD}}^i} \sum_{t=0}^\infty \gamma^t r_j(s_{\mathcal{I}^j_R}(t),a_{\mathcal{I}^j_R}(t)) | s(0) = s, a(0) = a\bigg]=\nonumber\\
   &\hspace{160pt}\sum_{j \backslash \mathcal{I}_{\text{GD}}^i} \left[r_j(s_{\mathcal{I}^j_R},a_{\mathcal{I}^j_R})+\gamma\int_{s(1)} p(s(1)|s,a)  V^\pi_j(s(1))ds(1)\right].
   \label{eq:notlavf}
   \end{align}
   Using~\eqref{eq:lavf} and~\eqref{eq:notlavf},~\eqref{eq:rewdecomp} is rewritten as
   \begin{align}
   Q^\pi(s,a) &= \sum_{j \in \mathcal{I}_{\text{GD}}^i} [r_j(s_{\mathcal{I}_R^j}, a_{\mathcal{I}_R^j}) + \gamma  \int_{s(1) \in \mathcal{S}} p(s(1)|s, \pi_\theta(s)) V^{\pi_{\theta_j}}(s(1))]\nonumber \\&\hspace{50pt}+\sum_{k \notin \mathcal{I}_{\text{GD}}^i} [r_k(s_{\mathcal{I}_R^k}, a_{\mathcal{I}_R^k}) + \gamma  \int_{s' \in \mathcal{S}} p(s(1)|s, \pi_{\theta}(s)) V^{\pi_{\theta_k}}(s(1))].
   \label{eq:graddecomp}
\end{align}
}
{
Using the linearity of expectation,{~\eqref{eq:graddecomp} can be rewritten as}
\begin{align}
    &Q^\pi(s,a)=\sum_{j \in \mathcal{I}_{\text{GD}}^i}\mathbb{E}_\pi\bigg[  \sum_{t=0}^\infty \gamma^t r_j(s_{\mathcal{I}^j_R}(t),a_{\mathcal{I}^j_R}(t)) | s(0) = s, a(0) = a\bigg] \hspace{100pt} \nonumber\\&\hspace{60pt}+ \sum_{{k \not\in \mathcal{I}_{\text{GD}}^i}}
   \mathbb{E}_\pi\bigg[  \sum_{t=0}^\infty \gamma^t r_k(s_{\mathcal{I}^k_R}(t),a_{\mathcal{I}^k_R}(t)) | s(0) = s, a(0) = a\bigg]\\
   &= \sum_{j \in \mathcal{I}_{\text{GD}}^i} Q^\pi_j(s,a) + \sum_{k \not\in \mathcal{I}_{\text{GD}}^i} Q^\pi_k(s,a)   \label{eq:qdecomp}\\
 &= \widehat{Q}^\pi_i(s,a) + \bar{Q}^\pi_i(s,a),
   \label{eq:QbarQhat}
\end{align}
where $\bar{Q}^\pi_i(s,a)  = \sum_{k \not\in\mathcal{I}_{\text{GD}}^i} Q^\pi_k(s,a)$.
}
    From Theorem~\ref{thm:Qsetdecomp},
the reward of each agent $i \in \mathcal{V}$ depends on $s_j(t)$, $a_j(t)$ $\forall$ $j \in \I^i_Q$ and $\mathcal{E}_{\text{GD}} = \mathcal{E}^\intercal_{\text{VD}}$ by the definition of $\mathcal{G}_{\text{GD}}$. Therefore, if $k \notin \mathcal{I}^i_{\text{GD}}$, then $i \notin \I^k_Q$. Hence, $\sum_{k \not\in \mathcal{I}^i_{\text{GD}}} r_k(s_{\mathcal{I}_R^k}(t), a_{\mathcal{I}_R^k}(t))$ is independent of $a_i(t)$ and thus $\theta_i$. It then follows that $ Q^\pi_k(\cdot)$, $ V^\pi_k(\cdot)$ are independent of $\theta_i$, $\forall$ $k \notin \mathcal{I}^i_{\text{GD}}$, which implies
\begin{align}
    \nabla_{\theta_i} \  \bar{Q}^\pi_i &=  \nabla_{\theta_i} \sum_{k \not\in \mathcal{I}^i_{\text{GD}}} \bigg[r_k(s_{\mathcal{I}_R^k}, a_{\mathcal{I}_R^k})+ \gamma  \int_{s' \in \mathcal{S}} p(s'|s, \pi_\theta(s)) V^\pi_k(s')\bigg]\nonumber\\
     &=  \sum_{k \not\in \mathcal{I}^i_{\text{GD}}} \bigg[\nabla_{\theta_i} r_k(s_{\mathcal{I}_R^k}, a_{\mathcal{I}_R^k})+ \gamma  \int_{s' \in \mathcal{S}} p(s'|s, \pi_\theta(s)) \nabla_{\theta_i} [V^\pi_k(s')] \bigg] = 0,
  \label{eq:gradindependence}
\end{align}
 where the second line in~\eqref{eq:gradindependence} is obtained by interchanging the derivative and integral assuming that each $V^\pi_k(s')$ is sufficiently smooth in $s'$.
Then the gradient of the global action value function with respect to $\theta_i$ is given by
\begin{align}
   \nabla_{\theta_i}  Q^\pi(s,a) &=  
   \nabla_{\theta_i}  [\widehat{Q}^\pi_i + \bar{Q}^\pi_i ]= \nabla_{\theta_i} \widehat{Q}^\pi_i.
   \label{eq:gradlvf}
\end{align}
\end{proof}

\section{Proof of Theorem~\ref{thm:pgt}}
\label{sec:pgt}

\begin{proof}
    \begin{enumerate}[(a)]
      \item  
 We follow a similar approach as the single agent deterministic policy gradient theorem in~\cite{silver14} except that we leverage the topology of the network defining the multi-agent couplings. The regularity conditions required for the proof are stated as and when needed.
To compute the policy gradient for an agent $i \in \mathcal{V}$, consider the gradient of the state value function with respect to $\theta_i$: 
\begin{align}
    &\nabla_{\theta_i }V^{\pi_\theta}(s) = \nabla_{\theta_i} Q^{\pi_\theta}(s,\pi_\theta(s))
\end{align}
Using Theorem~\ref{thm:graddecomp}, the gradient of GVF w.r.t $\theta_i$ can be rewritten 
\begin{align}
&\nabla_{\theta_i }V^{\pi_\theta}(s) = \nabla_{\theta_i} \widehat{Q}^{\pi_\theta}(s,\pi_\theta(s))\\
&=\nabla_{\theta_i }\sum_{j \in \mathcal{I}_{\text{GD}}^i} \bigg[r_j(s_{\mathcal{I}_R^j}, \pi_{\theta_{\mathcal{I}_R^j}}(s_{\mathcal{I}^j_O}))+ \gamma  \int_{s' \in \mathcal{S}} p(s'|s, \pi_\theta(s)) V^{\pi_{\theta_j}}(s')ds'\bigg].
\end{align}
In the remainder of the proof, we suppress the notation $a_i = \pi_{\theta_i}(s_{\mathcal{I}^j_O})$ for brevity. Assuming $p(s'|s, \pi_\theta(s)),~\pi_\theta(s),~V^{\pi_{\theta_j}}(s')$, $j \in \mathcal{I}_{\text{GD}}^i$ and their derivatives are sufficiently smooth in $s,~a,~\theta$, interchanging the gradient and integral yields
\begin{align}
  &\nabla_{\theta_i }V^{\pi_\theta}(s)  =\nonumber\\
  &~~\nabla_{\theta_i} \pi_{\theta_i}(s_{\mathcal{I}^j_O}) \sum_{j \in \mathcal{I}^i_{\text{GD}}} \nabla_{a_i}r_j(s_{\mathcal{I}_R^j}, a_{\mathcal{I}_R^j})\bigg|_{a_i} + \gamma \sum_{j \in \mathcal{I}_{\text{GD}}^i} \int_{s' \in \mathcal{S}} \nabla_{\theta_i} [p(s'|s, \pi_\theta(s))  V^{\pi_{\theta_j}}(s')]ds'\nonumber\\
  &= \nabla_{\theta_i} \pi_{\theta_i}(s_{\mathcal{I}^j_O}) \sum_{j \in \mathcal{I}^i_{\text{GD}}} \nabla_{a_i}r_j(s_{\mathcal{I}_R^j}, a_{\mathcal{I}_R^j})\bigg|_{a_i} + \gamma \sum_{j \in \mathcal{I}_{\text{GD}}^i} \int_{s' \in \mathcal{S}} p(s'|s, \pi_\theta(s)) [ \nabla_{\theta_i} V^{\pi_{\theta_j}}(s')]ds'\nonumber\\
  &~~+ \gamma \sum_{j \in \mathcal{I}_{\text{GD}}^i} \int_{s' \in \mathcal{S}} \nabla_{\theta_i} \pi_{\theta_i}(s_{\mathcal{I}^j_O})\nabla_{a_i}p(s'|s, \pi_\theta(s))\bigg|_{a_i} V^{\pi_{\theta_j}}(s')ds' \nonumber\\
  &= \nabla_{\theta_i} \pi_{\theta_i}(s_{\mathcal{I}^j_O}) \sum_{j \in \mathcal{I}^i_{\text{GD}}} \nabla_{a_i}\bigg[r_j(s_{\mathcal{I}_R^j}, a_{\mathcal{I}_R^j})+ \gamma \int_{s' \in \mathcal{S}} p(s'|s, \pi_\theta(s)) V^{\pi_{\theta_j}}(s') ds'\bigg]_{a_i} \nonumber\\
  &~~+ \gamma \sum_{j \in \mathcal{I}_{\text{GD}}^i} \int_{s' \in \mathcal{S}} p(s'|s, \pi_\theta(s)) \nabla_{\theta_i} V^{\pi_{\theta_j}}(s')ds'\\
  &= \nabla_{\theta_i} \pi_{\theta_i}(s_{\mathcal{I}^j_O}) \nabla_{a_i} \widehat{Q}^\pi_i(s,a)\big|_{a_i}+ \gamma  \int_{s' \in \mathcal{S}} p(s'|s, \pi_\theta(s)) \nabla_{\theta_i} \sum_{j \in \mathcal{I}_{\text{GD}}^i}V^{\pi_{\theta_j}}(s')ds'.
\label{eq:recursivegrad}
\end{align}

Let $\rho^{\pi_\theta}(s \rightarrow x,k)$ denote the visitation probability of transitioning from $s$ to $x$ in $k$ steps under policy $\pi_\theta$. The transition from state $s$ to $x$ in $k$ steps can be broken down into transition from $s$ to an intermediate state $p$ in $k-1$ steps and transition from $p$ to $x$ in the last step. Hence, the visitation probability can be expressed recursively as
\begin{align}
    \rho^{\pi_\theta}(s \rightarrow x,k) = \int_p \rho^{\pi_\theta}(s \rightarrow p,k-1)\rho^{\pi_\theta}(p \rightarrow x,1) dp.
    \label{eq:visitprob1}
\end{align}
For brevity, define $\phi(s) = \nabla_{\theta_i} \pi_{\theta_i}(s_{\mathcal{I}^j_O}) \nabla_{a_i} \widehat{Q}^\pi_i(s,a)\big|_{a_i}$. Using \eqref{eq:visitprob1}, the gradient in \eqref{eq:recursivegrad} is expressed in a recursive form as
\begin{align}
    \nabla_{\theta_i} V^{\pi_\theta}(s) &= \phi(s)+ \gamma  \int_{s' \in \mathcal{S}} \rho^{\pi_\theta}(s \rightarrow s',1) \nabla_{\theta_i} V^{\pi_{\theta}}(s')ds'\nonumber\\
    &= \phi(s) + \gamma  \int_{s' \in \mathcal{S}} \rho^{\pi_\theta}(s \rightarrow s',1) \phi(s') ds'\nonumber\\&\hspace{40pt}+ \gamma^2  \int_{s' \in \mathcal{S}} \rho^{\pi_\theta}(s \rightarrow s',1)\int_{s'' \in \mathcal{S}} \rho^{\pi_\theta}(s' \rightarrow s'',1) \nabla_{\theta_i} V^{\pi_{\theta}}(s'')ds'' ds'.
    \end{align}
    For the sake of analysis, we assume that the state space $\mathcal{S}$ is compact. Hence, we assume that $|\nabla_{\theta_i} V^{\pi_{\theta}}(s)|$ is finite and apply Fubini's theorem to exchange the order of integration which yields  
    \begin{align}
    \nabla_{\theta_i} V^{\pi_\theta}(s)&= \phi(s)\nonumber+ \gamma  \int_{s' \in \mathcal{S}} \rho^{\pi_\theta}(s \rightarrow s',1) \phi(s')\big|_{a'_i} ds'\nonumber\nonumber\\&\hspace{40pt}+ \gamma^2 \int_{s'' \in \mathcal{S}} \int_{s' \in \mathcal{S}} \rho^{\pi_\theta}(s \rightarrow s',1)  \rho^{\pi_\theta}(s' \rightarrow s'',1) \nabla_{\theta_i} V^{\pi_{\theta}}(s'') ds' ds''\\
    &= \phi(s)+ \gamma  \int_{s' \in \mathcal{S}} \rho^{\pi_\theta}(s \rightarrow s',1) \phi(s') \gamma^2  \int_{s'' \in \mathcal{S}} \rho^{\pi_\theta}(s \rightarrow s'',2)\phi(s'') + \cdots\\
    &= \int_{s'\in S} \sum_{k=0}^\infty \gamma^t \rho^{\pi_\theta}(s \rightarrow s',k) \phi(s') ds'.
    \label{eq:pgtheorem1}
    \end{align}
Hence, the gradient of the objective can be rewritten as
\begin{align}
    \nabla_{\theta_i} J(\theta) &= \nabla_{\theta_i }\int_{s' \in \mathcal{S}}p(s)V^{\pi_\theta}(s') ds'   
\end{align}
Assuming $V^{\pi_\theta}(s')$, $p(s)$ and their derivatives with respect to $\theta_i$ are sufficiently smooth, applying Leibniz rule to exchange the derivative and integral and applying Fubini's theorem to exchange the order of integration assuming $|\nabla_{\theta_i} V^{\pi_{\theta}}(s)|$ is finite yields
\begin{align}
\nabla_{\theta_i} J(\theta) &= \int_{s \in \mathcal{S}}p(s)\int_{s'\in S} \sum_{k=0}^\infty \gamma^k \rho^{\pi_\theta}(s \rightarrow s',k) \phi(s') ds' ds\\
&=  \int_{s'\in \mathcal{S}} \int_{s \in \mathcal{S}} \gamma^k p(s) \sum_{k=0}^\infty \rho^{\pi_\theta}(s \rightarrow s',k)\phi(s') ds ds'.
\end{align}
Denote $\eta(s') = \sum_{k=0}^\infty \rho^{\pi_\theta}(s \rightarrow s',k)$
\begin{align}
    \nabla_\theta J(\theta) &= \int_{s'\in \mathcal{S}}\left(\int_{s'\in \mathcal{S}} \eta(s') \right) \int_{s\in \mathcal{S}}p(s)\frac{\eta(s')}{ \int_{s'\in \mathcal{S}} \eta(s')} \phi(s')ds ds'\\
    &\propto \int_{s'\in \mathcal{S}} \int_{s\in \mathcal{S}}p(s) \frac{\eta(s')}{ \int_{s'\in \mathcal{S}} \eta(s')} \phi(s')ds'\\
    &= \int_{s'\in \mathcal{S}} d^\pi(s') \nabla_{\theta_i} \pi_{\theta_i}(s'_{\mathcal{I}^i_O}) \nabla_{a'_i} \widehat{Q}_i(s',a')\big|_{a'_i = \pi_{\theta_i}(s'_{\mathcal{I}^i_O})} ds'
    \label{eq:mapgt1}
\end{align}

The last two equations follow from the fact that $\int_{s'\in S} \eta(s')$ is a constant (equal to the length of episode in episodic case and 1 in continuous case) and $d^\pi(s') = \int_{s \in \mathcal{S}} p(s)\frac{\eta(s')}{ \int_{s'\in S} \eta(s')}$ is a stationary distribution. Define $\I^i_{\widehat{Q}} = \bigcup_{j\in \I^i_{\text{GD}}} \I^j_Q$. Using Theorem~\ref{thm:Qsetdecomp}, \eqref{eq:mapgt1} can be rewritten as
\begin{align}
   \nabla_{\theta_i} J(\theta) &=\int_{s'_{\I^i_{\widehat{Q}}}\in \mathcal{S}} d^\pi(s'_{\I^i_{\widehat{Q}}}) \nabla_{\theta_i} \pi_{\theta_i}(s'_{\mathcal{I}^i_O}) \nabla_{a'_i} \widehat{Q}^\pi_i(s'_{\I^i_{\widehat{Q}}},a'_{\I^i_{\widehat{Q}}})\big|_{a'_i = \pi_{\theta_i}(s'_{\mathcal{I}^i_O})} ds'\nonumber\\
   &=\underset{s'_{\I^i_{\widehat{Q}}}\sim d^\pi(s'_{\I^i_{\widehat{Q}}})}{\mathbb{E}} \left[ \nabla_{\theta_i} \pi_{\theta_i}(s'_{\mathcal{I}^i_O}) \nabla_{a'_i} \widehat{Q}^\pi_i(s'_{\I^i_{\widehat{Q}}},a'_{\I^i_{\widehat{Q}}})\big|_{a'_i = \pi_{\theta_i}(s'_{\mathcal{I}^i_O})}\right].
\end{align}
Replacing $s'$ by $s$ and $a'$ by $a$ gives the result in Theorem~\ref{thm:pgt} a.
\item 
    
 We follow a similar approach as the single agent stochastic policy gradient theorem in~\cite{sutton2018reinforcement} except that we leverage the topology of the network defining the multi-agent couplings. The regularity conditions required for the proof are stated as and when needed.
To compute the policy gradient for an agent $i \in \mathcal{V}$, consider the gradient of the state value function with respect to $\theta_i$: 
\begin{align}
    &\nabla_{\theta_i }V^{\pi_\theta}(s) = \nabla_{\theta_i} \int_{a\in \A} \pi_{\theta}(a|s) Q^{\pi_\theta}(s,a) da\\
    &=\nabla_{\theta_i} \int_{a\in \A} \pi_{\theta}(a|s) \left(\sum_{j \in \mathcal{I}_{\text{GD}}^i} Q^\pi_j(s,a) + \sum_{k \backslash \mathcal{I}_{\text{GD}}^i} Q^\pi_k(s,a)\right) da\\
    &=\nabla_{\theta_i} \int_{a\in \A} \pi_{\theta}(a|s) \sum_{j \in \mathcal{I}_{\text{GD}}^i} Q^\pi_j(s,a) da + \nabla_{\theta_i} \int_{a\in \A} \pi_{\theta}(a|s) \sum_{k \backslash \mathcal{I}_{\text{GD}}^i} Q^\pi_k(s,a) da
\end{align}
Assume that $\pi_\theta(a|s)$, $Q^{\pi_\theta}(s,a)$ are sufficiently smooth and the $\pi_\theta(a|s) = \prod_{i=1}^N \pi_{\theta_i}(a_i|s_{\I^i_O}) .$ Then, interchanging the derivative and integral and using Theorem~\ref{thm:graddecomp} yields 
\begin{align}
&\nabla_{\theta_i }V^{\pi_\theta}(s) = \nonumber\\
&\sum_{j \in \mathcal{I}_{\text{GD}}^i} \int_{a\in \A} \left(\nabla_{\theta_i}\pi_{\theta_i}(a_i|s_{\I^i_O})\left( \prod_{j\setminus i} \pi_{\theta_j}(a_j|s_{\I^j_O}) \right) Q^\pi_j(s,a) + \pi_{\theta}(a|s) \nabla_{\theta_i}Q^\pi_j(s,a) \right)da.
\label{eq:V}
\end{align}


Recall that
\begin{align}
Q^\pi_j(s,a) &= r_j(s_{\I^j_R}, a_{\I^j_R}) + \gamma  \int_{s' \in \mathcal{S}} p(s'|s, a) V^{\pi_{\theta_j}}(s')ds'.
\end{align}

Assuming a deterministic reward and a sufficiently smooth $V^{\pi_{\theta_j}}(s)$ and its derivative with respect to $\theta_i$ $\forall~j$ yields 
\begin{align}
\int_{a\in \A}\nabla_{\theta_i}Q^\pi_{j}(s,a) &=\int_{a\in \A} \gamma  \int_{s' \in \mathcal{S}} p(s'|s, a)\nabla_{\theta_i } V^{\pi_{\theta_j}}(s')ds'.
\label{eq:Qdecomp}
\end{align}

Substituting \eqref{eq:Qdecomp} in \eqref{eq:V} yields
\begin{align}
\nabla_{\theta_i }V^{\pi_\theta}(s) &=  \sum_{j \in \mathcal{I}_{\text{GD}}^i} \int_{a\in \A} \nabla_{\theta_i}\pi_{\theta_i}(a_i|s_{\I^i_O})\left( \prod_{j\setminus i} \pi_{\theta_j}(a_j|s_{\I^j_O}) \right) Q^\pi_j(s,a) da\nonumber\\&\hspace{50pt} + \gamma \sum_{j \in \mathcal{I}_{\text{GD}}^i} \int_{a\in \A}  \pi_{\theta}(a|s) \int_{s' \in \mathcal{S}} p(s'|s, a)\nabla_{\theta_i } V^{\pi_{\theta_j}}(s')ds' da.
\label{eq:Vdecomp}
\end{align}

Let $\rho^{\pi_\theta}(s \rightarrow x,k)$ denote the visitation probability of transitioning from $s$ to $x$ in $k$ steps under policy $\pi_\theta$. The transition from state $s$ to $x$ in $k$ steps can be broken down into transition from $s$ to an intermediate state $p$ in $k-1$ steps and transition from $p$ to $x$ in the last step. Hence, the visitation probability can be expressed recursively as
\begin{align}
    \rho^{\pi_\theta}(s \rightarrow x,k) = \int_p \rho^{\pi_\theta}(s \rightarrow p,k-1)\rho^{\pi_\theta}(p \rightarrow x,1) dp.
    \label{eq:visitprob}
\end{align}
Define $\phi(s) = \sum_{j \in \mathcal{I}_{\text{GD}}^i}\int_{a\in \A} \nabla_{\theta_i}\pi_{\theta_i}(a_i|s_{\I^i_O})\left( \prod_{j\setminus i} \pi_{\theta_j}(a_j|s_{\I^j_O}) \right) Q^\pi_{j}(s,a) da$. For the sake of analysis, we assume that the state space $\mathcal{S}$ is compact which implies $|\nabla_{\theta_i} V^{\pi_{\theta}}(s)|$ is finite. Therefore, in the remainder of the proof, we apply Fubini's theorem to exchange the order of integration when necessary. Then, using  \eqref{eq:visitprob}, the gradient in \eqref{eq:Vdecomp} is expressed in a recursive form as
\begin{align}
    \nabla_{\theta_i} V^{\pi_\theta}(s) &= \phi(s)+ \gamma  \sum_{j \in \mathcal{I}_{\text{GD}}^i} \int_{a\in \A}  \pi_{\theta}(a|s) \int_{s' \in \mathcal{S}} p(s'|s, a)\nabla_{\theta_i } V^{\pi_{\theta_j}}(s')ds' da\nonumber \\
    &= \phi(s)+ \gamma  \sum_{j \in \mathcal{I}_{\text{GD}}^i}  \int_{s' \in \mathcal{S}} \int_{a\in \A}  \pi_{\theta}(a|s) p(s'|s, a)\nabla_{\theta_i } V^{\pi_{\theta_j}}(s') da ds'\nonumber \\
    &= \phi(s)+ \gamma  \int_{s' \in \mathcal{S}} \rho^{\pi_\theta}(s \rightarrow s',1) \nabla_{\theta_i} V^{\pi_{\theta}}(s')ds'\nonumber\\&(\text{Since, $\nabla_{\theta_i} V^{\pi_{\theta}}(s) = \sum_{j \in \mathcal{I}_{\text{GD}}^i}\nabla_{\theta_i} V^\pi_{j}(s)$ (Lemma 2, \cite{jing2024distributed}})\nonumber\\
    &= \phi(s) + \gamma  \int_{s' \in \mathcal{S}} \rho^{\pi_\theta}(s \rightarrow s',1) \phi(s') ds'\nonumber\\&\hspace{40pt}+ \gamma^2  \int_{s' \in \mathcal{S}} \rho^{\pi_\theta}(s \rightarrow s',1)\int_{s'' \in \mathcal{S}} \rho^{\pi_\theta}(s' \rightarrow s'',1) \nabla_{\theta_i} V^{\pi_{\theta}}(s'')ds'' ds'.
    \end{align}
    \begin{align}
    \nabla_{\theta_i} V^\pi(s)&= \phi(s)+ \gamma  \int_{s' \in \mathcal{S}} \rho^{\pi_\theta}(s \rightarrow s',1) \phi(s')ds'\nonumber\\&\hspace{40pt}+ \gamma^2 \int_{s'' \in \mathcal{S}} \int_{s' \in \mathcal{S}} \rho^{\pi_\theta}(s \rightarrow s',1)  \rho^{\pi_\theta}(s' \rightarrow s'',1) \nabla_{\theta_i} V^{\pi_{\theta}}(s'') ds' ds''\\
    &= \phi(s)+ \gamma  \int_{s' \in \mathcal{S}} \rho^{\pi_\theta}(s \rightarrow s',1) \phi(s') +\gamma^2  \int_{s'' \in \mathcal{S}} \rho^{\pi_\theta}(s \rightarrow s'',2)\phi(s'') + \cdots\nonumber\\
    &= \int_{s'\in S} \sum_{k=0}^\infty \gamma^t \rho^{\pi_\theta}(s \rightarrow s',k) \phi(s') ds'
    \label{eq:pgtheorem}
    \end{align}
Hence, the gradient of the objective can be rewritten as
\begin{align}
    \nabla_{\theta_i} J(\theta) &= \nabla_{\theta_i }\int_{s' \in \mathcal{S}}p(s)V^{\pi_\theta}(s') ds'  \\ 
&= \int_{s \in \mathcal{S}}p(s)\int_{s'\in S} \sum_{k=0}^\infty \gamma^k \rho^{\pi_\theta}(s \rightarrow s',k) \phi(s') ds' ds\\
&=  \int_{s'\in \mathcal{S}} \int_{s \in \mathcal{S}}  p(s) \sum_{k=0}^\infty \gamma^k \rho^{\pi_\theta}(s \rightarrow s',k)\phi(s') ds ds'
\end{align}
Denote $\eta(s') = \sum_{k=0}^\infty \gamma^k \rho^{\pi_\theta}(s \rightarrow s',k)$
\begin{align}
    \nabla_{\theta_i} J(\theta) &= \int_{s'\in S}\left(\int_{s'\in S} \eta(s') \right) \int_{s\in \mathcal{S}}p(s)\frac{\eta(s')}{ \int_{s'\in S} \eta(s')} \phi(s')ds ds'\\
    &\propto \int_{s'\in S} \int_{s\in \mathcal{S}}p(s) \frac{\eta(s')}{ \int_{s'\in S} \eta(s')} \phi(s')ds'\\
    &= \int_{s'\in S} d^\pi(s') \int_{a'\in \A} \nabla_{\theta_i}\pi_{\theta_i}(a'_i|s'_{\I^i_O})\left( \prod_{j\setminus i} \pi_{\theta_j}(a'_j|s'_{\I^j_O}) \right) \left(\sum_{j \in \mathcal{I}_{\text{GD}}^i} Q^\pi_{j}(s',a')\right) da' ds'\nonumber\\
    &=\int_{s'\in S} d^\pi(s') \int_{a'\in \A} \nabla_{\theta_i}\pi_{\theta_i}(a'_i|s'_{\I^i_O})\left( \prod_{j\setminus i} \pi_{\theta_j}(a'_j|s'_{\I^j_O}) \right) \widehat{Q}^\pi_{i}(s',a') da' ds'
    \label{eq:mapgt}
\end{align}

The last two equations follow from the fact that $\int_{s'\in S} \eta(s')$ is a constant (equal to the length of episode in episodic case and 1 in continuous case) and $d^\pi(s') = \int_{s \in \mathcal{S}} p(s)\frac{\eta(s')}{ \int_{s'\in S} \eta(s')}$ is a stationary distribution. Define $\I^i_{\widehat{Q}} = \bigcup_{j\in \I^i_{\text{GD}}} \I^j_Q$. From the definition of $Q-$set, $Q^j(s,a)$ only depends on $s_k, a_k$ such that $k \in \I^j_Q$. Hence, \eqref{eq:mapgt} can be rewritten as
\begin{align}
   \nabla_{\theta_i} J(\theta) &=\int_{s'\in S} d^\pi(s') \int_{a'\in \A} \nabla_{\theta_i}\pi_{\theta_i}(a'_i|s'_{\I^i_O})\left( \prod_{j \in \I^i_{\widehat{Q}}\setminus i} \pi_{\theta_j}(a'_j|s_{\I^j_O}) \right) \widehat{Q}^\pi_{i}(s'_{\I^i_{\widehat{Q}}},a'_{\I^i_{\widehat{Q}}}) da' ds'\nonumber\\
   &=\int_{s'\in S} d^\pi(s') \int_{a'\in \A} \dfrac{\nabla_{\theta_i}\pi_{\theta_i}(a'_i|s'_{\I^i_O})}{\pi_{\theta_i}(a'_i|s'_{\I^i_O})} \underbrace{\left(\prod_{j \in \I^i_{\widehat{Q}}} \pi_{\theta_j}(a'_j|s'_{\I^j_O})\right)}_{\hat{\pi}^i}  \widehat{Q}^\pi_{i}(s'_{\I^i_{\widehat{Q}}},a'_{\I^i_{\widehat{Q}}}) da' ds'\nonumber\\
   &=\underset{s'\sim d^\pi(s'),a' \sim \hat{\pi}^i}{\mathbb{E}} \left[ \widehat{Q}^\pi_{i}(s'_{\I^i_{\widehat{Q}}},a'_{\I^i_{\widehat{Q}}} )\nabla_{\theta_i}\ln\pi_{\theta_i}(a'_i|s'_{\I^i_O}) \right].
\end{align}

Replacing $s'$ by $s$ and $a'$ by $a$ gives the result in Theorem~\ref{thm:pgt} b.
\end{enumerate}

\end{proof}

 \section{Schematic of the structured actor critic (MAStAC) algorithm}\label{sec:alg}
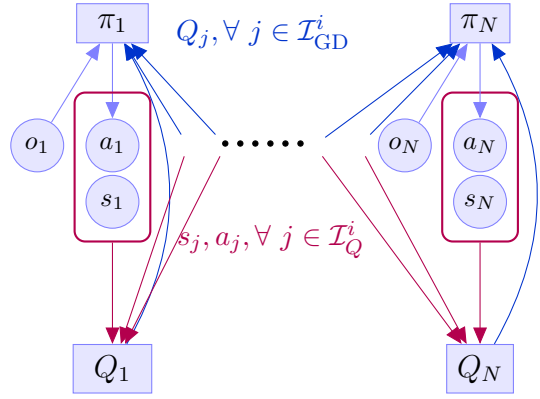
\begin{figure}[htpb]
\begin{center}
\begin{tikzpicture}[scale=1, transform shape]
    \node[rectangle,draw=blue!50, fill=blue!10,xshift = 7cm] (p1) {$~\pi_1~$};
    \node[latent,draw=blue!50,fill=blue!10, below=of p1, xshift= -1cm] (o1) {$o_1$};
    \node[latent,draw=blue!50 ,fill=blue!10,below=of p1,xshift=0cm] (a1) {$a_1$};
    \node[latent,draw=blue!50,fill=blue!10, below=of p1,xshift=0cm, yshift=-0.75cm] (s1) {$s_1$};
    \node[rectangle,draw=blue!50, fill=blue!10,below= of p1, yshift= -3 cm] (q1) {$~Q_1~$};
    \node[rectangle,draw=blue!50, fill=blue!10,right= of p1,xshift=3cm] (pN) {$\pi_N$};
    \node[latent,draw=blue!50,fill=blue!10, below=of pN, xshift= -1cm] (oN) {$o_N$};
    \node[latent,draw=blue!50 ,fill=blue!10,below=of pN,xshift=0cm] (aN) {$a_N$};
    \node[latent,draw=blue!50,fill=blue!10, below=of pN,xshift=0cm, yshift=-0.75cm] (sN) {$s_N$};
    \node[rectangle,draw=blue!50, fill=blue!10,below= of pN, yshift= -3 cm] (qN) {$Q_N$};
    \node[mark size=1pt,color=black] at (8.5,-1.6) {\pgfuseplotmark{*}};
    \node[mark size=1pt,color=black] at (8.7,-1.6) {\pgfuseplotmark{*}};
    \node[mark size=1pt,color=black] at (8.9,-1.6) {\pgfuseplotmark{*}};
    \node[mark size=1pt,color=black] at (9.1,-1.6) {\pgfuseplotmark{*}};
    \node[mark size=1pt,color=black] at (9.3,-1.6) {\pgfuseplotmark{*}};
    \node[mark size=1pt,color=black] at (9.5,-1.6) {\pgfuseplotmark{*}};
    \node[mark size=1pt,color=black] at (2.5,-1.6) {\pgfuseplotmark{}};
    \node[right=of a1] (empty) {$ $};
    \node[right=of a1,xshift= -0.5cm] (empty1) {$ $};
    \node[above=of empty,xshift= 0.5cm] (name) {\color{blue!80!green}{\small$ Q_j, \forall~j \in \I^i_{\text{GD}}$}};
    \node[below=of name,xshift= 0.1cm,yshift= -1cm] (name1) {\color{red!70!blue}{\small$ s_j,a_j,\forall~j \in \I^i_{Q}$}};
    \edge[draw= blue!50] {o1} {p1};
    \edge [draw= blue!50] {p1} {a1};
    \node[left=of aN,xshift= -0.1cm] (empty2) {$ $};
    \node[left=of aN,xshift= -0.7cm] (empty3) {$ $};
    \plate[draw=red!70!blue,  thick, yshift=0.2cm] {plate1} { %
    (a1) %
    (s1) %
  } {} 
  \plate[draw=red!70!blue,  thick, yshift=0.2cm] {plate2} { %
    (aN) %
    (sN) %
  } {} 
  \edge [draw= red!70!blue]{plate1} {q1};
  \edge [draw= red!70!blue]{plate2} {qN};
    \path [draw,->] (q1) edge [bend right,draw= blue!80!green] node [right] {} (p1);
    \edge[draw= blue!50] {oN} {pN};
    \edge [draw= blue!50] {pN} {aN};
    \edge [draw= blue!80!green]{empty} {p1};
    \edge [draw= blue!80!green]{empty1} {p1};
    \edge [draw= blue!80!green]{empty2} {pN};
    \edge [draw= blue!80!green]{empty3} {pN};
    \edge [draw= red!70!blue]{empty} {q1};
    \edge [draw= red!70!blue]{empty1} {q1};
    \edge [draw= red!70!blue]{empty2} {qN};
    \edge [draw= red!70!blue]{empty3} {qN};
    \path [draw,->] (qN) edge [bend right,draw= blue!80!green] node [right] {} (pN);
\end{tikzpicture}
\end{center}
\caption{The schematic of the multi-agent structured actor critic (MAStAC) algorithm.}
\label{fig:NN}
\end{figure}
\newpage
\section{Proof of Theorem~\ref{thm:vardiff}}
\label{sec:vardiff}

\begin{proof}
    Let 
    \begin{align*}
        \mathbf{g}^i_C &= (Q^\pi(s,a)-\dq)\nabla_{\theta_i} \ln{\pi_{\theta_i}(a_i|s_{\I^i_O})}\\
        \mathbf{g}^i_Q &= (\widehat{Q}^\pi_i(\hat{s}_i,\hat{a}_i) - \dqh)\nabla_{\theta_i} \ln{\pi_{\theta_i}(a_i|s_{\I^i_O})}
    \end{align*}
     be the centralized and the decomposed Q-estimators over the distribution $s\sim d^\pi(s),a  \sim {\pi_\theta}$, respectively. Let $\gpi = \nabla_{\theta_i} \ln{\pi_{\theta_i}(a_i|s_{\I^i_O})}.$ Then, 
     $$\gc = (Q^\pi(s,a)-\dq)\gpi,~~\gq= (\widehat{Q}^\pi_i(\hat{s}_i,\hat{a}_i) - \dqh)\gpi.$$
 Using the law of total variance, we express the difference in the variance of the $\gc$ and $\gq$ as
     \begin{align}
         &\var_{s\sim d^\pi(s),a  \sim {\pi}, \dq \sim p(\mu_Q, \sigma^2_Q)} [\gc] - \var_{s\sim d^\pi(s),a  \sim {\pi},\dqh \sim p(\mu_{\widehat{Q}}, \sigma^2_{\widehat{Q}})} [\gq] = \nonumber\\&\hspace{40pt}\left[\var_{s\sim d^\pi(s), \dq \sim p(\mu_Q, \sigma^2_Q)}\left[\mathbb{E}_{a  \sim {\pi}} [\gc]\right] + \mathbb{E}_{s\sim d^\pi(s), \dq \sim p(\mu_Q, \sigma^2_Q)}\left[\var_{a  \sim {\pi}} [\gc]\right]\right] -\nonumber\\&\hspace{40pt}
         \left[\var_{s\sim d^\pi(s),\dqh \sim p(\mu_{\widehat{Q}}, \sigma^2_{\widehat{Q}})}\left[\mathbb{E}_{a  \sim {\pi}} [\gq]\right] + \mathbb{E}_{s\sim d^\pi(s),\dqh \sim p(\mu_{\widehat{Q}}, \sigma^2_{\widehat{Q}})}\left[\var_{a  \sim {\pi}} [\gq]\right]\right].
         \nonumber\\&
         =\var_{s\sim d^\pi(s)} \left[\mathbb{E}_{\dq \sim p(\mu_{{Q}}, \sigma^2_{{Q}})}\left[\mathbb{E}_{a  \sim {\pi}} [\gc]\right]\right] + \mathbb{E}_{s\sim d^\pi(s)} \left[\var_{\dq \sim p(\mu_{{Q}}, \sigma^2_{{Q}})}\left[\mathbb{E}_{a  \sim {\pi}} [\gc]\right]\right]\nonumber\\ &\quad
         + \mathbb{E}_{s\sim d^\pi(s),\dq \sim p(\mu_{{Q}}, \sigma^2_{{Q}})}\left[\var_{a  \sim {\pi}} [\gc]\right] -\var_{s\sim d^\pi(s)} \left[\mathbb{E}_{\dqh \sim p(\mu_{\widehat{Q}}, \sigma^2_{\widehat{Q}})}\left[\mathbb{E}_{a  \sim {\pi}} [\gq]\right]\right]\nonumber\\ &\quad - \mathbb{E}_{s\sim d^\pi(s)} \left[\var_{\dqh \sim p(\mu_{\widehat{Q}}, \sigma^2_{\widehat{Q}})}\left[\mathbb{E}_{a  \sim {\pi}} [\gq]\right]\right]
         - \mathbb{E}_{s\sim d^\pi(s),\dqh \sim p(\mu_{\widehat{Q}}, \sigma^2_{\widehat{Q}})}\left[\var_{a  \sim {\pi}} [\gq]\right].
         \label{eq:gc-gq}
     \end{align}

Assuming that $S$ is compact and $|\gc|,~|\gq|$ are finite, we employ Fubini's theorem to exchange the order of integration as and when needed. Consider the first and fourth terms of~\eqref{eq:gc-gq}:
{
\begin{align}
    &\var_{s\sim d^\pi(s)} \left[\mathbb{E}_{\dq \sim p(\mu_{{Q}}, \sigma^2_{{Q}})}\left[\mathbb{E}_{a  \sim {\pi}} [\gc]\right]\right] - \var_{s\sim d^\pi(s)} \left[\mathbb{E}_{\dqh \sim p(\mu_{\widehat{Q}}, \sigma^2_{\widehat{Q}})}\left[\mathbb{E}_{a  \sim {\pi}} [\gq]\right]\right]\nonumber\\
    =&\var_{s\sim d^\pi(s)} \left[\mathbb{E}_{a  \sim {\pi}}\left[\mathbb{E}_{\dq \sim p(\mu_{{Q}}, \sigma^2_{{Q}})} [\gc]\right]\right] - \var_{s\sim d^\pi(s)} \left[\mathbb{E}_{a  \sim {\pi}}\left[\mathbb{E}_{\dqh \sim p(\mu_{\widehat{Q}}, \sigma^2_{\widehat{Q}})}[\gq]\right]\right]\nonumber\\
    =&\var_{s\sim d^\pi(s)} \bigg[\mathbb{E}_{a  \sim {\pi}}\bigg[\mathbb{E}_{\dq \sim p(\mu_{{Q}}, \sigma^2_{{Q}})} [(Q^\pi(s,a)-\dq)]\gpi\bigg]\bigg]- \nonumber\\
    &\hspace{120pt}\var_{s\sim d^\pi(s)} \bigg[\mathbb{E}_{a  \sim {\pi}}\bigg[ \mathbb{E}_{\dqh \sim p(\mu_{\widehat{Q}}, \sigma^2_{\widehat{Q}})}[(\widehat{Q}^\pi_i(\hat{s}_i,\hat{a}_i) - \dqh)]\gpi\bigg]\bigg]\nonumber\\
    =&\var_{s\sim d^\pi(s)} \bigg[\mathbb{E}_{a  \sim {\pi}}\bigg[(Q^\pi(s,a)-\mu_Q)\gpi\bigg]\bigg]- \var_{s\sim d^\pi(s)} \bigg[\mathbb{E}_{a  \sim {\pi}}\bigg[ (\widehat{Q}^\pi_i(\hat{s}_i,\hat{a}_i) - \mu_{\widehat{Q}})]\gpi\bigg]\bigg]\nonumber\\
\label{eq:gcgq1}
\end{align}
We note that
\begin{align}
     &\var_{s\sim d^\pi(s)} \bigg[\mathbb{E}_{a  \sim {\pi}}\bigg[(Q^\pi(s,a)-\mu_Q)\gpi\bigg]\bigg]= \nonumber\\
     &\text{Tr}\bigg(\mathbb{E}_{s\sim d^\pi(s)}\bigg[\mathbb{E}_{a  \sim {\pi}}[(Q^\pi(s,a)-\mu_Q)\gpi]\mathbb{E}_{a  \sim {\pi}}[(Q^\pi(s,a)-\mu_Q)(\gpi)^\intercal] \bigg] \nonumber\\
     &\hspace{70pt}- \mathbb{E}_{s\sim d^\pi(s)}\mathbb{E}_{a  \sim {\pi}}[(Q^\pi(s,a)-\mu_Q)\gpi]\mathbb{E}_{s\sim d^\pi(s)}\mathbb{E}_{a  \sim {\pi}}[(Q^\pi(s,a)-\mu_Q)(\gpi)^\intercal] \bigg)\nonumber\\
     &=\text{Tr}\bigg(\mathbb{E}_{s\sim d^\pi(s)}\bigg[\mathbb{E}_{a  \sim {\pi}}[Q^\pi(s,a)\gpi]\mathbb{E}_{a  \sim {\pi}}[Q^\pi(s,a)(\gpi)^\intercal] + \mu_Q^2 \mathbb{E}_{a  \sim {\pi}}[\gpi]\mathbb{E}_{a  \sim {\pi}}[(\gpi)^\intercal] \nonumber\\
     &\hspace{70pt}- \mu_Q \mathbb{E}_{a  \sim {\pi}}[Q^\pi(s,a)\gpi]\mathbb{E}_{a  \sim {\pi}}[(\gpi)^\intercal] - \mu_Q \mathbb{E}_{a  \sim {\pi}}[\gpi]\mathbb{E}_{a  \sim {\pi}}[Q^\pi(s,a)(\gpi)^\intercal]\bigg]\nonumber\\
     &\hspace{40pt}- \mathbb{E}_{s\sim d^\pi(s)}\mathbb{E}_{a  \sim {\pi}}[(Q^\pi(s,a)-\mu_Q)\gpi]\mathbb{E}_{s\sim d^\pi(s)}\mathbb{E}_{a  \sim {\pi}}[(Q^\pi(s,a)-\mu_Q)(\gpi)^\intercal] \bigg)\nonumber\\
    &=\text{Tr}\bigg(\mathbb{E}_{s\sim d^\pi(s)}\bigg[\mathbb{E}_{a  \sim {\pi}}[(Q^\pi(s,a))^2\gpi(\gpi)^\intercal] - \Var_{a\sim \pi}(Q^\pi(s,a)\gpi) + \mu_Q^2 \mathbb{E}_{a  \sim {\pi}}[\gpi(\gpi)^\intercal] - \mu_Q^2 \Var_{a\sim \pi}(\gpi)\nonumber\\
    &\hspace{70pt}+ \mu_Q \textbf{Cov}(Q^\pi(s,a)\gpi,(\gpi)^\intercal) - \mu_Q \mathbb{E}_{a  \sim {\pi}}[Q^\pi(s,a)\gpi (\gpi)^\intercal] + \mu_Q \textbf{Cov}(\gpi,Q^\pi(s,a)(\gpi)^\intercal) \nonumber\\
     &\hspace{40pt}- \mu_Q \mathbb{E}_{a  \sim {\pi}}[Q^\pi(s,a)\gpi (\gpi)^\intercal]\bigg]- \mathbb{E}_{s\sim d^\pi(s)}\mathbb{E}_{a  \sim {\pi}}[(Q^\pi(s,a)-\mu_Q)\gpi]\mathbb{E}_{s\sim d^\pi(s)}\mathbb{E}_{a  \sim {\pi}}[(Q^\pi(s,a)-\mu_Q)(\gpi)^\intercal] \bigg)\nonumber\\
     &=\text{Tr}\bigg(\mathbb{E}_{s\sim d^\pi(s)}\bigg[- \Var_{a\sim \pi}(Q^\pi(s,a)\gpi) - \mu_Q^2 \Var_{a\sim \pi}(\gpi) + \mu_Q \textbf{Cov}_{a\sim \pi}(Q^\pi(s,a)\gpi,(\gpi)^\intercal)\nonumber\\
     &\hspace{20pt}+ \mu_Q \textbf{Cov}_{a\sim \pi}(\gpi,Q^\pi(s,a)(\gpi)^\intercal)\bigg] + \mu_Q^2\mathbb{E}_{s\sim d^\pi(s), a  \sim {\pi}}[\gpi(\gpi)^\intercal]-\mu_Q^2\mathbb{E}_{s\sim d^\pi(s), a  \sim {\pi}}[\gpi]\mathbb{E}_{s\sim d^\pi(s), a  \sim {\pi}}[(\gpi)^\intercal]\nonumber\\
     &+ \mathbb{E}_{s\sim d^\pi(s), a  \sim {\pi}}[(Q^\pi(s,a))^2\gpi(\gpi)^\intercal] - \mathbb{E}_{s\sim d^\pi(s), a  \sim {\pi}}[Q^\pi(s,a)\gpi]\mathbb{E}_{s\sim d^\pi(s), a  \sim {\pi}}[Q^\pi(s,a)(\gpi)^\intercal]\nonumber\\
     &- \mu_Q\mathbb{E}_{s\sim d^\pi(s), a  \sim {\pi}}[Q^\pi(s,a)\gpi (\gpi)^\intercal] + \mu_Q \mathbb{E}_{s\sim d^\pi(s), a  \sim {\pi}}[Q^\pi(s,a)\gpi] \mathbb{E}_{s\sim d^\pi(s), a  \sim {\pi}}[(\gpi)^\intercal]\nonumber\\
     &- \mu_Q\mathbb{E}_{s\sim d^\pi(s), a  \sim {\pi}}[Q^\pi(s,a)\gpi (\gpi)^\intercal] + \mu_Q \mathbb{E}_{s\sim d^\pi(s), a  \sim {\pi}}[\gpi] \mathbb{E}_{s\sim d^\pi(s), a  \sim {\pi}}[Q^\pi(s,a)(\gpi)^\intercal]\bigg)\nonumber\\
      &=\text{Tr}\bigg(\mathbb{E}_{s\sim d^\pi(s)}\bigg[- \Var_{a\sim \pi}(Q^\pi(s,a)\gpi) - \mu_Q^2 \Var_{a\sim \pi}(\gpi) + \mu_Q \textbf{Cov}_{a\sim \pi}(Q^\pi(s,a)\gpi,(\gpi)^\intercal)\nonumber\\
     &\hspace{20pt}+ \mu_Q \textbf{Cov}_{a\sim \pi}(\gpi,Q^\pi(s,a)(\gpi)^\intercal)\bigg] + \Var_{s\sim d^\pi(s),a\sim \pi}(Q^\pi(s,a)\gpi) + \mu_Q^2 \Var_{s\sim d^\pi(s), a\sim \pi}(\gpi) \nonumber\\
     &\hspace{20pt}- \mu_Q \textbf{Cov}_{s\sim d^\pi(s),a\sim \pi}(Q^\pi(s,a)\gpi,(\gpi)^\intercal)- \mu_Q \textbf{Cov}_{s\sim d^\pi(s),a\sim \pi}(\gpi,Q^\pi(s,a)(\gpi)^\intercal)\bigg).
     \label{eq:gcgq1simpa}
\end{align}
Similarly,
\begin{align}
    &\var_{s\sim d^\pi(s)} \bigg[\mathbb{E}_{a  \sim {\pi}}\bigg[(\widehat{Q}^\pi_i(\hat{s}_i,\hat{a}_i)-\mu_{\widehat{Q}})\gpi\bigg]\bigg]= \nonumber\\
    &\text{Tr}\bigg(\mathbb{E}_{s\sim d^\pi(s)}\bigg[- \Var_{a\sim \pi}(\widehat{Q}^\pi_i(\hat{s}_i,\hat{a}_i)\gpi) - \mu_{\widehat{Q}}^2 \Var_{a\sim \pi}(\gpi) + \mu_{\widehat{Q}} \textbf{Cov}_{a\sim \pi}(\widehat{Q}^\pi_i(\hat{s}_i,\hat{a}_i)\gpi,(\gpi)^\intercal)\nonumber\\
     &\hspace{20pt}+ \mu_{\widehat{Q}} \textbf{Cov}_{a\sim \pi}(\gpi,\widehat{Q}^\pi_i(\hat{s}_i,\hat{a}_i)(\gpi)^\intercal)\bigg] + \Var_{s\sim d^\pi(s),a\sim \pi}(\widehat{Q}^\pi_i(\hat{s}_i,\hat{a}_i)\gpi) + \mu_{\widehat{Q}}^2 \Var_{s\sim d^\pi(s), a\sim \pi}(\gpi) \nonumber\\
     &\hspace{20pt}- \mu_{\widehat{Q}} \textbf{Cov}_{s\sim d^\pi(s),a\sim \pi}(\widehat{Q}^\pi_i(\hat{s}_i,\hat{a}_i)\gpi,(\gpi)^\intercal)- \mu_{\widehat{Q}} \textbf{Cov}_{s\sim d^\pi(s),a\sim \pi}(\gpi,\widehat{Q}^\pi_i(\hat{s}_i,\hat{a}_i)(\gpi)^\intercal)\bigg).
     \label{eq:gcgq1simpb}
\end{align}
}
    Consider
    \begin{align}
        &\var_{\dq \sim p(\mu_{{Q}}, \sigma^2_{{Q}})}\left[\mathbb{E}_{a  \sim {\pi}} [\gc]\right] =\nonumber\\&
        \tr{\mathbb{E}_{\dq \sim p(\mu_{{Q}}, \sigma^2_{{Q}})}[\mathbb{E}_{a  \sim {\pi}} [\gc](\mathbb{E}_{a  \sim {\pi}} [\gc])^\intercal] - \mathbb{E}_{\dq \sim p(\mu_{{Q}}, \sigma^2_{{Q}}), a  \sim {\pi}} [\gc](\mathbb{E}_{\dq \sim p(\mu_{{Q}}, \sigma^2_{{Q}}),a  \sim {\pi}} [\gc])^\intercal}\nonumber\\
        &=\text{Tr}\bigg(\mathbb{E}_{\dq \sim p(\mu_{{Q}}, \sigma^2_{{Q}})}\bigg[(\mathbb{E}_{a  \sim {\pi}}[Q^\pi(s,a)\gpi])(\mathbb{E}_{a  \sim {\pi}}[Q^\pi(s,a)\gpi])^\intercal + \delta^2_Q (\mathbb{E}_{a  \sim {\pi}}[\gpi])(\mathbb{E}_{a  \sim {\pi}}[\gpi])^\intercal\nonumber\\&\hspace{20pt} - 2\delta_Q (\mathbb{E}_{a  \sim {\pi}}[Q^\pi(s,a)\gpi])(\mathbb{E}_{a  \sim {\pi}}[\gpi])^\intercal\bigg] -(\mathbb{E}_{a  \sim {\pi}}[Q^\pi(s,a)\gpi])(\mathbb{E}_{a  \sim {\pi}}[Q^\pi(s,a)\gpi])^\intercal\nonumber\\
        &~~- \left(\mu^2_Q (\mathbb{E}_{a  \sim {\pi}}[\gpi])(\mathbb{E}_{a  \sim {\pi}}[\gpi])^\intercal- 2\mu_Q (\mathbb{E}_{a  \sim {\pi}}[Q^\pi(s,a)\gpi])\right)(\mathbb{E}_{a  \sim {\pi}}[\gpi])^\intercal\bigg)\nonumber\\
        &= \tr{\sigma^2_Q(\mathbb{E}_{a  \sim {\pi}} [\gpi])(\mathbb{E}_{a  \sim {\pi}} [\gpi])^\intercal}~{(\because \sigma^2_Q = \mathbb{E}_{\dq \sim p(\mu_{{Q}}, \sigma^2_{{Q}})}[\delta^2_Q] - \mu^2_Q)}.
        \label{eq:dvargc}
    \end{align}
Similarly, it is straightforward to show that 
\begin{align}
    \var_{\dqh \sim p(\mu_{\widehat{Q}}, \sigma^2_{\widehat{Q}})}\left[\mathbb{E}_{a  \sim {\pi}} [\gq]\right] &= \tr{\sigma^2_{\widehat{Q}}(\mathbb{E}_{a  \sim {\pi}} \gpi)(\mathbb{E}_{a  \sim {\pi}} \gpi)^\intercal}.
    \label{eq:dvargq}
\end{align}
Then, using \eqref{eq:dvargc} and~\eqref{eq:dvargq}, the second term and the fifth term in \eqref{eq:gc-gq} can be expressed as
\begin{align}
    \mathbb{E}_{s\sim d^\pi(s)} \left[\var_{\dq \sim p(\mu_{{Q}}, \sigma^2_{{Q}})}\left[\mathbb{E}_{a  \sim {\pi}} [\gc]\right] - \var_{\dqh \sim p(\mu_{\widehat{Q}}, \sigma^2_{\widehat{Q}})}\left[\mathbb{E}_{a  \sim {\pi}} [\gq]\right]\right]\qquad&\nonumber\\
    =(\sigma^2_{{Q}} -\sigma^2_{\widehat{Q}} )\mathbb{E}_{s\sim d^\pi(s)} \left[\tr{(\mathbb{E}_{a  \sim {\pi}} \gpi)(\mathbb{E}_{a  \sim {\pi}} \gpi)^\intercal }\right]&.
    \label{eq:gcgqsimp2}
\end{align}

For the third term in~\eqref{eq:gc-gq}, we note 
     \begin{align}
        &\mathbb{E}_{\dq \sim p(\mu_{{Q}}, \sigma^2_{{Q}})}\left[\var_{a  \sim {\pi}} [\gc]\right]=\mathbb{E}_{\dq \sim p(\mu_{{Q}}, \sigma^2_{{Q}})}\left[\tr{\mathbb{E}_{a  \sim {\pi}}\left[ \gc (\gc)^\intercal\right] - \mathbb{E}_{a  \sim {\pi}}\left[ \gc \right]\mathbb{E}\left[\gc\right]^\intercal}\right]\nonumber\\
        &= \tr{\mathbb{E}_{\dq \sim p(\mu_{{Q}}, \sigma^2_{{Q}})}\left[\mathbb{E}_{a  \sim {\pi}}\left[ \gc (\gc)^\intercal\right] - \mathbb{E}_{a  \sim {\pi}}\left[ \gc \right]\mathbb{E}\left[\gc\right]^\intercal\right]}\nonumber\\
        &= \text{Tr}\bigg(\mathbb{E}_{\dq \sim p(\mu_{{Q}}, \sigma^2_{{Q}})}\bigg[\mathbb{E}_{a  \sim {\pi}}\left[(Q^\pi(s,a) -\dq)^2\left(\gpi\right)\left(\gpi\right)^\intercal\right] \nonumber\\&\hspace{120pt}- \mathbb{E}_{a  \sim {\pi}}\left[ (Q^\pi(s,a) -\dq) \gpi\right]\mathbb{E}_{a  \sim {\pi}}\left[(Q^\pi(s,a) -\dq) \gpi\right]^\intercal\bigg]\bigg)\nonumber\\
        &= \text{Tr}\bigg(\mathbb{E}_{a  \sim {\pi}}\left[(Q^\pi(s,a))^2\left(\gpi\right)\left(\gpi\right)^\intercal\right] -2\mu_Q \mathbb{E}_{a  \sim {\pi}}\left[Q^\pi(s,a)\left(\gpi\right)\left(\gpi \right)^\intercal\right] \nonumber\\
        &\qquad+ (\sigma^2_Q +\mu^2_Q)\left(\mathbb{E}_{a  \sim {\pi}}\left[\left(\gpi \right)\left(\gpi \right)^\intercal\right] -\mathbb{E}_{a  \sim {\pi}}\left[\gpi \right]\mathbb{E}_{a  \sim {\pi}}\left[\gpi \right]^\intercal\right)\nonumber\\
        &\qquad- \mathbb{E}_{a  \sim {\pi}}\left[Q^\pi(s,a)\gpi \right]\mathbb{E}_{a  \sim {\pi}}\left[Q^\pi(s,a)\gpi \right]^\intercal
        + \mu_Q \bigg[\mathbb{E}_{a  \sim {\pi}}\left[\gpi \right]\mathbb{E}_{a  \sim {\pi}}\left[Q^\pi(s,a)\gpi \right]^\intercal \bigg] \nonumber\\
        &\hspace{50pt}+ \mu_Q \bigg[ \mathbb{E}_{a  \sim {\pi}}\left[Q^\pi(s,a)\gpi\right]\mathbb{E}_{a  \sim {\pi}}\left[\gpi\right]^\intercal\bigg]\bigg)\nonumber\\
        &= \text{Tr}\bigg(\Var_{a  \sim {\pi}}\left[Q^\pi(s,a)\gpi\right]+ (\sigma^2_Q +\mu^2_Q)\Var_{a  \sim {\pi}}\left[\gpi\right]\nonumber\\
        &\hspace{50pt} -\mu_Q \textbf{Cov}_{a  \sim {\pi}}\left[Q^\pi(s,a)\gpi,\left(\gpi\right)^\intercal\right] -\mu_Q \textbf{Cov}_{a  \sim {\pi}}\left[\gpi,\left(Q^\pi(s,a)\gpi\right)^\intercal\right] \bigg).
        \label{eq:avarc}
     \end{align}
Similarly, for the sixth term in~\eqref{eq:gc-gq}, it follows that
\begin{align}
    &\mathbb{E}_{\dqh \sim p(\mu_{\widehat{Q}}, \sigma^2_{\widehat{Q}})}\left[\var_{a  \sim {\pi}} [\gq]\right]=\nonumber\\
    &= \text{Tr}\bigg(\Var_{a  \sim {\pi}}\left[\widehat{Q}^\pi(s,a)\gpi\right]+ (\sigma^2_{\widehat{Q}} +\mu^2_{\widehat{Q}})\Var_{a  \sim {\pi}}\left[\gpi\right]\nonumber\\
        &\hspace{50pt} -\mu_{\widehat{Q}} \textbf{Cov}_{a  \sim {\pi}}\left[\widehat{Q}^\pi(s,a)\gpi,\left(\gpi\right)^\intercal\right] -\mu_{\widehat{Q}} \textbf{Cov}_{a  \sim {\pi}}\left[\gpi,\left(\widehat{Q}^\pi(s,a)\gpi\right)^\intercal\right] \bigg).
        \label{eq:avarq}
\end{align}
Using \eqref{eq:avarc} and~\eqref{eq:avarq}, we obtain
\begin{align}
    &\mathbb{E}_{\dq \sim p(\mu_{{Q}}, \sigma^2_{{Q}})}\left[\var_{a  \sim {\pi}} [\gc]\right] - \mathbb{E}_{\dqh \sim p(\mu_{\widehat{Q}}, \sigma^2_{\widehat{Q}})}\left[\var_{a  \sim {\pi}} [\gq]\right]=\nonumber\\
    &\text{Tr}\bigg(\Var_{a  \sim {\pi}}\left[Q^\pi(s,a)\gpi\right] - \Var_{a  \sim {\pi}}\left[\widehat{Q}(s,a)\gpi\right]\bigg)+\text{Tr}\bigg((\sigma^2_Q  -\sigma^2_{\widehat{Q}} + \mu^2_Q -\mu^2_{\widehat{Q}})\Var_{a  \sim {\pi}}\left[\gpi\right]\bigg)\nonumber\\
    &-\text{Tr}\bigg(\mu_Q \textbf{Cov}_{a  \sim {\pi}}\left[Q^\pi(s,a)\gpi,\left(\gpi\right)^\intercal\right]- \mu_{\widehat{Q}}\textbf{Cov}_{a  \sim {\pi}}\left[\widehat{Q}^\pi(s,a)\gpi,\left(\gpi\right)^\intercal\right]\bigg)\nonumber\\
        & -\text{Tr}\bigg(\mu_Q \textbf{Cov}_{a  \sim {\pi}}\left[\gpi,\left(Q^\pi(s,a)\gpi\right)^\intercal\right]- \mu_{\widehat{Q}} \textbf{Cov}_{a  \sim {\pi}}\left[\gpi,\left(\widehat{Q}^\pi(s,a)\gpi\right)^\intercal\right]\bigg).
        \label{eq:gcgqsimp3}
\end{align}
{
Substituting~\eqref{eq:gcgq1simpa},~\eqref{eq:gcgq1simpb},~\eqref{eq:gcgqsimp2}, and~\eqref{eq:gcgqsimp3} into~\eqref{eq:gc-gq} yields
\begin{align}
         &\var_{s\sim d^\pi(s),a  \sim {\pi}, \dq \sim p(\mu_Q, \sigma^2_Q)} [\gc] - \var_{s\sim d^\pi(s),a  \sim {\pi},\dqh \sim p(\mu_{\widehat{Q}}, \sigma^2_{\widehat{Q}})} [\gq] = \nonumber\\&
         \text{Tr}\bigg(\mathbb{E}_{s\sim d^\pi(s)}\bigg[ \cancel{\Var_{a\sim \pi}(\widehat{Q}^\pi_i(\hat{s}_i,\hat{a}_i)\gpi)}- \cancel{\Var_{a\sim \pi}(Q^\pi(s,a)\gpi)} - \cancel{\mu_Q^2 \Var_{a\sim \pi}(\gpi)} + \cancel{\mu_{\widehat{Q}}^2 \Var_{a\sim \pi}(\gpi)}\nonumber\\
         &+ \cancel{\mu_Q \textbf{Cov}_{a\sim \pi}(Q^\pi(s,a)\gpi,(\gpi)^\intercal)} - \cancel{\mu_{\widehat{Q}} \textbf{Cov}_{a\sim \pi}(\widehat{Q}^\pi_i(\hat{s}_i,\hat{a}_i)\gpi,(\gpi)^\intercal)}+ \cancel{\mu_Q \textbf{Cov}_{a\sim \pi}(\gpi,Q^\pi(s,a)(\gpi)^\intercal)} \nonumber\\
     &-  \cancel{\mu_{\widehat{Q}} \textbf{Cov}_{a\sim \pi}(\gpi,\widehat{Q}^\pi_i(\hat{s}_i,\hat{a}_i)(\gpi)^\intercal)}\bigg] + \Var_{s\sim d^\pi(s),a\sim \pi}(Q^\pi(s,a)\gpi)  - \Var_{s\sim d^\pi(s),a\sim \pi}(\widehat{Q}^\pi_i(\hat{s}_i,\hat{a}_i)\gpi)\nonumber\\
     &+ \mu_Q^2 \Var_{s\sim d^\pi(s), a\sim \pi}(\gpi) - \mu_{\widehat{Q}}^2 \Var_{s\sim d^\pi(s), a\sim \pi}(\gpi) \nonumber\\
     &- \mu_Q \textbf{Cov}_{s\sim d^\pi(s),a\sim \pi}(Q^\pi(s,a)\gpi,(\gpi)^\intercal)+ \mu_{\widehat{Q}} \textbf{Cov}_{s\sim d^\pi(s),a\sim \pi}(\widehat{Q}^\pi_i(\hat{s}_i,\hat{a}_i)\gpi,(\gpi)^\intercal)\nonumber\\
     &- \mu_Q \textbf{Cov}_{s\sim d^\pi(s),a\sim \pi}(\gpi,Q^\pi(s,a)(\gpi)^\intercal) + \mu_{\widehat{Q}} \textbf{Cov}_{s\sim d^\pi(s),a\sim \pi}(\gpi,\widehat{Q}^\pi_i(\hat{s}_i,\hat{a}_i)(\gpi)^\intercal)\bigg)\bigg)\nonumber\\
         &+(\sigma^2_{{Q}} -\sigma^2_{\widehat{Q}} )\mathbb{E}_{s\sim d^\pi(s)} \left[\tr{(\mathbb{E}_{a  \sim {\pi}} \gpi)(\mathbb{E}_{a  \sim {\pi}} \gpi)^\intercal }\right] + \nonumber\\
         &\mathbb{E}_{s\sim d^\pi(s)} \bigg[\text{Tr}\bigg(\cancel{\Var_{a  \sim {\pi}}\left[Q(s,a)\gpi\right] }- \cancel{\Var_{a  \sim {\pi}}\left[\widehat{Q}(s,a)\gpi\right]}\bigg)\nonumber\\
         &+\text{Tr}\bigg((\sigma^2_Q  -\sigma^2_{\widehat{Q}} + \cancel{\mu^2_Q }-\cancel{\mu^2_{\widehat{Q}}})\Var_{a  \sim {\pi}}\left[\gpi\right]\bigg)\nonumber\\
    &-\text{Tr}\bigg(\cancel{\mu_Q \textbf{Cov}_{a  \sim {\pi}}\left[Q(s,a)\gpi,\left(\gpi\right)^\intercal\right]}- \cancel{\mu_{\widehat{Q}}\textbf{Cov}_{a  \sim {\pi}}\left[\widehat{Q}(s,a)\gpi,\left(\gpi\right)^\intercal\right]}\bigg)\nonumber\\
        & -\text{Tr}\bigg(\cancel{\mu_Q \textbf{Cov}_{a  \sim {\pi}}\left[\gpi,\left(Q(s,a)\gpi\right)^\intercal\right]}- \cancel{\mu_{\widehat{Q}} \textbf{Cov}_{a  \sim {\pi}}\left[\gpi,\left(\widehat{Q}(s,a)\gpi\right)^\intercal\right]}\bigg)\bigg]\nonumber\\
\end{align}
\begin{align}
        &=(\sigma^2_{{Q}} -\sigma^2_{\widehat{Q}} )\mathbb{E}_{s\sim d^\pi(s)} \left[\tr{(\mathbb{E}_{a  \sim {\pi}} \gpi)(\mathbb{E}_{a  \sim {\pi}} \gpi)^\intercal }\right] +\text{Tr}\bigg((\sigma^2_Q  -\sigma^2_{\widehat{Q}})\Var_{a  \sim {\pi}}\left[\gpi\right]\bigg)+\nonumber\\
        &\text{Tr}\bigg(\Var_{s\sim d^\pi(s),a\sim \pi}(Q^\pi(s,a)\gpi)  - \Var_{s\sim d^\pi(s),a\sim \pi}(\widehat{Q}^\pi_i(\hat{s}_i,\hat{a}_i)\gpi)\nonumber\\
     &+ \mu_Q^2 \Var_{s\sim d^\pi(s), a\sim \pi}(\gpi) - \mu_{\widehat{Q}}^2 \Var_{s\sim d^\pi(s), a\sim \pi}(\gpi) \nonumber\\
     &- \mu_Q \textbf{Cov}_{s\sim d^\pi(s),a\sim \pi}(Q^\pi(s,a)\gpi,(\gpi)^\intercal)+ \mu_{\widehat{Q}} \textbf{Cov}_{s\sim d^\pi(s),a\sim \pi}(\widehat{Q}^\pi_i(\hat{s}_i,\hat{a}_i)\gpi,(\gpi)^\intercal)\nonumber\\
     &- \mu_Q \textbf{Cov}_{s\sim d^\pi(s),a\sim \pi}(\gpi,Q^\pi(s,a)(\gpi)^\intercal) + \mu_{\widehat{Q}} \textbf{Cov}_{s\sim d^\pi(s),a\sim \pi}(\gpi,\widehat{Q}^\pi_i(\hat{s}_i,\hat{a}_i)(\gpi)^\intercal)\bigg)\bigg).
         \label{eq:gc-gqsimp}
     \end{align}
     }
 \noindent
If $\mu_Q = \mu_{\widehat{Q}} = 0$, then
\begin{align}
    &\var_{s\sim d^\pi(s),a  \sim {\pi}, \dq \sim p(\mu_Q, \sigma^2_Q)} [\gc] - \var_{s\sim d^\pi(s),a  \sim {\pi},\dqh \sim p(\mu_{\widehat{Q}}, \sigma^2_{\widehat{Q}})} [\gq] = \nonumber\\
    & (\sigma^2_{{Q}} -\sigma^2_{\widehat{Q}} )\mathbb{E}_{s\sim d^\pi(s)} \left[\tr{(\mathbb{E}_{a  \sim {\pi}} \gpi)(\mathbb{E}_{a  \sim {\pi}} \gpi)^\intercal }\right] + \nonumber\\
         &\text{Tr}\bigg(\Var_{{s\sim d^\pi(s),}a  \sim {\pi}}\left[Q^\pi(s,a)\gpi\right] - \Var_{{s\sim d^\pi(s),}a  \sim {\pi}}\left[\widehat{Q}^\pi(s,a)\gpi\right]\bigg)+\mathbb{E}_{s\sim d^\pi(s)} \bigg[\text{Tr}\bigg((\sigma^2_Q  -\sigma^2_{\widehat{Q}})\Var_{a  \sim {\pi}}\left[\gpi\right]\bigg)\bigg]\nonumber\\
         &=\text{Tr}\bigg(\Var_{{s\sim d^\pi(s),}a  \sim {\pi}}\left[Q^\pi(s,a)\gpi\right] - \Var_{{s\sim d^\pi(s),}a  \sim {\pi}}\left[\widehat{Q}^\pi(s,a)\gpi\right]\bigg) \nonumber\\
         &\hspace{100pt}+(\sigma^2_Q  -\sigma^2_{\widehat{Q}})\text{Tr}\bigg(\mathbb{E}_{s\sim d^\pi(s)} \bigg[\mathbb{E}_{a  \sim {\pi}}\left[\left(\gpi \right)\left(\gpi \right)^\intercal\right]\bigg]\bigg).
\end{align}
{
We note that
 \begin{align*}
     &\mathbb{E}_{s\sim d^\pi(s)}\bigg[\mathbb{E}_{a  \sim {\pi}}\bigg[\left[Q^\pi(s,a)- \widehat{Q}^\pi_i(\hat{s}_i,\hat{a}_i)\right]\gpi\bigg] \bigg]=\nonumber\\ &\int_{s\sim d^\pi(s)} d^\pi(s)\int_{a_{\mathcal{V}\setminus i}} \prod_{j \in \mathcal{V}\setminus i} \pi_{\theta_j}(a_j|s_{\I^j_O})\left[\bar{Q}^\pi(s,a)\right]\int_{a_i \sim \pi_{\theta_i}}  \pi_{\theta_i}(a_i|s_{\I^i_O})\nabla_{\theta_i} \ln{\pi_{\theta_i}(a_i|s_{\I^i_O})} da ds\\&\hspace{50pt}\text{ [$\because$, $\bar{Q}(\cdot)$ $\perp$ $a_i$ from Theorem~\ref{thm:graddecomp}]}\\
     &= \int_{s\sim d^\pi(s)} d^\pi(s)\int_{a_{\mathcal{V}\setminus i}} \prod_{j \in \mathcal{V}\setminus i} \pi_{\theta_j}(a_j|s_{\I^j_O})\left[\bar{Q}^\pi(s,a)\right]\int_{a_i \sim \pi_{\theta_i}}  \nabla_{\theta_i} \pi_{\theta_i}(a_i|s_{\I^i_O}) da ds.
\end{align*}
Assuming that $\pi_{\theta_i}(\cdot)$ is sufficiently smooth, we interchange the derivative and integral to obtain
\begin{align*}
&\hspace{-100pt}= \int_{s\sim d^\pi(s)} d^\pi(s) \int_{a_{\mathcal{V}\setminus i}} \prod_{j \in \mathcal{V}\setminus i} \pi_{\theta_j}(a_j|s_{\I^j_O})\left[\bar{Q}^\pi(s,a)\right]\nabla_{\theta_i}[1] da_{\mathcal{V}\setminus i}ds = 0.
 \end{align*}
 Therefore, 
 \begin{align}
 \mathbb{E}_{s\sim d^\pi(s)}\bigg[\mathbb{E}_{a  \sim {\pi}}\bigg[Q^\pi(s,a)\gpi\bigg] \bigg] = \mathbb{E}_{s\sim d^\pi(s)}\bigg[\mathbb{E}_{a  \sim {\pi}}\bigg[ \widehat{Q}^\pi_i(\hat{s}_i,\hat{a}_i)\gpi\bigg] \bigg],
 \label{eq:QeqQhat}
 \end{align} which means
\begin{align}
&\var_{s\sim d^\pi(s),a  \sim {\pi}, \dq \sim p(\mu_Q, \sigma^2_Q)} [\gc] - \var_{s\sim d^\pi(s),a  \sim {\pi},\dqh \sim p(\mu_{\widehat{Q}}, \sigma^2_{\widehat{Q}})} [\gq] = \nonumber\\
       &=\mathbb{E}_{s\sim d^\pi(s)}\bigg[(\sigma^2_Q  -\sigma^2_{\widehat{Q}})\text{Tr}\bigg(\mathbb{E}_{a  \sim {\pi}}\left[\left(\gpi \right)\left(\gpi \right)^\intercal\right]\bigg)\bigg] + \tr{\mathbb{E}_{s\sim d^\pi(s)} \mathbb{E}_{a  \sim {\pi}}\left[ [(Q^\pi(s,a))^2 \gpi (\gpi)^\intercal\right]} \nonumber\\
&- \tr{\mathbb{E}_{s\sim d^\pi(s)}\mathbb{E}_{a  \sim {\pi}}\left[ Q^\pi(s,a) \gpi\right]\mathbb{E}_{s\sim d^\pi(s)}\mathbb{E}_{a  \sim {\pi}}\left[Q^\pi(s,a)\gpi\right]^\intercal} -\tr{\mathbb{E}_{s\sim d^\pi(s)} \mathbb{E}_{a  \sim {\pi}}\left[ [(\widehat{Q}^\pi_i(\hat{s}_i,\hat{a}_i))^2 \gpi (\gpi)^\intercal\right]} \nonumber\\
&+ \tr{\mathbb{E}_{s\sim d^\pi(s)}\mathbb{E}_{a  \sim {\pi}}\left[ \widehat{Q}^\pi_i(\hat{s}_i,\hat{a}_i)\gpi\right]\mathbb{E}_{s\sim d^\pi(s)}\mathbb{E}_{a  \sim {\pi}}\left[\widehat{Q}^\pi_i(\hat{s}_i,\hat{a}_i) \gpi\right]^\intercal}\nonumber\\
        &= \tr{\mathbb{E}_{s\sim d^\pi(s)}\mathbb{E}_{a  \sim {\pi}}\left[ \left(\left(Q^\pi(s,a)\right)^2 -\left(\widehat{Q}^\pi_i(\hat{s}_i,\hat{a}_i)\right)^2\right)\gpi \left(\gpi\right)^\intercal\right]}\nonumber\\
&\hspace{100pt}+\mathbb{E}_{s\sim d^\pi(s)} \bigg[(\sigma^2_Q  -\sigma^2_{\widehat{Q}})\text{Tr}\bigg(\mathbb{E}_{a  \sim {\pi}}\left[\left(\gpi \right)\left(\gpi \right)^\intercal\right]\bigg)\bigg] ~(\text{from}~\eqref{eq:QeqQhat}).
         \label{eq:case1int}
\end{align}
}
Define $\bar{s}_i,~\bar{a}_i$ such that $\hat{s}_i \cup \bar{s}_i = s$, $\hat{a}_i \cup \bar{a}_i = a$ and $\hat{s}_i \cap \bar{s}_i = \emptyset$, $\hat{a}_i \cap \bar{a}_i = \emptyset.$ Since $\widehat{Q}^\pi_i(\hat{s}_i,\hat{a}_i) = \widehat{Q}^\pi_i(s,\hat{a}_i) = \mathbb{E}_{\bar{a}_i \sim \bar{\pi}} Q^\pi(s,a)$ (due to Theorem~\ref{thm:Qsetdecomp}), \eqref{eq:case1int} can be rewritten as
{
\begin{align}
&=\tr{\mathbb{E}_{s\sim d^\pi(s),\hat{a}_i  \sim \hat{\pi}_i}\left[ \gpi \left(\gpi\right)^\intercal \mathbb{E}_{\bar{a}_i  \sim \bar{\pi}_i}\left[\left(Q^\pi(s,a)\right)^2 -\left(\widehat{Q}^\pi_i(\hat{s}_i,\hat{a}_i)\right)^2\right]\right]}\nonumber\\
&\hspace{170pt} +\mathbb{E}_{s\sim d^\pi(s)} \bigg[(\sigma^2_Q  -\sigma^2_{\widehat{Q}})\text{Tr}\bigg(\mathbb{E}_{a  \sim {\pi}}\left[\left(\gpi \right)\left(\gpi \right)^\intercal\right]\bigg)\bigg]\nonumber\\
       &=\tr{\mathbb{E}_{s\sim d^\pi(s),\hat{a}_i  \sim \hat{\pi}_i}\left[ \gpi \left(\gpi\right)^\intercal \mathbb{E}_{\bar{a}_i  \sim \bar{\pi}_i}\left[\left(Q^\pi(s,a)\right)^2 -2\left(\widehat{Q}^\pi_i(\hat{s}_i,\hat{a}_i)\right)^2+ \left(\widehat{Q}^\pi_i(\hat{s}_i,\hat{a}_i)\right)^2\right]\right]}\nonumber\\
&\hspace{170pt} +\mathbb{E}_{s\sim d^\pi(s)} \bigg[(\sigma^2_Q  -\sigma^2_{\widehat{Q}})\text{Tr}\bigg(\mathbb{E}_{a  \sim {\pi}}\left[\left(\gpi \right)\left(\gpi \right)^\intercal\right]\bigg)\bigg]\nonumber\\
&=\tr{\mathbb{E}_{s\sim d^\pi(s),\hat{a}_i  \sim \hat{\pi}_i}\left[ \gpi \left(\gpi\right)^\intercal \mathbb{E}_{\bar{a}_i  \sim \bar{\pi}_i}\left[\left(Q^\pi(s,a)\right)^2 -2\left(Q^\pi(s,a)\right)\left(\widehat{Q}^\pi_i(\hat{s}_i,\hat{a}_i)\right)+ \left(\widehat{Q}^\pi_i(\hat{s}_i,\hat{a}_i)\right)^2\right]\right]}\nonumber\\
&\hspace{170pt} +\mathbb{E}_{s\sim d^\pi(s)} \bigg[(\sigma^2_Q  -\sigma^2_{\widehat{Q}})\text{Tr}\bigg(\mathbb{E}_{a  \sim {\pi}}\left[\left(\gpi \right)\left(\gpi \right)^\intercal\right]\bigg)\bigg]\nonumber\\
&\left[\because \mathbb{E}_{\bar{a}_i \sim \bar{\pi}}[Q(s,a)] = \widehat{Q}(s, \hat{a}_i)\right]\nonumber\\
    &= \tr{\mathbb{E}_{s\sim d^\pi(s),\hat{a}_i  \sim \hat{\pi}_i}\left[ \gpi \left(\gpi\right)^\intercal \mathbb{E}_{\bar{a}_i  \sim \bar{\pi}_i}\left[\left(Q^\pi(s,a) -\widehat{Q}^\pi_i(\hat{s}_i,\hat{a}_i)\right)^2\right]\right]}\nonumber\\
&\hspace{170pt} +(\sigma^2_Q  -\sigma^2_{\widehat{Q}})\mathbb{E}_{s\sim d^\pi(s)} \bigg[\text{Tr}\bigg(\mathbb{E}_{a  \sim {\pi}}\left[\left(\gpi \right)\left(\gpi \right)^\intercal\right]\bigg)\bigg]\nonumber \\
    &= \tr{\mathbb{E}_{s\sim d^\pi(s),{a}  \sim {\pi}}\left[ \left(Q^\pi(s,a) -\widehat{Q}^\pi_i(\hat{s}_i,\hat{a}_i)\right)^2 \gpi \left(\gpi\right)^\intercal\right]}\nonumber\\
&\hspace{170pt} +(\sigma^2_Q  -\sigma^2_{\widehat{Q}})\mathbb{E}_{s\sim d^\pi(s)} \bigg[\text{Tr}\bigg(\mathbb{E}_{a  \sim {\pi}}\left[\left(\gpi \right)\left(\gpi \right)^\intercal\right]\bigg)\bigg]\nonumber \end{align}
       \begin{align}
    &= \mathbb{E}_{s\sim d^\pi(s),{a}  \sim {\pi}}\left[ \left(Q^\pi(s,a) -\widehat{Q}^\pi_i(\hat{s}_i,\hat{a}_i)\right)^2 ||\gpi||^2\right]\nonumber\\
&\hspace{170pt} +(\sigma^2_Q  -\sigma^2_{\widehat{Q}})\mathbb{E}_{s\sim d^\pi(s)} \bigg[\mathbb{E}_{a  \sim {\pi}}\left[||\gpi ||^2\right]\bigg]\nonumber\\
    &= \mathbb{E}_{s\sim d^\pi(s),{a}  \sim {\pi}}\left[ \left(A_{\mathcal{V}\setminus \I^i_{\widehat{Q}}}(s,\hat{a}_i,\bar{a}_i)\right)^2 ||\gpi||^2\right]\bigg]+(\sigma^2_Q  -\sigma^2_{\widehat{Q}})\mathbb{E}_{s\sim d^\pi(s)} \bigg[\mathbb{E}_{a  \sim {\pi}}\left[||\gpi ||^2\right]\bigg].
    \label{eq:diffeq}
    \end{align}
}
\noindent
Upper bounding $||\gpi||$ in the first term by $M_i \!= \!\underset{s,a}{\text{sup }} ||\gpi||$ 
yields
\begin{align}
&\mathbb{E}_{s\sim d^\pi(s)} \bigg[\mathbb{E}_{{a}  \sim {\pi}}\left[ \left(A_{\mathcal{V}\setminus \I^i_{\widehat{Q}}}(s,\hat{a}_i,\bar{a}_i)\right)^2 ||\gpi||^2\right]\bigg]\leq M^2_i~\mathbb{E}_{s\sim d^\pi(s)} \bigg[\mathbb{E}_{{a}  \sim {\pi}}\left[ \left(A_{\mathcal{V}\setminus \I^i_{\widehat{Q}}}(s,\hat{a}_i,\bar{a}_i)\right)^2\right]\bigg]\\ 
    &=  M^2_i~\mathbb{E}_{s\sim d^\pi(s)} \bigg[\mathbb{E}_{\hat{a}_i  \sim \hat{\pi}_i}\left[ \var_{\bar{a}_i \sim \bar{\pi}_i} \left(A_{\mathcal{V}\setminus \I^i_{\widehat{Q}}}(s,\hat{a}_i,\bar{a}_i)\right)\right]\bigg]\nonumber\\
    &\leq M^2_i~\mathbb{E}_{s\sim d^\pi(s)} \left[\mathbb{E}_{\hat{a}_i  \sim \hat{\pi}_i}\left[ \sum_{j \in \mathcal{V}\setminus \I^i_{\widehat{Q}}} \var_{\bar{a}_i \sim \bar{\pi}_i} \left(A_j(s,a_{-j},{a}_j)\right)\right]\right]\text{[Using Lemma 3,~\cite{kuba2021settling}]}\nonumber\\
    &= M^2_i~\mathbb{E}_{s\sim d^\pi(s)} \left[\sum_{j \in \mathcal{V}\setminus \I^i_{\widehat{Q}}}\mathbb{E}_{\hat{a}_i  \sim \hat{\pi}_i}\left[  \mathbb{E}_{\bar{a}_i \sim \bar{\pi}_i} \left(A_j(s,a_{-j},{a}_j)^2\right)\right]\right].
    \label{eq:case2final}
\end{align}
Assumption~\ref{assume:boundedr} implies that the advantage function is bounded. Therefore, by the completeness axiom, $\epsilon_i = \underset{s,a}{\text{sup }} |A_i(s,a_{-i},a_i)|$ exists.
Thus,
\begin{align}
    M^2_i\mathbb{E}_{s\sim d^\pi(s)} \bigg[\mathbb{E}_{{a}  \sim {\pi}}\left[ \left(A_{\mathcal{V}\setminus \I^i_{\widehat{Q}}}(s,\hat{a}_i,\bar{a}_i)\right)^2\right]\bigg] &\leq M^2_i\mathbb{E}_{s\sim d^\pi(s)} \left[\sum_{j \in \mathcal{V}\setminus \I^i_{\widehat{Q}}}\mathbb{E}_{a  \sim \pi}\left[A_j(s,a_{-j},{a}_j)^2\right]\right]\nonumber\\
    &\leq M^2_i ~\mathbb{E}_{s\sim d^\pi(s)} \left[\sum_{j \in \mathcal{V}\setminus \I^i_{\widehat{Q}}}\epsilon^2_j\right].
    \label{eq:upperbound}
\end{align}
Therefore, from~\eqref{eq:case2final} and~\eqref{eq:upperbound} we obtain 
\begin{align}
    \mathbb{E}_{s\sim d^\pi(s)} \bigg[\mathbb{E}_{{a}  \sim {\pi}}\left[ \left(A_{\mathcal{V}\setminus \I^i_{\widehat{Q}}}(s,\hat{a}_i,\bar{a}_i)\right)^2 ||\gpi||^2\right] &\leq M^2_i \mathbb{E}_{s\sim d^\pi(s)} \left[\sum_{j \in \mathcal{V}\setminus \I^i_{\widehat{Q}}}\epsilon^2_j\right].
    \label{eq:upperbound1}
\end{align}

To characterize the lower bound we note that from the definition of multi-agent advantage function and~\eqref{eq:QbarQhat}, $A_{\mathcal{V}\setminus \I^i_{\widehat{Q}}}(s,\hat{a}_i,\bar{a}_i) = \bar{Q}_i(s,\bar{a}_i)$ and $\bar{Q}_i(\cdot) $ is independent of $a_i$ from Theorem~\ref{thm:graddecomp}. Hence, we rewrite the first term in~\eqref{eq:diffeq}  as
\begin{align}
    \mathbb{E}_{s\sim d^\pi(s)} \bigg[\mathbb{E}_{{a}  \sim {\pi}}\left[ \left(A_{\mathcal{V}\setminus \I^i_{\widehat{Q}}}(s,\hat{a}_i,\bar{a}_i)\right)^2 ||\gpi||^2\right]\bigg] &= \mathbb{E}_{s\sim d^\pi(s)} \left[\mathbb{E}_{a_{-i}  \sim \pi_{-i}} \left[ \left(\bar{Q}^\pi_i(s,\bar{a}_i)\right)^2 \mathbb{E}_{a_{i}  \sim \pi_{i}} ||\gpi||^2\right]\right].
\end{align}
Let $N_i = \underset{s,a_i}{\text{inf }}\mathbb{E}_{a_{i}  \sim \pi_{i}} ||\gpi||^2$. Then
\begin{align}
    &\mathbb{E}_{s\sim d^\pi(s)} \left[\mathbb{E}_{a_{-i}  \sim \pi_{-i}} \left[ \left(\bar{Q}^\pi_i(s,\bar{a}_i)\right)^2 \mathbb{E}_{a_{i}  \sim \pi_{i}} ||\gpi||^2\right]\right] \geq N^2_i ~\mathbb{E}_{s\sim d^\pi(s)} \left[\mathbb{E}_{a_{-i}  \sim \pi_{-i}} \left[ \left(\bar{Q}^\pi_i(s,\bar{a}_i)\right)^2\right]\right]\nonumber\\
    &=  N^2_i ~\mathbb{E}_{s\sim d^\pi(s)} \left[\mathbb{E}_{a_{-i}  \sim \pi_{-i}} \left[ \left(\sum_{j\in \mathcal{V}\setminus \I^i_{\widehat{Q}}}{Q}^\pi_j(s,{a}_{\I^j_Q})\right)^2\right]\right].~\text{(from~\eqref{eq:QbarQhat})}
    \label{eq:lowerbound}
\end{align}
Therefore, from~\eqref{eq:upperbound1} and~\eqref{eq:lowerbound} we conclude that
\begin{align}
   & ~\mathbb{E}_{s\sim d^\pi(s)} \left[N^2_i\mathbb{E}_{a_{-i}  \sim \pi_{-i}} \left[ \left(\sum_{j\in \mathcal{V}\setminus \I^i_{\widehat{Q}}}{Q}^\pi_j(s,{a}_{\I^j_Q})\right)^2\right]+(\sigma^2_Q  -\sigma^2_{\widehat{Q}})\mathbb{E}_{a  \sim {\pi}}\left[||\gpi ||^2\right]\right]\leq \nonumber\\&\hspace{100pt}\var_{s\sim d^\pi(s),a  \sim {\pi}, \dq \sim p(\mu_Q, \sigma^2_Q)} [\gc] - \var_{s\sim d^\pi(s),a  \sim {\pi},\dqh \sim p(\mu_{\widehat{Q}}, \sigma^2_{\widehat{Q}})} [\gq]  \nonumber\\&\hspace{100pt}\leq\mathbb{E}_{s\sim d^\pi(s)} \left[ M^2_i \sum_{j \in \mathcal{V}\setminus \I^i_{\widehat{Q}}}\left(\epsilon_j\right)^2 +(\sigma^2_Q  -\sigma^2_{\widehat{Q}})\mathbb{E}_{a  \sim {\pi}}\left[||\gpi ||^2\right]\right]. 
   \label{eq:bound}
\end{align}

The upper bound implies that if $\sigma^2_Q \geq \sigma^2_{\widehat{Q}}$,  the variance of $\gq$ decreases linearly with the increase of the number of agents in $\mathcal{V}\setminus \I^i_{\widehat{Q}}$. A similar upper bound was obtained in~\cite{kuba2021settling} for a decentralized policy gradient estimator. However, to the best of our knowledge this is the first attempt to analyze the lower bound of the total variance difference. If $N_i \neq 0$ and $\sigma_Q \geq \sigma_{\widehat{Q}},$ then the lower bound in~\eqref{eq:bound} is non-zero which means that $\var[\gq]$ is strictly less than $\var[\gc]$, and the lower bound increases as the variance of $\sum Q_j(\cdot)~\forall~j\in \mathcal{V}\setminus \I^i_{\widehat{Q}}$ increases. For example, when all the agents take exploratory actions, then the variance of $\gc$ is greater than the variance of $\gq$ owing to the decomposition in Theorem~\ref{thm:Qsetdecomp},~\ref{thm:graddecomp}. This signifies the effect of the proposed approach in achieving variance reduction compared to the centralized policy gradient estimator.
\end{proof}

\section{Warehouse resource allocation problem}
\label{sec:warehouse}

Consider a group of $N$-warehouses that can consume, observe and transfer resources among each other based on a set of pre-determined graphs. Given a state graph $\mathcal{G}_S = \{\mathcal{V}, \mathcal{E}_S\}$, an observation graph $\mathcal{G}_O = \{\mathcal{V}, \mathcal{E}_O\}$, and a reward graph $\mathcal{G}_R = \{\mathcal{V}, \mathcal{E}_R\}$, the goal of the resource allocation problem is to find an optimal allocation strategy that guarantees adequate resources for each warehouse to meet its local demand. Each warehouse is denoted by a vertex in the graph. At time $t$, each warehouse $i \in \mathcal{V}$ stores resources of amount $m_i(t) \in \mathbb{R}$, receives a local demand $d_i(t) \in \mathbb{R}$, a local supply of $y_i(t) \in \mathbb{R}$, , sends a partial amount of resources to its out-neighbors $\mathcal{N}^{\text{out},i}_S$ in $\mathcal{G}_S$, and receives a partial amount of resources from its in-neighbors $ \mathcal{N}^{\text{in},i}_S$ in $\mathcal{G}_S$. Let $z_i(t) = y_i(t) - d_i(t)$, then the agent $i$ follows the dynamics
\begin{align}
    m_i(t+1) &= m_i(t) - \sum_{j \in \mathcal{N}^{out,i}_S} \alpha_i(t) b_{ij}(o_i(t))m_i(t) + \sum_{j \in \mathcal{N}^{in,i}_S} \alpha_j(t) b_{ji}(o_j(t))m_j(t) + z_i(t),\nonumber\\
    z_i(t) &= A_i \sin(\omega_i t + \phi) + \omega_i,
\end{align}
where $\alpha_i(t) = \begin{cases}
    1& \text{ , if $m_i(t) \geq 0$}\\ 0& \text{ , otherwise}
\end{cases} $ is the indicator function that allows distribution of resources only if $m_i(t) \geq 0$ $\forall ~ i$, $b_{ij}(o_i(t)) \in [0,1]$ is the fraction of resources that agent $i$ sends to agent $j$ at time $t$, $0 < A_i < m_i(0)$ is a constant, $\omega_i$ is a bounded random quantity $\forall$ $i$, and $\phi$ is a positive scalar. We define the state $s_i(t) = (m_i(t), z_i(t))^\intercal$, action $a_i(t) = (\cdots,b_{ij}(o_i(t)),\cdots)^\intercal_{j \in \mathcal{N}^{i,\text{out}}_S}$, observation $o_i(t) = (\{m_j(t)\}_{j \in \mathcal{I}^i_O}, z_i(t))^\intercal$,and the reward $r_i(t) = \sum_{j \in \mathcal{N}^i_R} \tau_j(t)$, $\forall$ $i \in \mathcal{V}$,  where
\begin{align}
    \tau_j(t) &= \begin{cases}
        0 & \text{ if $m_j(t) \geq 0$}\\
        -m^2_j(t) & \text{ otherwise.}
    \end{cases}
\end{align}

\begin{figure}
    \begin{tikzpicture}
        \begin{scope}
    \node[latent,draw=blue!50] (A) {$1$};
    \node[latent,draw=blue!50, below=of A] (B) {$4$};
    \node[latent,draw=blue!50 ,fill=blue!10,right=of A] (C) {$2$};
    \node[latent,draw=blue!50,fill=blue!10, below=of C] (D) {$3$};
    \node[latent,draw=blue!50, right=of C] (E) {$6$};
    \node[latent,draw=blue!50, right=of E] (F) {$9$};
    \node[latent,draw=blue!50 ,fill=blue!10, below=of E] (G) {$5$};
    \node[latent,draw=blue!50, below=of F] (H) {$7$};
    \node[latent,draw=blue!50, right=of H] (I) {$8$};
    \node[above=of C,yshift= -1cm] {$\mathcal{G}_S$};
    \edge {C} {A};
    \edge [ultra thick]{C} {D};
    \edge {D} {B};
    \edge [ultra thick]{G} {D};
    \edge {G} {E};
    \edge {G} {H};
    \edge {F} {H};
    \edge {H} {I};
    \edge {I} {F};
\end{scope}
\begin{scope}[xshift=8cm]
    \node[latent,draw=blue!50, fill=blue!10] (A) {$1$};
    \node[latent,draw=blue!50,fill=blue!10, below=of A] (B) {$4$};
    \node[latent,draw=blue!50 ,fill=blue!10,right=of A] (C) {$2$};
    \node[latent,draw=blue!50,fill=blue!10, below=of C] (D) {$3$};
    \node[latent,draw=blue!50, fill=blue!10, right=of C] (E) {$6$};
    \node[latent,draw=blue!50, right=of E] (F) {$9$};
    \node[latent,draw=blue!50 ,fill=blue!10, below=of E] (G) {$5$};
    \node[latent,draw=blue!50, below=of F] (H) {$7$};
    \node[latent,draw=blue!50, right=of H] (I) {$8$};
    \node[above=of C,yshift= -1cm] {$\mathcal{G}_O$};
    \edge[draw= red, ultra thick] {A} {C};
    \edge [draw= red, ultra thick] {B} {D};
    \edge [draw= red, ultra thick]{E} {G};
\end{scope}
\begin{scope}[xshift= 5cm, yshift=-4cm]
    \node[latent,draw=blue!50, fill=blue!10] (A) {$1$};
    \node[latent,draw=blue!50,fill=blue!10, below=of A] (B) {$4$};
    \node[latent,draw=blue!50 ,fill=blue!10,right=of A] (C) {$2$};
    \node[latent,draw=blue!50,fill=blue!10, below=of C] (D) {$3$};
    \node[latent,draw=blue!50, fill=blue!10, right=of C] (E) {$6$};
    \node[latent,draw=blue!50, right=of E] (F) {$9$};
    \node[latent,draw=blue!50 ,fill=blue!10, below=of E] (G) {$5$};
    \node[latent,draw=blue!50, below=of F] (H) {$7$};
    \node[latent,draw=blue!50, right=of H] (I) {$8$};
    \node[above=of C,yshift= -1cm] {$\mathcal{G}_R$};
    \edge[draw= violet, ultra thick] {A} {C};
    \edge [draw= violet, ultra thick] {B} {D};
    \edge [draw= violet, ultra thick]{E} {G};
    \edge [draw=violet]{H} {F};
    \edge [draw=violet]{I} {H};
    \edge [draw=violet]{F} {I};
\end{scope}
    \end{tikzpicture}
\caption{$\mathcal{G}_S$, $\mathcal{G}_O$, and $\mathcal{G}_R$ for Example 1.}
\label{fig:warehouse9}
\end{figure}
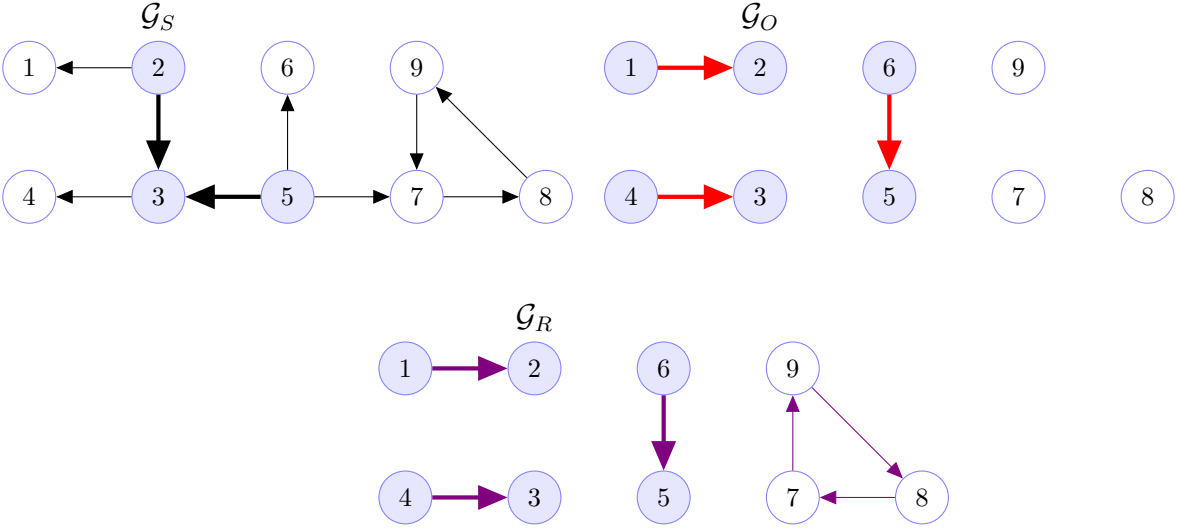


\section{Multi-zone temperature control problem}\label{sec:temp_control}
Consider a building  divided into $N$ zones. 
Given a state graph $\mathcal{G}_S = \{\mathcal{V}, \mathcal{E}_S\}$, an observation graph $\mathcal{G}_O = \{\mathcal{V}, \mathcal{E}_O\}$, and a reward graph $\mathcal{G}_R = \{\mathcal{V}, \mathcal{E}_R\}$, the goal of the multi-zone temperature control problem is for the temperature controller in each zone to find an optimal policy that controls a quantity related to the airflow rate and achieves the desired temperature within each zone. Each zone is denoted by a vertex in the graph. Assuming finite thermal resistance of the walls separating adjacent zones, the temperature of each zone is affected by the heat flow from the adjacent zones through walls leading to a dynamic coupling between the zones i.e.,
$\forall$ $i,j \in \mathcal{V}$, 
if zone $i$ and zone $j$ share a common wall, then  $j \in \I^i_S,~i\in \I^j_S$. The temperature in zone $i$, $i \in \mathcal{V}$, evolves according to
\begin{align}
    x^i_{t+1} &= \underbrace{(1 - \frac{\Delta}{\nu_i \zeta_i}) x^i_t + \frac{\Delta}{\nu^i} u^i_t }_{\text{self interaction}} + \underbrace{\sum_{j \in \mathcal{N}^i} (x^j_t - x^i_t) \frac{\Delta}{\nu^i \zeta^{ij}}}_{\text{neighbor interaction}}  + \underbrace{\frac{\Delta }{\nu^i \zeta^i}\epsilon_0 + \frac{\Delta }{\nu^i}\pi^i}_{\text{environment interaction}} + \underbrace{\frac{\sqrt{\Delta}}{\nu^i} w^i_t}_{\text{process noise}},
    \label{eq:tempzone}
\end{align}
where $x^i \in\mathbb{R}$ is the temperature in the zone $i$, $u^i \in \mathbb{R}$ is the corresponding control input, $\epsilon_0$ is the outdoor temperature, $\Delta$ is the time resolution, $\nu^i$ is the thermal capacitance of zone $i$, $\zeta^i$ denotes the thermal resistance of windows and walls between zone $i$ and the environment, $\zeta^{ij}$ denotes the thermal resistance between zones $i$, $j$, $\pi^i$ represents the constant heat addition from external sources into zone $i$, and $w^i\sim\mathcal{N}(0,1)$ represents the process noise in the evolution of temperature in zone $i$.
Moreover, the controller in each zone $i$ observes $o_i = (x^j)_{j\in \I^i_O}.$ After each step, each zone receives a reward $r_i(t) = -(x^i_t- x_i^*)^2 - \beta_i (u_t^i)^2$. Thus, $\I_R^i=\{i\}.$
 The simulation parameters are summarized in Table~\ref{table:tcsim}.
\begin{table}[h]
\caption{ {Simulation parameters for the temperature control example.}}
\label{table:tcsim}
\centering
   {
\begin{tabular}{ll}
    Simulation parameters & Value\\
      \hline
      Number of zones, $N$ & $40$\\
      Step size, $\Delta$ & $60~\mbox{sec}$\\
      Distribution of initial temperature for each zone, $x^i_0 $& $\mathcal{N}(30,2.5)$\\
      Initial state covariance, $\Sigma_{x_0}$ & $0.0001 \mathbb{I}_{N}$\\
      Target temperature, $z_T$ & $22\degree\mbox{C}$\\
      Outside temperature,$\epsilon_0$ & $30\degree\mbox{C}$\\
      Process noise covariance, $\Sigma_{\eta_t}$ &  $\dfrac{\Delta\times 6.25}{(\nu^i)^2}$\\
      State cost matrix, $Q_t$ &  $\mathbb{I}_{N}$\\
      Terminal cost matrix, $Q_T$ &  $\mathbb{I}_{N}$\\
      Control cost matrix, $R_t$ &  $\beta_i\mathbb{I}_{N}$\\
      Control cost parameter, $\beta_i$ $\forall~i$& 0.01\\
      Constant heat from environment to zone $i$, $\pi^i$ & $1 ~\mbox{kW}$\\
      Thermal capacitance of zone $i$,  $\nu^i$ & $200~\mbox{kJ/}\degree\mbox{C}$\\
      Thermal resistance between zone $i$ and environment, $\zeta^{i}$  & $1\degree\mbox{C/kW}$\\
      Thermal resistance between zones $i$ and $j$, {$\zeta^{ij}$}& {$\begin{cases}
          $1\degree\mbox{C/kW}$,& \text{if $(i,j) \in\mathcal{E}_S$}\\ 
          $0$,& \text{otherwise}
      \end{cases}$}\\
    \hline
  \end{tabular}}
\end{table}

\section{Additional results}\label{sec:additional}

\noindent
\textbf{Effect of $\kappa$ on the rate of convergence of MAStAC algorithm}\label{sec:kappacomp}

To investigate the effect of $\kappa$ on the rate of convergence of the MAStAC algorithm, we run the 40-warehouse example with $\kappa = 2,4,6,8.$ Fig.~\ref{fig:kappacomp} shows the comparison of the mean and standard deviation of total reward for various choices of $\kappa$ in the `MAStAC approximated' algorithm. We observe that as the value of $\kappa$ increases, the rate of convergence decreases which can be attributed to the fact that a higher value of $\kappa$ yields a denser $\mathcal{G}^\kappa_{\text{VD}}.$ 
\begin{figure}[H]
    \centering
    \includegraphics[width=0.5\textwidth]{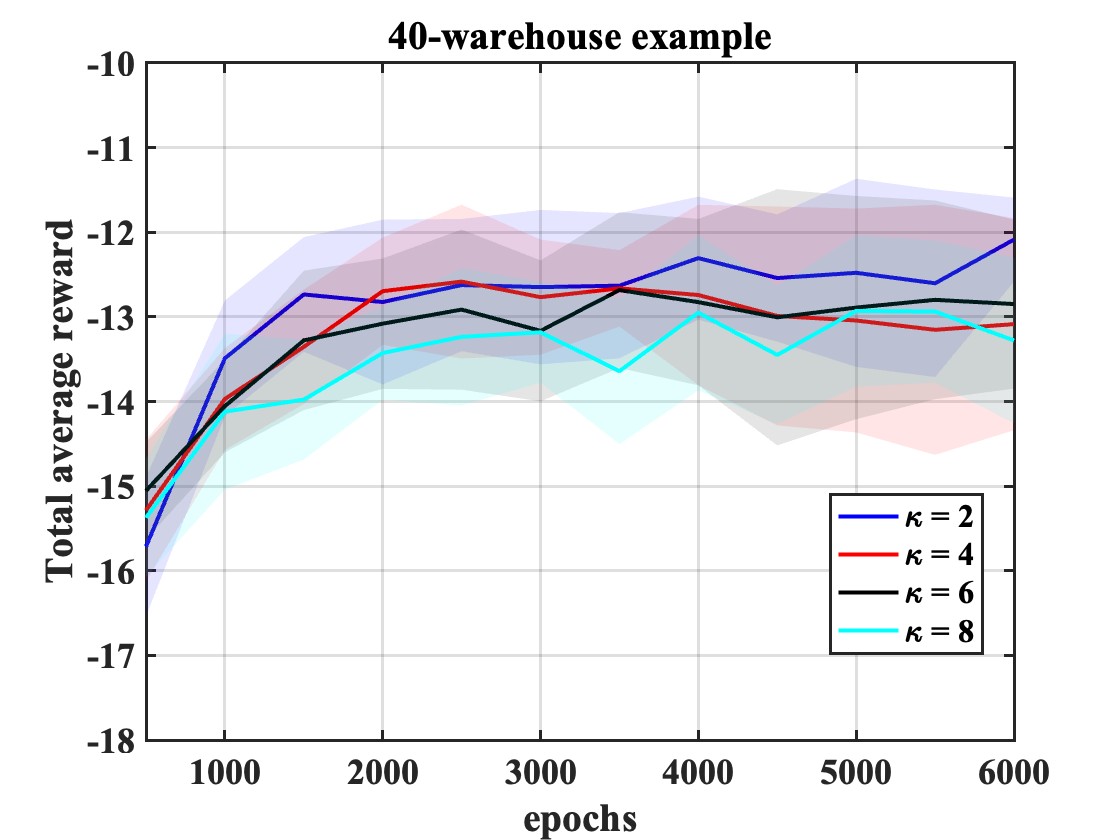}
    \caption{Comparison of the total average reward for 15 MC simulations for $\kappa= 2,4,6,8.$}
    \label{fig:kappacomp}
\end{figure}


\section{Simulation parameters}\label{sec:sim_para}
\begin{table*}[htpb]
\caption{Parameters used in the MAStAC simultaneous algorithm for the three examples.}
\label{table:simparam}
\begin{center}
\begin{tabular}{cccc}
Parameter & 9-warehouse& 40-warehouse& \makecell{40-zone \\temperature control}\\
    \hline
Number hidden layers (actor)& 3&3&\longdash[2]\\
Number hidden layers (critic)&3&3&3\\
Actor output activation& softmax& softmax& tanh\\
Learning rate actor& 1e-4&5e-4&1e-4\\
Learning rate actor& 1e-3&5e-3&1e-3\\
Batch size& 256 &256& 256\\
Episode length& 8& 8& 40\\
Epochs&3500&6000&5000\\
Discount factor $\gamma$& 0.95&0.95&0.9\\
\hline
\end{tabular}
\end{center}
\end{table*}

\end{document}